\documentclass[11pt,a4paper]{article}

\usepackage{bhms-journal}

\title{\bfseries 
  Minimal Distance of Propositional Models%
  \thanks{This paper is the final version of a trilogy of conference papers
  \emph{Give Me Another One!}~\cite{BehrischHMSisaac-15} (ISAAC 2015),
  \emph{As Close As It Gets}~\cite{BehrischHMSwalcom-16} (WALCOM 2016), and
  \emph{The Next Whisky Bar}~\cite{BehrischHMScsr-16} (CSR 2016).}}
\author{Mike Behrisch%
  \footnotemark[4]\\
  Technische Universit\"at Wien, Vienna, Austria\\
  behrisch@logic.at
  \and
  Miki Hermann%
  \thanks{Supported by ANR-11-ISO2-003-01 Blanc International grant
    ALCOCLAN. Part of the work was done during his stay at the
    Wolfgang Pauli Institute (ICP, UMI CNRS 2842) in Vienna, Austria.}\\
  LIX (UMR CNRS 7161),
  \'Ecole Polytechnique, Palaiseau, France\\
  hermann@lix.polytechnique.fr
  \and
  Stefan Mengel%
  \thanks{Research of this author was done during his post-doctoral
    stay in LIX at \'Ecole Polytechnique.
    Supported by a QUALCOMM grant.}\\
  CRIL (UMR CNRS 8188),
  Universit\'e d'Artois, Lens, France.\\
  mengel@cril.fr
  \and
  Gernot Salzer%
  \thanks{Supported by Austrian Science Fund (FWF) grant I836-N23.}\\
  Technische Universit\"at Wien, Vienna, Austria\\
  salzer@logic.at
}

\begin{document}

\maketitle

\begin{abstract}
  We investigate the complexity of three optimization problems in
  Boolean propositional logic related to information theory: Given a
  conjunctive formula over a set of relations, find a satisfying
  assignment with minimal Hamming distance to a given assignment that
  satisfies the formula ($\NextSol$, $\XSOL$) or that does not need to
  satisfy it ($\NearestSol$, $\NSOL$). The third problem asks for two
  satisfying assignments with a minimal Hamming distance among all
  such assignments ($\MinSolDistance$, $\MSD$).

  For all three problems we give complete classifications with respect
  to the relations admitted in the formula. We give polynomial time
  algorithms for several classes of constraint languages. For all
  other cases we prove hardness or completeness regarding $\APX$,
  $\pAPX$, or equivalence to well-known hard optimization
  problems.
\end{abstract}

\section{Introduction}
\label{sec:intro}

We investigate the solution spaces of Boolean constraint satisfaction
problems built from atomic constraints by means of conjunction and
variable identification. We study three minimization problems in
connection with Hamming distance: Given an instance of a constraint
satisfaction problem in the form of a generalized conjunctive formula
over a set of atomic constraints, the first problem asks to find a
satisfying assignment with minimal Hamming distance to a given
assignment ($\NearestSol$, $\NSOL$). Note that for this problem we
assume neither that the given assignment satisfies the formula nor that
the solution is different from the assignment. The second problem
is similar to the first one, but this time the given assignment
has to satisfy the formula and we look for another solution with
minimal Hamming distance ($\NextSol$, $\XSOL$).  The
third problem is to find two satisfying assignments with minimal
Hamming distance among all satisfying assignments ($\MinSolDistance$,
$\MSD$). Note that the dual problem $\MaxSolDistance$ has been studied
in~\cite{CrescenziR-02}.

The $\NSOL$ problem appears in several guises throughout
literature. E.g., a common problem in Artificial Intelligence is to
find solutions of constraints close to an initial configuration; our
problem is an abstraction of this setting for the Boolean
domain. Bailleux and Marquis~\cite{BailleuxM-06} describe such
applications in detail and introduce the decision problem
$\DistanceSAT$: Given a propositional formula~$\varphi$, a
partial interpretation~$I$, and a bound~$k$, is there a satisfying
assignment differing from~$I$ in no more than $k$ variables? It is
straightforward to show that $\DistanceSAT$ corresponds to
the decision variant of our problem with existential quantification
(called $\dNSOLpp$ later on). While \cite{BailleuxM-06} investigates
the complexity of $\DistanceSAT$ for a few relevant classes
of formulas and empirically evaluates two algorithms, we analyze the
decision and the optimization problem for arbitrary semantic
restrictions on the formulas.

Hamming distance also plays an important role in belief revision. The
result of revising/updating a formula~$\varphi$ by another
formula~$\psi$ is characterized by the set of models of~$\psi$ that
are closest to the models of~$\varphi$. Dalal~\cite{Dalal-88}
selects the models of~$\psi$ having a minimal Hamming distance to
models of~$\varphi$ to be the models that result from the change.

As is common, we analyze the complexity of our optimization
problems modulo a parameter that specifies the atomic constraints
allowed to occur in the constraint satisfaction problem.  We give a
complete classification of the approximation complexity with respect
to this parameterization. It turns out that our problems can either be
solved in polynomial time, or they are complete for a well-known
optimization class, or else they are equivalent to well-known hard
optimization problems.

Our study can be understood as a continuation of the minimization
problems investigated by Khanna et al.\ in~\cite{KhannaSTW-01},
especially that of $\OptMinOnes$. The $\OptMinOnes$ optimization
problem asks for a solution of a constraint satisfaction problem with
the minimal Hamming weight, i.e., minimal Hamming distance to the
$0$-vector. Our work generalizes these results by allowing the given
vector to be also different from zero.

Our work can also be seen as a generalization of questions
in coding theory. In fact, our problem $\MSD$ restricted to affine
relations is the well-known problem $\MinDist$ of computing the
minimum distance of a linear code. This quantity is of central
importance in coding theory, because it determines the number of
errors that the code can detect and correct. Moreover, our problem
$\NSOL$ restricted to affine relations is the problem $\NCW$ of
finding the nearest codeword to a given word, which is the basic
operation when decoding messages received through a noisy
channel. Thus our work can be seen as a generalization of these
well-known problems from affine to general relations.

In the case of $\NearestSol$ we are able to apply methods from clone
theory, even though the problem turns out to be more intricate than
pure satisfiability. The other two problems, however, cannot be shown
to be compatible with existential quantification easily, which makes
classical clone theory inapplicable. Therefore we have to resort to
weak co-clones that require only closure under conjunction and
equality. In this connection, we apply the theory developed
in~\cite{SchnoorS-08,SchnoorDiss} as well as the minimal weak bases of
Boolean co-clones from~\cite{Lagerkvist-14}.

This paper is structured as follows. Section~\ref{sec:prelim} recalls
basic definitions and notions.
Section~\ref{sec:results} introduces the trilogy of optimization
problems studied in this paper, namely Nearest Solution (denoted by
$\NSOL$), Nearest Other Solution (denoted by
$\XSOL$), and Minimum Solution Distance (denoted by
$\MSD$), as well as their decision versions.  It
also states our three main results, i.e., a complete classification of
complexity for these optimization problems, depicted in
Figures~\ref{fig:nsol-coclones}, \ref{fig:xsol-coclones},
and~\ref{fig:msd-coclones}.
Section~\ref{sec:proofsPP} investigates the (non-)applicability of
clone theory to our problems. It also provides a duality result for
the constraint languages used as parameters.
Section~\ref{sec:proofsNSOL} contains the proofs of complexity
classification results for $\NearestSol$,
Section~\ref{sec:proofsXSOL} for $\NextSol$, and
Section~\ref{sec:proofsMSD} for $\MinSolDistance$.
Finally, the concluding remarks in Section~\ref{sec:conclusion}
compare our theorems to previously existing similar results and put our
results into perspective.

\section{Preliminaries}
\label{sec:prelim}

\subsection{Boolean Relations and Relational Clones}

An $n$-ary \emph{Boolean relation}~$R$ is a subset of $\set{0,1}^n$;
its elements $(b_1, \ldots, b_n)$ are also written as
$b_1 \cdots b_n$.  Let~$V$ be a set of variables. An \emph{atomic
  constraint}, or an \emph{atom}, is an expression $R(\vec x)$,
where~$R$ is an $n$-ary relation and $\vec x$ is an $n$-tuple of
variables from~$V$.  Let~$\Gamma$ be a non-empty finite set of Boolean
relations, also called a \emph{constraint language}. A (conjunctive)
\emph{$\Gamma$-formula} is a finite conjunction of atoms
$R_1(\vec{x_1}) \land \cdots \land R_k(\vec{x_k})$, where the~$R_i$
are relations from~$\Gamma$ and the $\vec{x_i}$ are variable tuples of
suitable arity. For technical reasons in connection with reductions we
also allow empty conjunctions ($k=0$) here. Such formulas elegantly take
care of certain marginal cases at the cost of adding only one additional
trivial problem instance.

An \emph{assignment} is a mapping $m\colon V \rightarrow \set{0,1}$
assigning a Boolean value $m(x)$ to each variable $x \in V$.
In a given context we can assume $V$ to be finite, by restricting it
e.g.\ to the variables occurring in a formula. If we impose
an arbitrary but fixed order on the variables, say $x_1, \ldots, x_n$,
then the assignments can be identified with elements from $\set{0,1}^n$.
The $i$-th component of a tuple~$m\in\set{0,1}^n$
is denoted by~$m[i]$ and corresponds to the value of the $i$-th
variable, i.e., $m[i]=m(x_i)$.  The \emph{Hamming weight}
$\hw(m) = \card{\Set{i}{m[i]=1}}$ of~$m$ is the number of~$1$s in the
tuple~$m$. The \emph{Hamming distance}
$\hd(m,m') = \card{\Set{i}{m[i] \neq m'[i]}}$ of~$m$ and~$m'$ is the
number of coordinates on which the tuples disagree. The
complement~$\cmpl{m}$ of a tuple~$m$ is its pointwise complement,
$\cmpl m[i] = 1- m[i]$.

An assignment~$m$ satisfies a constraint $R(x_1,\ldots ,x_n)$ if
$(m(x_1), \ldots , m(x_n))\in R$ holds.  It satisfies the
formula~$\varphi$ if it satisfies all its atoms; $m$ is said to be
a \emph{model} or \emph{solution} of~$\varphi$ in this case. We use $[\varphi]$ to
denote the set of models of~$\varphi$.
For a term~$t$, $[t]$ is the set of assignments for which $t$ evaluates to~$1$.
Note that $[\varphi]$ and $[t]$ represent Boolean relations.
If the variables of~$\varphi$ are not
explicitly enumerated in parentheses as parameters, they are
implicitly considered to be ordered lexicographically.
In sets of relations represented this
way we usually omit the brackets.  A \emph{literal} is a variable~$v$,
or its negation~$\neg v$. Assignments are extended to literals by
defining $m(\neg v)=1-m(v)$.

Table~\ref{tab:funrel} defines Boolean functions and relations needed
later on, in particular exclusive or $[x \oplus y]$,
not-all-equal $\nae^3$, $k$-ary disjunction $\bor^k$,
and $k$-ary negated conjunction $\nand^k$.

\begin{table}[t]
  \caption{List of some Boolean functions and relations}%
  \label{tab:funrel}
  \centering
  \begin{displaymath}
    \begin{array}[t]{@{}rcl@{\qquad}rcl@{}}
      x \oplus y &=& x + y \pmod{2}
      &
      \bor^k &=& \set{0,1}^k\smallsetminus\set{0 \dotsm 0} \text{ for $k \geq 1$}\\
      {\eq}  &=& \set{00,11}
      &
      \nand^k &=& \set{0,1}^k\smallsetminus\set{1 \dotsm 1} \text{ for $k \geq 1$}\\
      \dup^3 &=& \set{0,1}^3 \smallsetminus \set{010, 101}
      &\even^k &=& \set{(a_1, \ldots, a_k) \in \set{0,1}^k \mid
                  \sum_{i=1}^k a_i \text{ is even}}\\
      \nae^3 &=& \set{0,1}^3 \smallsetminus \set{000, 111}
      &
      \odd^k &=& \set{(a_1, \ldots, a_k) \in \set{0,1}^k \mid
                  \sum_{i=1}^k a_i \text{ is odd}}\\
      S_0 &=& [(x_1\land x_4) \eq (x_2\land x_3)]
      &
      S_1 &=& [S_0(\neg x_1,\neg x_2,\neg x_3,\neg x_1)]\\
      S_2 &=& [(\neg x_1\lor\neg x_2) \to \neg x_3]\\
      \multicolumn{6}{c}{\even^k_{k\neq} = \set{(a_1, \ldots, a_{2k}) \in \set{0,1}^{2k} \mid
                          \even^k(a_1, \ldots, a_k) \land \bigwedge_{i=1}^k \left(a_{k+i}\eq \neg a_i  \right)}}
    \end{array}
  \end{displaymath}
\end{table}

\begin{table}[b]
  \caption{Some relevant Boolean co-clones with bases}
  \label{tab:clones}
  \centering
  \begin{displaymath}
    \begin{array}[t]{@{}ll@{}}
     \iS_0^k    & \set{\bor^k} \\
     \iS_1^k    & \set{\nand^k} \\
     \iS_{00}^k & \set{\bor^k, x \to y, \neg x, x} \\
     \iS_{10}^k & \set{\nand^k, \neg x, x, x \to y } \\
     \iD_1      & \set{x \oplus y, x} \\
     \iD_2      & \set{x \oplus y, x \to y}
   \end{array}
   \qquad
   \begin{array}[t]{@{}ll@{}}
      \iL        & \set{\even^4} \\
      \iL_2      & \set{\even^4, \neg x, x} \\
      \iV        & \set{x \lor y \lor \neg z} \\
      \iV_2      & \set{x \lor y \lor \neg z, \neg x, x} \\
      \iE        & \set{\neg x \lor \neg y \lor z} \\
      \iE_2      & \set{\neg x \lor \neg y \lor z, \neg x, x}
    \end{array}
    \qquad
    \begin{array}[t]{@{}ll@{}}
       \iN        & \set{\dup^3} \\
       \iN_2      & \set{\nae^3} \\
       \iI        & \set{\even^4, x \to y} \\
       \iI_0      & \set{\even^4, x \to y, \neg x}\\
       \iI_1      & \set{\even^4, x \to y, x}\\
       \iM_2      & \set{x \to y, \neg x, x}
     \end{array}
  \end{displaymath}
\end{table}

Throughout the text we refer to different types of Boolean constraint
relations following Schaefer's terminology~\cite{Schaefer-78} (see
also the monograph~\cite{CreignouKS-01} and the
survey~\cite{BoehlerCRV-04}).  A Boolean relation~$R$ is
\begin{inparaenum}[(1)]
\item \emph{$1$-valid} if $1 \cdots 1 \in R$ and
  \emph{$0$-valid} if $0 \cdots 0 \in R$,
\item \emph{Horn} (\emph{dual Horn}) if~$R$ can be represented by a
  formula in conjunctive normal form (CNF) with at most one
  unnegated (negated) variable per clause,
\item \emph{monotone} if it is both Horn and dual Horn,
\item \emph{bijunctive} if it can be represented by a CNF formula with
  at most two literals per clause,
\item \emph{affine} if it can be represented by an affine system of
  equations $A x = b$ over~$\ZZ_2$,
\item \emph{complementive} if for each $m \in R$ also $\cmpl m \in R$,
\item \emph{implicative hitting set-bounded$+$} with bound~$k$ (denoted by $\IHSBp k$)
  if~$R$ can be represented by a CNF formula with clauses of the form
  $(x_1 \lor \cdots \lor x_k)$, $(\neg x \lor y)$, $x$, and~$\neg x$,
\item \emph{implicative hitting set-bounded$-$} with bound~$k$ (denoted by $\IHSBm k$)
  if~$R$ can be represented by a CNF formula with clauses of the form
  $(\neg x_1 \lor \cdots \lor \neg x_k)$, $(\neg x \lor y)$, $x$,
  and~$\neg x$.
\end{inparaenum}
A set~$\Gamma$ of Boolean relations is called $0$-valid ($1$-valid,
Horn, dual Horn, monotone, affine, bijunctive, complementive, $\IHSBp k$,
$\IHSBm k$) if \emph{every} relation in~$\Gamma$ is $0$-valid ($1$-valid,
Horn, dual Horn, monotone, affine, bijunctive, complementive, $\IHSBp k$,
$\IHSBm k$).

A formula constructed from atoms by conjunction, variable
identification, and existential quantification is called a
\emph{primitive positive formula} (\emph{pp-formula}). If $\varphi$ is
such a formula, we write again $[\varphi]$ for its set of models, i.e.,
the Boolean relation defined by $\varphi$. As above the coordinates of
this relation are understood to be the variables of~$\varphi$ in
lexicographic order, unless otherwise stated by explicit enumeration. We
denote by
$\cc \Gamma$ the set of all relations that can be expressed using
relations from $\Gamma\cup\set{\eq}$, conjunction, variable identification
(and permutation), cylindrification, and existential
quantification, i.e., the set of all relations that are primitive
positively definable from $\Gamma$ and equality. The set~$\cc \Gamma$ is called the \emph{co-clone}
generated by~$\Gamma$. A \emph{base} of a co-clone~$\B$ is a set of
relations~$\Gamma$ such that $\cc \Gamma = \B$, i.e., just a generating set
with regard to primitive positive definability including equality. Note
that traditionally
(e.g.~\cite{JanovMucnik1959}),
the notion of \emph{base} also involves minimality with
respect to set inclusion.
Our use
of the term \emph{base} is in accordance with~\cite{BoehlerRSV-05}, where
finite bases for all Boolean co-clones have been determined.
Some of these are listed in Table~\ref{tab:clones}. The sets of relations being
$0$-valid, $1$-valid, complementive, Horn, dual Horn, affine,
bijunctive, 2affine (both bijunctive and affine), monotone, $\IHSBp k$,
and $\IHSBm k$ each form a co-clone denoted by
$\iI_0$, $\iI_1$, $\iN_2$, $\iE_2$, $\iV_2$, $\iL_2$, $\iD_2$,
$\iD_1$, $\iM_2$, $\iS_{00}^k$, and~$\iS_{10}^k$, respectively;
see Table~\ref{tab:cclones}.

\begin{table}[t]
  \caption{Sets of Boolean relations with their names determined by co-clone inclusions}
  \label{tab:cclones}
  \centering
  \begin{displaymath}
    \begin{array}{lcl@{\hspace*{2cm}}lcl}
      \Gamma \subseteq \iI_0 & \Leftrightarrow & \Gamma \text{ is $0$-valid} &
      \Gamma \subseteq \iI_1 & \Leftrightarrow & \Gamma \text{ is $1$-valid}\\
      \Gamma \subseteq \iE_2 &  \Leftrightarrow & \Gamma \text{ is Horn} &
      \Gamma \subseteq \iV_2 &  \Leftrightarrow & \Gamma \text{ is dual Horn}\\
      \Gamma \subseteq \iM_2 &  \Leftrightarrow & \Gamma \text{ is monotone} &
      \Gamma \subseteq \iD_2 &  \Leftrightarrow & \Gamma \text{ is bijunctive}\\
      \Gamma \subseteq \iL_2 &  \Leftrightarrow & \Gamma \text{ is affine} &
      \Gamma \subseteq \iD_1 &  \Leftrightarrow & \Gamma \text{ is 2affine}\\
      \Gamma \subseteq \iN_2 &  \Leftrightarrow & \Gamma \text{ is complementive} &
      \Gamma \subseteq \iI   &  \Leftrightarrow & \Gamma \text{ is both $0$- and $1$-valid}\\
      \Gamma \subseteq \iS_{00}^k &  \Leftrightarrow & \Gamma \text{ is $\IHSBp k$} &
                                                                                                 \Gamma \subseteq \iS_{10}^k &  \Leftrightarrow & \Gamma \text{ is $\IHSBm k$}
    \end{array}
  \end{displaymath}
\end{table}

We will also use a weaker closure than $\cc \Gamma$, called
\emph{conjunctive closure} and denoted by~$\ccc \Gamma$, where the
constraint language~$\Gamma$ is closed under conjunctive definitions,
but not under existential quantification or addition of explicit
equality constraints.
Sets of relations of the form $W = \ccc{W\cup\set{\eq}}$ are called
\emph{weak systems} and are in a one-to-one correspondence with so-called
strong partial clones~\cite{Romov1981}. It is a well-known consequence of
the Galois theory developed in~\cite{Romov1981} that for every co-clone
$\cc{\Gamma'}$ whose corresponding clone is finitely generated (this
presents no restriction in the Boolean case), there is a largest partial
clone whose total part coincides with that clone,
cf.~\cite[Theorem~20.7.2]{Lau-06} or
see~\cite[Theorems~4.6, 4.7, 4.11]{SchnoorS-08} for a proof in the
Boolean case. This largest partial clone even is a strong partial clone,
and hence, there is a least weak system~$W$ under inclusion such that
$\cc W = \cc{\Gamma'}$. Any \emph{finite} weak generating set~$\Gamma$ of
this weak system~$W$, i.e., $W =\ccc{\Gamma\cup\set{\eq}}$, is called a
\emph{weak base} of~$\cc{\Gamma'}$,
see~\cite[Definition~4.2]{SchnoorS-08}. Such a set~$\Gamma$, in
particular, is a finite base of the co-clone~$\cc{\Gamma'}$. Finally, to
get from the closure operator~$\ccc{\Gamma\cup\set{\eq}}$
(which is hard to handle in the context of our problems)
to~$\ccc{\Gamma}$ (which is easy to handle), one needs the notion of
irredundancy. A relation~$R$ is called \emph{irredundant}, if it has
neither duplicate nor fictitious coordinates. It can be observed from the
proofs of Proposition~5.2 and Corollary~5.6 in~\cite{SchnoorS-08} or
from~\cite[Proposition~3.11]{SchnoorDiss}, that
$R\in \ccc{\Gamma\cup\set{\eq}}$ implies $R\in \ccc{\Gamma}$ for any
irredundant relation~$R$. Following Schnoor~\cite[p.~30]{SchnoorDiss}, we
call a weak base of~$\cc{\Gamma'}$ consisting exclusively of irredundant
relations an \emph{irredundant weak base}. 
Thus, if~$\Gamma$ is an irredundant weak base of~$\cc{\Gamma'}$, then the
minimality of the weak system $W = \ccc{\Gamma\cup\set{\eq}}$ implies that
$\Gamma\subseteq W\subseteq
\ccc{\Gamma'\cup\set{\eq}}$
(cf.~\cite[Corollary~4.3]{SchnoorS-08}), and thus
$\Gamma\subseteq \ccc{\Gamma'}$ because of irredundancy.
Hence, we obtain the following useful tool.

\begin{theorem}[{Schnoor~\cite[Corollary~3.12]{SchnoorDiss}}]\label{thm:weakbases}
  If\/~$\Gamma$ is an irredundant weak base of a co-clone~$\iC$,
  e.g.\ a minimal weak base of\/~$\iC$, then
  $\Gamma\subseteq \ccc{\Gamma'}$ holds for any base~$\Gamma'$
  of\/~$\iC$.
\end{theorem}

According to Lagerkvist~\cite{Lagerkvist-14}, a \emph{minimal weak base}
is an irredundant weak base satisfying an
additional minimality property that ensures small cardinality.
The utility of Theorem~\ref{thm:weakbases} comes in particular from the fact that Lagerkvist determined
minimal weak bases for all finitely generated Boolean co-clones
in~\cite{Lagerkvist-14}. For our purposes we note that each of the
co-clones $\iV$, $\iV_0$, $\iV_1$, $\iV_2$, $\iN$, $\iN_2$, and $\iI$ is
generated by a minimal weak base consisting of a single relation
(Table~\ref{tab:weakbases}).

Another source of weak base relations without duplicate coordinates comes
from the following construction:
let $\graphic$ be the $2^{n}$-ary relation that is given by the value
tables (in some chosen enumeration) of the~$n$ distinct $n$-ary
projection functions. More formally, let
$\beta\colon 2^{n}\to \set{0,1}^n$ be the reader's preferred bijection
between the index set $2^n = \set{0,\dotsc,2^{n-1}}$ and the set of all
arguments of an $n$-ary Boolean function~-- often lexicographic
enumeration is chosen here for presentational purposes, but the order of
enumeration of the $n$-tuples does not matter as long as it remains
fixed. Then $\graphic = \Set{e_i\circ \beta}{1\leq i\leq n}$ where
$e_i\colon \set{0,1}^n\to\set{0,1}$ denotes the projection function onto
the $i$-th coordinate. Let~$C$ be a clone with corresponding
co-clone~$\iC$. Since $\iC$ is closed with respect to intersection of
relations of identical arity, for any $k$-ary relation $R$, there is a
least $k$-ary relation in $\iC$ containing $R$, scilicet
$\GammaF{C}{R}
    :=\bigcap\Set{R'\in \iC}{R'\supseteq R, R'\text{ $k$-ary}}$.
Traditionally, e.g.~\cite[Sect.~2.8, p.~134]{Lau-06} or
\cite[Definition~1.1.16, p.~48]{PoeschelK-79}, this relation is denoted
by $\Gamma_C(R)$, but here we have chosen a different notation to avoid
confusion with constraint languages.
It is well known, e.g.~\cite[Satz~1.1.19(i), p.~50]{PoeschelK-79}, and
easy to see that $\GammaF{C}{R}$ is completely
determined by the $\ell$-ary part of~$C$ whenever $\ell \geq \card{R}$:
given any enumeration of $\emptyset\neq R = \set{r_1,\dotsc,r_{\ell}}$
(for technical reasons we have to exclude the case $\ell = 0$ in
this presentation because we do not consider clones with nullary
operations here) we have
$\GammaF{C}{R}
  =\Set{f\circ(r_1,\dotsc,r_{\ell})}{f\in C, f \text{ $\ell$-ary}}$,
where $f\circ(r_1,\dotsc,r_{\ell})$ denotes the row-wise application
of~$f$ to a matrix whose columns are formed by the tuples
$r_1,\dotsc,r_{\ell}$.
Relations of the form $\GammaF{C}{\graphic}$ represent the $n$-ary part
of the clone $C$ as a $2^{n}$-ary relation and are called \emph{$n$-th
graphic} of $C$ (cf.\ e.g.~\cite[p.~133 and Theorem~2.8.1(b)]{Lau-06}).
Indeed, the previous characterization of $\GammaF{C}{\graphic}$ yields
$\GammaF{C}{\graphic}
=\Set{f\circ(e_1\circ\beta,\dotsc,e_n\circ\beta)}{
                        f\in C, f \text{ $n$-ary}}
=\Set{f\circ(e_1,\dotsc,e_n)\circ\beta}{
                        f\in C, f \text{ $n$-ary}}
=\Set{f\circ\beta}{f\in C, f \text{ $n$-ary}}$.
With the help of this description of $\GammaF{C}{\graphic}$
and standard clone theoretic manipulations, one can easily verify the
following result, identifying possible candidates for irredundant
singleton weak bases.

\begin{theorem}[{\cite[Theorem~4.11]{SchnoorS-08}}]%
\label{thm:weakbases-from-graphics}
  Let $C$ be a clone and $R =\GammaF{C}{\set{r_1,\dotsc,r_n}}$ with
  $n\geq 1$, then $\GammaF{C}{\graphic}$ gives a singleton weak base of
  $\cc{\set{R}}$ without duplicate coordinates.
\end{theorem}

\begin{table}[b]
  \caption{Minimal weak bases for some co-clones}
  \label{tab:weakbases}
  \centering
  \begin{displaymath}
    \begin{array}[t]{@{}lcl@{}}
      R_{\iL} &=& \even^4\\
      R_{\iL_0} &=& {\even^3} \times \set{0}\\
      R_{\iL_1} &=& {\odd^3} \times \set{1}\\
      R_{\iL_2} &=& {\even^3_{3\neq}} \times \set{0} \times \set{1}\\
      R_{\iL_3} &=& \even^4_{4\neq}\\
      R_{\iN} &=& {\even^4}\cap S_0
    \end{array}
    \qquad
    \begin{array}[t]{@{}lcl@{}}
      R_{\iV} &=& (S_1\times \set{0,1})\cap (\set{0,1}\times S_2)\\
      R_{\iV_0} &=& S_1\times \set{0}\\
      R_{\iV_1} &=& R_{\iV}\times \set{1}\\
      R_{\iV_2} &=& S_1\times \set{0}\times \set{1}\\
      R_{\iN_2} &=& [R_{\iN}(x_1,\dotsc,x_4)\land \bigwedge_{i=1}^4 x_{i+4}
                    \eq\neg x_i]\\
      R_{\iI} &=& [S_1(\neg x_1,\neg x_2,\neg x_3)\land S_1(x_4,x_2,x_3)]
    \end{array}
  \end{displaymath}
\end{table}

\subsection{Approximability, Reductions, and Completeness}

We assume that the reader has a basic knowledge of approximation
algorithms and complexity theory.  We recall some basic notions of
approximation algorithms and complexity theory; for details see the
monographs~\cite{AusielloCGKMSP-99,CreignouKS-01}.

A \emph{combinatorial optimization problem}~$\Pcal$ is a quadruple
$(I, \sol, \obj, \goal)$, where:
\begin{compactitem}[$\bullet$]
\item $I$ is the set of admissible \emph{instances} of~$\Pcal$.
\item $\sol(x)$ denotes the set of \emph{feasible solutions} for every
  instance~$x\in I$.
\item $\obj(x,y)$ denotes the non-negative integer measure of~$y$
  for every instance~$x\in I$ and every feasible
  solution~$y\in\sol(x)$;
  $\obj$ is also called \emph{objective function}.
\item $\goal \in \set{\min, \max}$ denotes the
      \emph{optimization goal} for~$\Pcal$.
\end{compactitem}
A combinatorial optimization problem is said to be an
\emph{$\NP$-optimization problem} ($\NPO$-problem) if
\begin{compactitem}[$\bullet$]
\item the instances and solutions are recognizable in polynomial
  time,
\item the size of the solutions in~$\sol(x)$ is polynomially bounded
  in the size of~$x$, and
\item the objective function $\obj$ is computable in polynomial
  time.
\end{compactitem}
The optimal value of the objective function for the solutions of an
instance~$x$ is denoted by $\OPT(x)$.  In our case the optimization
goal will always be minimization, i.e., $\OPT(x)$ will be the minimum.

Given an instance $x \in I$ with a feasible solution $y \in \sol(x)$ and
a real number $r\geq 1$, we say that $y$ is \emph{$r$-approximate} if
$\obj(x,y)\leq r\OPT(x)$ holds and our goal is minimization, or
$\obj(x,y)\geq \OPT(x)/r$ and we consider a maximization problem.

Let~$A$ be an algorithm that for any instance~$x$ of~$\Pcal$ such that
$\sol(x)\neq \emptyset$ returns a feasible solution $A(x) \in \sol(x)$.
Given an arbitrary
function $r\colon \NN \to [1,\infty)$, we say that~$A$ is an
$r(n)$-approximate algorithm for~$\Pcal$ if for any instance $x \in I$
having feasible solutions the algorithm returns an
$r(\card{x})$-approximate solution,
where~$\card x$ is the size of~$x$. If an $\NPO$ problem~$\Pcal$
admits an $r(n)$-approximate polynomial-time algorithm, we say
that~$\Pcal$ is approximable within $r(n)$.

An $\NPO$ problem~$\Pcal$ is in the class $\PO$ if the optimum is
computable in polynomial time (i.e. if~$\Pcal$ admits a
$1$-approximate polynomial-time algorithm). $\Pcal$ is in the class
$\APX$ ($\pAPX$) if it is approximable within a constant (polynomial)
function in the size of the instance~$x$.  $\NPO$ is the class of all
$\NPO$ problems and $\NPOPB$ is the class of all $\NPO$ problems where
the objective function is polynomially bounded. The following
inclusions hold for these approximation complexity classes:
$\PO \subseteq \APX \subseteq \pAPX \subseteq \NPO$. All inclusions
are strict unless $\P = \NP$.

For reductions among decision problems we use the polynomial-time
many-one reduction denoted by~$\mle$. Many-one equivalence between
decision problems is denoted by~$\meq$.  For reductions among
optimization problems we use approximation preserving reductions, also
called $\AP$-reductions, denoted by~$\aple$. $\AP$-equivalence between
optimization problems is denoted by~$\apeq$.

We say that an optimization problem $\Pcal$ \emph{$\AP$-reduces} to
another optimization problem $\Qcal$, denoted $\Pcal \aple \Qcal$, if
there are two polynomial-time computable functions~$f$ and~$g$ and a
real constant~$\alpha\geq 1$ such that for all $r>1$ and all
$\Pcal$-instances~$x$ the following conditions hold.
\begin{compactitem}[$\bullet$]
\item $f(x)$ is a $\Qcal$-instance or the generic unsolvable
      instance~$\bot$ (which is not part of~$\Qcal$).
\item If $x$ admits feasible solutions, then $f(x)$ is different from $\bot$
  and also admits feasible solutions.
\item For any feasible solution~$y'$ of~$f(x)$, $g(x,y')$ is a
  feasible solution of~$x$.
\item If~$y'$ is an $r$-approximate solution of the $\Qcal$-instance~$f(x)$,
  then $g(x,y')$ is an $(1+(r-1)\alpha + \Lo(1))$-approximate solution
  of the $\Pcal$-instance~$x$, where $\Lo(1)$ refers to the size of~$x$.
\end{compactitem}
Our definition of $\AP$-reducibility slightly extends the one
in~\cite{AusielloCGKMSP-99} by introducing a generic unsolvable
instance~$\bot$.  This extension allows us to reduce problems with
unsolvable instances to such without as long as the unsolvable
instances can be detected in polynomial time, by making $f$
map the unsolvable instances to~$\bot$.
This practice has been implicit in previous work, e.g.~\cite{KhannaSTW-01}.

We also need a slightly non-standard variation of $\AP$-reductions. We
say that an optimization problem $\Pcal$ \emph{$\AP$-Turing-reduces}
to another optimization problem $\Qcal$ if there is a polynomial-time
oracle algorithm~$A$ and a constant~$\alpha\geq 1$ such that for all
$r>1$ on any input~$x$ for~$\Pcal$
\begin{compactitem}[$\bullet$]
\item if all oracle calls with a $\Qcal$-instance~$x'$ are answered
  with a feasible $\Qcal$-solution~$y$ for~$x'$, then~$A$
  outputs a feasible $\Pcal$-solution for~$x$, and
\item if for every call the oracle answers with an $r$-approximate
  solution, then $A$ computes a $(1+(r-1)\alpha + \Lo(1))$-approximate
  solution for the $\Pcal$-instance~$x$.
\end{compactitem}
It is straightforward to check that $\AP$-Turing-reductions are
transitive. Moreover, if $\Pcal$ $\AP$-Turing-reduces to $\Qcal$ with
constant $\alpha$ and $\Qcal$ has an $r(n)$-approximation algorithm,
then there is an $\alpha r(n)$-approximation algorithm for $\Pcal$.

We will relate our problems to well-known optimization problems, by
calling the problem~$\Pcal$ under investigation $\Qcal$-complete if
$\Pcal\apeq \Qcal$. This notion of completeness is stricter than the
one in~\cite{KhannaSTW-01}, since the latter relies on
$\mathrm{A}$-reductions. For $\Qcal$, we will consider the following
optimization problems analyzed in~\cite{KhannaSTW-01}.

\optproblem{$\OptMinOnes(\Gamma)$}%
{A conjunctive formula~$\varphi$ over relations from~$\Gamma$.}%
{An assignment~$m$ satisfying~$\varphi$.}%
{Minimum Hamming weight $\hw(m)$.}

\optproblem{$\OptWMinOnes(\Gamma)$}%
{A conjunctive formula~$\varphi$ over relations from~$\Gamma$ and a weight
  function $w\colon V \to \NN$ assigning non-negative integer weights
  to the variables of $\varphi$.}%
{An assignment~$m$ satisfying~$\varphi$.}%
{Minimum value $\sum_{x: m(x)=1}w(x)$.}

We now define some well-studied problems to which we will relate our problems.
Note that these problems do not depend on any parameter.

\optproblem{$\NCW$}%
{A matrix $A \in \ZZ_2^{k\times l}$ and a vector $m\in \ZZ_2^l$.}%
{A vector $x\in \ZZ_2^k$.}%
{Minimum Hamming distance $\hd(xA,m)$.}

\optproblem{$\MinDist$}%
{A matrix $A\in \ZZ_2^{k\times l}$.}%
{A non-zero vector $x \in \ZZ_2^l$ with $A x = 0$.}%
{Minimum Hamming weight $\hw(x)$.}

\optproblem{$\MinHD$}%
{A conjunctive formula~$\varphi$ over relations from~$\set{x\lor y\lor
  \neg z, x, \neg x}$.}%
{An assignment $m$ to $\varphi$.}%
{Minimum number of unsatisfied conjuncts of $\varphi$.}

$\NCW$, $\MinDist$ and $\MinHD$ are known to be $\NP$-hard to approximate
within a factor $2^{\LOmega(\log^{1-\varepsilon}(n))}$ for every
$\varepsilon > 0$~\cite{AroraBSS-97,DumerMS-03,KhannaSTW-01}.
Thus if a problem $\Pcal$ is equivalent to any of these problems,
it follows that $\Pcal \notin \APX$ unless $\P=\NP$.

\subsection{Satisfiability}

We also use the classic problem $\SAT(\Gamma)$ asking for the
satisfiability of a given conjunctive formula over a constraint
language~$\Gamma$. Schaefer~\cite{Schaefer-78} completely classified its complexity.
$\SAT(\Gamma)$ is polynomial-time decidable
if~$\Gamma$ is $0$-valid $(\Gamma \subseteq \iI_0)$, $1$-valid
$(\Gamma \subseteq \iI_1)$, Horn $(\Gamma \subseteq \iE_2)$, dual Horn
$(\Gamma \subseteq \iV_2)$, bijunctive $(\Gamma \subseteq \iD_2)$, or
affine $(\Gamma \subseteq \iL_2)$; otherwise it is
$\NP$-complete. Moreover, we need the decision
problem $\AnotherSat(\Gamma)$: Given a conjunctive formula
over~$\Gamma$ and a satisfying assignment~$m$, is there another
satisfying assignment~$m'$ different from~$m$?
The complexity of this problem was completely classified by
Juban~\cite{Juban-99}.
$\AnotherSat(\Gamma)$
is polynomial-time decidable if~$\Gamma$ is both $0$- and $1$-valid
$(\Gamma \subseteq \iI)$, complementive $(\Gamma \subseteq \iN_2)$,
Horn $(\Gamma \subseteq \iE_2)$, dual Horn $(\Gamma \subseteq \iV_2)$,
bijunctive $(\Gamma \subseteq \iD_2)$, or affine
$(\Gamma \subseteq \iL_2)$; otherwise it is $\NP$-complete.

\subsection{Linear and Integer Programming}

A \emph{unimodular matrix} is a square integer matrix having
determinant $+1$ or $-1$. A \emph{totally unimodular matrix} is a
matrix for which every square non-singular submatrix is unimodular. A
totally unimodular matrix need not be square itself. Any totally
unimodular matrix has only $0$, $+1$ or $-1$ entries.
If~$A$ is a totally unimodular matrix and
$\vec b$ is an integral vector, then for any given
linear functional~$f$ such that the linear program
$\min\Set{f(\vec x)}{A \vec x \geq\vec b}$
has a real minimum~$\vec{x}$, it also has an integral minimum point~$\vec{x}$.
That is, the feasible region $\Set{\vec{x}}{A\vec{x}\geq \vec{b}}$ is an
integral polyhedron. For this reason, linear programming methods can be
used to obtain the solutions for integer linear programs in this case.
Linear programs can be solved in polynomial time, hence so can integer
programs with totally unimodular matrices.
For details see the monograph by Schrijver~\cite{Schrijver-86}.

\section{Results}
\label{sec:results}

This section presents the problems we consider and our results; the
proofs follow in subsequent sections. The input to all our problems is
a conjunctive formula over a constraint language. The satisfying
assignments of the formula, i.e.\ its models or solutions, form a
Boolean relation that can be understood as an associated generalized
binary code. As for linear codes, the minimization target is always
the Hamming distance between the codewords or models.  Our three problems
differ in the information additionally available for computing the
required Hamming distance.

Given a formula and an arbitrary assignment, the first problem asks
for a solution closest to the given assignment.

\optproblem{$\NearestSol(\Gamma)$, $\NSOL(\Gamma)$}%
{A conjunctive formula~$\varphi$ over relations from~$\Gamma$ and an
  assignment~$m$ to the variables occurring in~$\varphi$, which is not
  required to satisfy~$\varphi$.}%
{An assignment~$m'$ satisfying~$\varphi$
(i.e.\ a codeword of the code described by~$\varphi$).}%
{Minimum Hamming distance $\hd(m,m')$.}

Note that the problem generalizes the $\OptMinOnes$ problem
from~\cite{KhannaSTW-01}. Indeed, if we take the all-zero assignment
$m = 0 \cdots 0$ as part of the input, we get exactly the
$\OptMinOnes$ problem as a special case.

\begin{theorem}[{\textnormal{illustrated in Figure~\ref{fig:nsol-coclones}}}]\label{thm:NSol}
  For a given Boolean constraint language~$\Gamma$ the optimization problem
  $\NSOL(\Gamma)$ is
  \begin{compactenum}[(i)]
  \item in~$\PO$ if\/ $\Gamma$ is
    \begin{compactenum}[(a)]
    \item 2affine $(\Gamma\subseteq\iD_1)$ or
    \item monotone $(\Gamma\subseteq\iM_2)$;
    \end{compactenum}
  \item $\APX$-complete if
    \begin{compactenum}[(a)]
    \item
      $\Gamma$ generates  $\iD_2$
      $(\cc \Gamma = \iD_2)$,
      or
    \item
      $[x\lor y] \in \cc \Gamma$
      and $\Gamma$ is $\IHSBp k$
      $(\iS_0^2 \subseteq \cc \Gamma \subseteq \iS_{00}^k)$ for some
      $k \in \NN$, $k\geq 2$, or
    \item $[\neg x \lor \neg y] \in \cc \Gamma$
      and $\Gamma$ is $\IHSBm k$
      $(\iS_1^2 \subseteq \cc \Gamma \subseteq \iS_{10}^k)$ for some
      $k \in \NN$, $k\geq 2$, or
    \end{compactenum}
  \item $\NCW$-complete if\/ $\Gamma$ is
    exactly affine $(\iL\subseteq\cc \Gamma\subseteq\iL_2)$;
  \item $\MinHD$-complete if\/ $\Gamma$ is
    \begin{compactenum}[(a)]
    \item exactly Horn $(\iE\subseteq\cc \Gamma\subseteq \iE_2)$ or
    \item exactly dual Horn $(\iV\subseteq\cc \Gamma\subseteq \iV_2)$;
    \end{compactenum}
  \item $\pAPX$-complete if\/
    $\Gamma$ does not contain an affine relation and it is
    \begin{compactenum}[(a)]
    \item $0$-valid $(\iN\subseteq\cc \Gamma \subseteq \iI_0)$ or
    \item $1$-valid $(\iN\subseteq\cc \Gamma \subseteq \iI_1)$;
    and
    \end{compactenum}
  \item otherwise $(\iN_2\subseteq\cc\Gamma)$ it is $\NP$-complete to
    decide whether a feasible solution for $\NSOL(\Gamma)$ exists.
  \end{compactenum}
\end{theorem}
\begin{proof}
  The proof is split into several propositions presented in
  Section~\ref{sec:proofsNSOL}.
  \begin{compactenum}[(i)]
  \item 
    See Propositions~\ref{prop:NSOL-iD1} and~\ref{prop:NSOL-iM2}.
  \item 
    See Propositions~\ref{prop:NSOL-iS0^2-iS1^2}, \ref{prop:NSOL-D_2},
    and~\ref{prop:NSOL-S_00}.
  \item 
    See Corollary~\ref{cor:NSOL-NCW-hard-L} and
    Proposition~\ref{prop:NSOL-MinOnes-to-weak_base-L_2}.
  \item 
    See Propositions~\ref{prop:iV_2_to_MinOnes} and~\ref{prop:MinOnes_to_iV_2}.
  \item 
    See Proposition~\ref{prop:duphardnc}.
  \item 
    See Proposition~\ref{prop:NSOL-iN_2-BR}.
    \qedhere
  \end{compactenum}
\end{proof}

\newcommand\NSOLstyle
    {R/.cstyle=PO,M/.cstyle=PO,D/.cstyle=PO,
     D2/.cstyle=APX,S2/.cstyle=APX,S3/.cstyle=APX,
     S0/.cstyle=NA,S1/.cstyle=NA,
     E/.cstyle=MHD,V/.cstyle=MHD,
     L/.cstyle=NCW,
     N/.cstyle=pAPX,N2/.cstyle=NPO,N/.lstyle={text=white},
     I/.cstyle=pAPX,I2/.cstyle=NPO,I/.lstyle={text=white}
    }
\begin{figure}
  \centering
  \expandafter\postlattice\expandafter[\NSOLstyle]
  \bigskip

  \begin{tabular}{lll}
    \makebox[\boxwidth][l]{\protect\postlegend{PO} in PO}
    & \makebox[\boxwidth][l]{\protect\postlegend{APX} $\APX$-complete}
      & \makebox[\boxwidth][l]{\protect\postlegend{NCW} $\NCW$-complete}
  \\
    \makebox[\boxwidth][l]{\protect\postlegend{MHD} $\MinHD$-complete}
    & \makebox[\boxwidth][l]{\protect\postlegend{pAPX} $\pAPX$-complete}
      & \makebox[\boxwidth][l]{\protect\postlegend{NPO} feasibility $\NP$-complete}
  \\
    & \makebox[\boxwidth][l]{\protect\postlegend{NA} not applicable}
      &
  \\
  \end{tabular}
  \caption{Lattice of co-clones with complexity classification for
    $\NSOL$.}
  \label{fig:nsol-coclones}
\end{figure}

Given a constraint and one of its solutions, the second problem asks
for another solution closest to the given one.

\optproblem{$\NextSol(\Gamma)$, $\XSOL(\Gamma)$}%
{A conjunctive formula~$\varphi$ over relations from~$\Gamma$ and a
  satisfying assignment~$m$ (to the variables mentioned in~$\varphi$).}%
{An assignment~$m'\neq m$ satisfying~$\varphi$.}%
{Minimum Hamming distance $\hd(m,m')$.}

The difference between the problems $\NearestSol$ and $\NextSol$ is
the knowledge, or its absence, whether the input assignment satisfies
the constraint. Moreover, for $\NearestSol$ we may output the given
assignment if it satisfies the formula while for $\NextSol$ we have to
output an assignment different from the one given as the input.

\begin{theorem}[{\textnormal{illustrated in Figure~\ref{fig:xsol-coclones}}}]\label{thm:NOSol}
  For every constraint language~$\Gamma$ the optimization problem
  $\XSOL(\Gamma)$ is
  \begin{compactenum}[(i)]
  \item in~$\PO$ if
    \begin{compactenum}[(a)]
    \item $\Gamma$ is bijunctive $(\Gamma\subseteq\iD_2)$ or
    \item $\Gamma$ is $\IHSBp k$
      $(\Gamma\subseteq\iS_{00}^k)$ for some $k \in \NN$, $k \geq 2$ or
    \item $\Gamma$ is $\IHSBm k$
      $(\Gamma\subseteq\iS_{10}^k)$ for some $k \in \NN$, $k \geq 2$;
    \end{compactenum}
  \item $\MinDist$-complete if\/
    $\Gamma$ is exactly affine
    $(\iL\subseteq\cc \Gamma\subseteq\iL_2)$;
  \item $\MinHD$-complete under $\AP$-Turing-reductions if\/~$\Gamma$ is
    \begin{compactenum}[(a)]
    \item exactly Horn $(\iE\subseteq\cc \Gamma\subseteq \iE_2)$ or
    \item exactly dual Horn $(\iV\subseteq\cc \Gamma\subseteq \iV_2)$;
    \end{compactenum}
  \item in $\pAPX$ if\/~$\Gamma$ is
    \begin{compactenum}[(a)]
    \item exactly both $0$-valid and $1$-valid $(\cc \Gamma = \iI)$ or
    \item exactly complementive $(\iN \subseteq \cc \Gamma \subseteq \iN_2)$,
    \end{compactenum}
    where $\XSOL(\Gamma)$ is $n$-approximable but not
    $(n^{1-\varepsilon})$-approximable for any $\varepsilon>0$ unless $\P=\NP$;
  \item and otherwise $(\iI_0\subseteq\cc \Gamma$ or
    $\iI_1\subseteq\cc \Gamma)$ it is $\NP$-complete to decide whether a
    feasible solution for $\XSOL(\Gamma)$ exists.
  \end{compactenum}
\end{theorem}
\begin{proof}
  The proof is split into several propositions presented in
  Section~\ref{sec:proofsXSOL}.
  \begin{compactenum}[(i)]
  \item 
    See Propositions~\ref{prop:XSOL-iD2} and~\ref{prop:XSOL-iI00m}.
  \item 
    See Proposition~\ref{prop:MinDist-hardness-XSOL}.
  \item 
    See Corollary~\ref{cor:XSOL-Horn-dual_Horn}.
  \item 
    See Propositions~\ref{prop:XSOL-iI_0-iI_1} and~\ref{prop:tightxsol}.
  \item 
    See Proposition~\ref{prop:XSOL-iI_0-iI_1}.
    \qedhere
  \end{compactenum}
\end{proof}

\newcommand\NOSOLstyle
    {R/.cstyle=PO,M/.cstyle=PO,D/.cstyle=PO,
     S2/.cstyle=PO,S3/.cstyle=PO,
     S0/.cstyle=NA,S1/.cstyle=NA,
     E/.cstyle=MHD,V/.cstyle=MHD,
     L/.cstyle=MD,
     N/.cstyle=pAPX,N/.lstyle={text=white},
     I/.cstyle=NPO,Ix/.cstyle=pAPX,I/.lstyle={text=white}
    }
\begin{figure}
  \centering
  \expandafter\postlattice\expandafter[\NOSOLstyle]
  \bigskip

  \begin{tabular}{lll}
    \makebox[\boxwidth][l]{\protect\postlegend{PO} in PO}
    & & \makebox[\boxwidth][l]{\protect\postlegend{MD} $\MinDist$-complete}
  \\
    \makebox[\boxwidth][l]{\protect\postlegend{MHD} $\MinHD$-complete}
    & \makebox[\boxwidth][l]{\protect\postlegend{pAPX} in $\pAPX$ and tight}
      & \makebox[\boxwidth][l]{\protect\postlegend{NPO} feasibility $\NP$-complete}
  \\
    & \makebox[\boxwidth][l]{\protect\postlegend{NA} not applicable}
      &
  \\
  \end{tabular}
  \caption{Lattice of co-clones with complexity classification for
    $\XSOL$.}
  \label{fig:xsol-coclones}
\end{figure}

The third problem does not take any assignments as input, but asks for
two solutions which are as close to each other as possible. We
optimize once more the Hamming distance between the solutions.

\optproblem{$\MinSolDistance(\Gamma)$, $\MSD(\Gamma)$}%
{A conjunctive formula~$\varphi$ over relations from~$\Gamma$.}%
{Two satisfying truth assignments~$m\neq m'$
 to the variables occurring in~$\varphi$.}%
{Minimum Hamming distance $\hd(m,m')$.}

The $\MinSolDistance$ problem enlarges the notion of minimum distance
of an error correcting code. The following theorem is a more
fine-grained analysis of the result published by Vardy
in~\cite{Vardy-97}, extended to an optimization problem.

\begin{theorem}[{\textnormal{illustrated in Figure~\ref{fig:msd-coclones}}}]\label{thm:MSD}
  For any constraint language~$\Gamma$
  the optimization problem $\MSD(\Gamma)$ is
  \begin{compactenum}[(i)]
  \item in~$\PO$ if\/~$\Gamma$ is
    \begin{compactenum}[(a)]
    \item bijunctive $(\Gamma\subseteq\iD_2)$ or
    \item Horn $(\Gamma\subseteq\iE_2)$ or
    \item dual Horn $(\Gamma\subseteq\iV_2)$;
    \end{compactenum}
  \item $\MinDist$-complete if\/~$\Gamma$ is exactly affine
    $(\iL\subseteq\cc \Gamma\subseteq\iL_2)$;
  \item in $\pAPX$ if\/ $\dup^3 \in \cc \Gamma$
    and~$\Gamma$ is both $0$-valid and $1$-valid
    $(\iN \subseteq \cc \Gamma \subseteq \iI)$, 
    where
    $\MSD(\Gamma)$ is $n$-approximable but not
    $(n^{1{-}\varepsilon})$-approximable for any $\varepsilon>0$ unless $\P=\NP$;
    and
  \item otherwise
    $(\iN_2\subseteq\cc \Gamma$ or
    $\iI_0\subseteq\cc \Gamma$ or
    $\iI_1\subseteq\cc \Gamma)$
    it is $\NP$-complete to decide whether a feasible solution for
    $\MSD(\Gamma)$ exists.
  \end{compactenum}
\end{theorem}
\begin{proof}
  The proof is split into several propositions presented in
  Section~\ref{sec:proofsMSD}.
  \begin{compactenum}[(i)]
  \item See Propositions~\ref{prop:MSD-iD2} and~\ref{prop:MSD-iE2-iV2}.
  \item See Proposition~\ref{prop:msdaffine}.
  \item For $\Gamma \subseteq \iI$, every formula~$\varphi$
    over~$\Gamma$ has at least two solutions since it is both
    $0$-valid and $1$-valid. Thus $\TSSAT(\Gamma)$ is in~$\P$, and
    Proposition~\ref{prop:linearapprox} yields that $\MSD(\Gamma)$ is
    $n$-approximable. By
    Proposition~\ref{prop:tightnessOfLinearapprox} this approximation
    is indeed tight.
  \item According to~\cite{Juban-99}, $\AnotherSat(\Gamma)$ is
    $\NP$-hard for $\iI_0\subseteq\cc \Gamma$, or
    $\iI_1\subseteq\cc \Gamma$. By Lemma~\ref{lem:asattotwo} it
    follows that $\TSSAT(\Gamma)$ is $\NP$-hard, too.  For
    $\iN_2\subseteq\cc \Gamma$ we can reduce the $\NP$-hard problem
    $\SAT(\Gamma)$ to $\TSSAT(\Gamma)$.  Hence it is $\NP$-complete to
    decide whether a feasible solution for $\MSD(\Gamma)$ exists in all
    three cases.%
    \qedhere
  \end{compactenum}
\end{proof}

\newcommand\MSDstyle
    {R/.cstyle=PO,M/.cstyle=PO,D/.cstyle=PO,
     E/.cstyle=PO,V/.cstyle=PO,
     S2/.cstyle=PO,S3/.cstyle=PO,
     S0/.cstyle=NA,S1/.cstyle=NA,
     L/.cstyle=MD,
     N/.cstyle=pAPX,N2/.cstyle=NPO,N/.lstyle={text=white},
     I/.cstyle=NPO,Ix/.cstyle=pAPX,I/.lstyle={text=white}
    }
\begin{figure}
  \centering
  \expandafter\postlattice\expandafter[\MSDstyle]
  \bigskip

  \begin{tabular}{lll}
    \makebox[\boxwidth][l]{\protect\postlegend{PO} in PO}
    & & \makebox[\boxwidth][l]{\protect\postlegend{MD} $\MinDist$-complete}
  \\
    & \makebox[\boxwidth][l]{\protect\postlegend{pAPX} in $\pAPX$ and tight}
      & \makebox[\boxwidth][l]{\protect\postlegend{NPO} feasibility $\NP$-complete}
  \\
    & \makebox[\boxwidth][l]{\protect\postlegend{NA} not applicable}
      &
  \\
  \end{tabular}
  \caption{Lattice of co-clones with complexity classification for
    $\MSD$.}
  \label{fig:msd-coclones}
\end{figure}

The three optimization problems can be transformed into
decision problems in the usual way. We add an integer bound $k$ to the
input and ask if the Hamming distance satisfies the inequality
$\hd(m,m') \leq k$. This way we obtain the corresponding decision
problems $\dXSOL$, $\dNSOL$, and $\dMSD$, respectively. Their
complexity follows immediately from the theorems above. All
cases in $\PO$ become polynomial-time decidable, whereas the other
cases, which are $\APX$-hard, become $\NP$-complete.
This way we obtain dichotomy theorems classifying the decision problems
as polynomial or $\NP$-complete for all sets~$\Gamma$ of relations.
We obtain the following dichotomies for each of the respective
decision problems.
\begin{corollary}\label{cor:dMSD}
  For each constraint language~$\Gamma$
  \begin{compactitem}
  \item $\dNSOL(\Gamma)$ is in~$\P$ if\/~$\Gamma$ is 2affine or
    monotone, and it is $\NP$-complete otherwise.
  \item $\dXSOL(\Gamma)$ is in~$\P$ if\/~$\Gamma$ is bijunctive, $\IHSBp k$, or $\IHSBm k$,
    and it is $\NP$-complete otherwise.
  \item $\dMSD(\Gamma)$ is in~$\P$ if\/~$\Gamma$ is bijunctive, Horn, or
    dual-Horn, and it is $\NP$-complete otherwise.
  \end{compactitem}
\end{corollary}

\section{Applicability of Clone Theory and Duality}
\label{sec:proofsPP}

We show that clone theory is applicable to the problem $\NSOL$, as
well as a possibility to exploit inner symmetries between co-clones,
which shortens several proofs in the following sections.

\subsection{Nearest Solution}

There are two natural versions of $\NSOL(\Gamma)$. In one version the
formula~$\varphi$ is quantifier free while in the other one we do
allow existential quantification. We call the former version
$\NSOL(\Gamma)$ and the latter $\NSOLpp(\Gamma)$ and show
that both versions are equivalent.

Let $\dNSOL(\Gamma)$ and $\dNSOLpp(\Gamma)$ be the decision problems
corresponding to $\NSOL(\Gamma)$ and $\NSOLpp(\Gamma)$, asking whether
there is a satisfying assignment within a given bound.

\begin{proposition}\label{prop:quantifiers}
  For any constraint language~$\Gamma$ we have the equivalences
  $\dNSOL(\Gamma)\meq\dNSOLpp(\Gamma)$ and
  $\NSOL(\Gamma)\apeq\NSOLpp(\Gamma)$.
\end{proposition}
\begin{proof}
  The reduction from left to right is trivial in both cases.  For the
  other direction, consider first an instance of $\dNSOLpp(\Gamma)$
  with formula~$\varphi$, assignment~$m$, and bound~$k$. Let
  $x_1, \ldots, x_n$ be the free variables of $\varphi$ and let
  $y_1, \ldots, y_\ell$ be the existentially quantified ones,
  which we can assume to be disjoint. By discarding variables $y_i$
  while not changing $[\varphi]$, we can assume that each variable
  $y_i$ occurs in at least one atom of~$\varphi$.
  We construct a quantifier-free
  formula $\varphi'$, where the non-quantified variables of~$\varphi$
  get duplicated by a factor $\lambda:=(n+\ell+1)^2$ such that the
  effect of quantified variables becomes negligible.  For each
  variable~$z$ we define the set $B(z)$ as follows:
  \begin{eqnarray*}
    B(z)
    &=&
    \begin{cases}
      \set{x_i^1,\dots,x_i^\lambda} & \text{if
        $z=x_i$ for some $i\in \set{1, \ldots, n}$},\\
      \set{y_i} & \text{if $z=y_i$ for some $i\in \set{1, \ldots, \ell}$}.
    \end{cases}
  \end{eqnarray*}
  For every atom $R(z_1, \ldots, z_s)$ in~$\varphi$, the
  quantifier-free formula $\varphi'$ over the variables
  $\bigcup_{i=1}^n B(x_i) \cup \bigcup_{i=1}^\ell B(y_i)$ contains the
  atom $R(z_1', \ldots, z_s')$ for every combination
  $(z_1', \ldots, z_s')$ from $B(z_1)\times \cdots \times
  B(z_s)$.
  Moreover, we construct an assignment $B(m)$ of~$\varphi'$ by
  assigning to every variable~$x_i^j$ the value $m(x_i)$ and to~$y_i$
  the value~$0$. Note that because there is an upper bound on the
  arities of relations from~$\Gamma$, this is a polynomial time
  construction.

  We claim that $\varphi$ has a solution $m'$ with $\hd(m, m') \le k$
  if and only if $\varphi'$ has a solution $m''$ with
  $\hd(B(m), m'')\le k\lambda +\ell$.
  First, observe that if $m'$
  with the desired properties exists, then there is an extension
  $m'_{\mathrm{e}}$ of $m'$ to the $y_i$ that satisfies all
  atoms. Define $m''$ by setting $m''(x_i^j):= m'(x_i)$ and
  $m''(y_i):= m'_{\mathrm{e}}(y_i)$ for all $i$ and $j$. Then $m''$ is
  clearly a satisfying assignment of $\varphi'$. Moreover, $m''$ and
  $B(m)$ differ in at most $k\lambda$ variables among the
  $x_i^j$. Since there are only~$\ell$ other variables~$y_i$, we get
  $\hd(m'', B(m))\leq k\lambda + \ell$ as desired.

  Now suppose $m''$ satisfies $\varphi'$ with
  $\hd(B(m), m'') \le k\lambda +\ell$. We may assume for each~$i$
  that $m''(x_i^1) = \cdots = m''(x_i^\lambda)$.  Indeed, if
  this is not the case, then setting all~$x_i^j$ to
  $B(m)(x_i^j)=m(x_i)$ will result in a satisfying assignment closer to
  $B(m)$. After at most~$n$ iterations we get some~$m''$ as
  desired. Now define an assignment~$m'$ for~$\varphi$ by setting
  $m'(x_i):=m''(x_i^1)$. Then~$m'$ satisfies~$\varphi$, because
  the variables~$y_i$ can be assigned values as in $m''$. Moreover,
  whenever $m(x_i)$ differs from $m'(x_i)$, the inequality
  $B(m)(x_i^j) \neq m''(x_i^j)$ holds for every~$j$. Thus we obtain
  $\lambda \hd(m, m') \leq \hd(B(m), m'') \leq
  k\lambda+\ell$.  Therefore, we have the inequality
  $\hd(m, m') \leq k + \ell /\lambda$ and hence
  $\hd(m, m')\leq k$, since $\ell /\lambda<1$. This completes the many-one reduction.

  To see that the construction above is also an $\AP$-reduction,
  let~$m''$ be an $r$-approximation for~$\varphi'$
  and $B(m)$, i.e., $\hd(B(m), m'')\leq r \cdot \OPT(\varphi',
  B(m))$. Construct~$m'$ as before, so
  $\lambda \hd(m, m') \leq \hd(B(m), m'') \leq r \cdot
  \OPT(\varphi', B(m))$. Since $\OPT(\varphi', B(m))$ is at most
  $\lambda \OPT(\varphi, m) + \ell$ as before, we get
  $\lambda \hd(m, m') \leq r ( \lambda \OPT(\varphi, m) +
  \ell)$. This implies the inequality
  $\hd(m,m') \leq r \cdot \OPT(\varphi, m) + r\cdot \ell /
  \lambda = (r + \Lo(1))\cdot \OPT(\varphi, m)$ and shows that
  the construction is an $\AP$-reduction with $\alpha = 1$.
\end{proof}

\begin{remark}\label{rem:zero}
  Note that in the reduction from $\dNSOLpp(\Gamma)$ to
  $\dNSOL(\Gamma)$ we construct the assignment $B(m)$ as an extension
  of~$m$ by setting all new variables to~$0$. In particular, if~$m$ is
  the constant $0$-assignment, then so is $B(m)$. We use this
  observation as we continue.
\end{remark}

The following result is a technical lemma, which allows us to consider
constraints with disjoint variables independently.

\begin{lemma}\label{lem:split}
  Let $\varphi(\vec x, \vec y) = \psi(\vec x) \land \chi(\vec y)$ be a
  $\Gamma$-formula over a constraint language~$\Gamma$ and~$m$ an
  assignment over disjoint variable blocks $\vec x$ and $\vec y$.  Let
  $(\varphi, m)$ be an instance of\/ $\NSOL(\Gamma)$. Then
  $\OPT(\varphi, m) = \OPT(\psi, m\Restriction_{\vec x}) + \OPT(\chi,
  m\Restriction_{\vec y})$.
\end{lemma}
\begin{proof}
  If $s\in[\varphi]$, then $s\Restriction_{\vec{x}} \in [\psi]$ and
  $s\Restriction_{\vec{y}}\in[\chi]$. Conversely, if $s_\psi\in[\psi]$ and
  $s_\chi\in[\chi]$, then
  $s:= s_\psi\cup s_\chi$ is a model of $\varphi$.
  If $s\in[\varphi]$ is optimal, i.e.\ $\hd(s,m)=\OPT(\varphi,m)$, then
  \begin{displaymath}
  \OPT(\varphi,m)=\hd(s,m)
  = \hd(s\Restriction_{\vec{x}},m\Restriction_{\vec{x}})
   +\hd(s\Restriction_{\vec{y}},m\Restriction_{\vec{y}})
  \geq \OPT(\psi,m\Restriction_{\vec{x}})
      +\OPT(\chi,m\Restriction_{\vec{y}}).
  \end{displaymath}
  Conversely, if $s_\psi\in[\psi]$ and $s_\chi\in[\chi]$ are optimal
  solutions for their respective problems, then $s := s_\psi\cup s_\chi$
  satisfies
  \begin{displaymath}
  \OPT(\varphi,m)\leq\hd(s,m)
  = \hd(s\Restriction_{\vec{x}},m\Restriction_{\vec{x}})
   +\hd(s\Restriction_{\vec{y}},m\Restriction_{\vec{y}})
  = \OPT(\psi,m\Restriction_{\vec{x}})
   +\OPT(\chi,m\Restriction_{\vec{y}}).\qedhere
  \end{displaymath}
\end{proof}

We can also show that introducing explicit equality constraints does not
change the complexity of our problem.
We need two introductory lemmas. The first one deals with
equalities that do not interfere with the other atoms of the given
formula.
\begin{lemma}\label{lem:pre-processing-equality}
 For constraint languages $\Gamma$, the problems
 $\NSOL(\Gamma\cup\set{\eq})$ and $\dNSOL(\Gamma\cup\set{\eq})$
 reduce to particular cases of the respective problem,
 where for each
 constraint $x\eq y$ in the given formula~$\varphi$ at least one of $x,y$
 occurs also in some $\Gamma$-atom of~$\varphi$.
\end{lemma}
\begin{proof}
  Let $(\varphi,m)$ be an instance of $\NSOL(\Gamma\cup\set{\eq})$.
  Without loss of generality we assume $\varphi$ to be of the form
  $\psi\land\varepsilon$, where $\psi$ is a $\Gamma$-formula and
  $\varepsilon$ is a $\set\eq$-formula. Let $(V_i)_{i\in I}$ be the
  unique finest partition of the variables in~$\varepsilon$
  satisfying that variables $x,y$ are in the same partition class if
  $x\eq y$ occurs in~$\varepsilon$.

  For each index $i\in I$ we designate a specific
  variable~$x_i\in V_i$.  Let $\psi'$ be the formula obtained
  from~$\psi$ by substituting all occurrences of variables $y\in V_i$
  by $x_i$. Moreover, let $I'$ be the set of indices $i\in I$ such
  that $x_i$ actually occurs in~$\psi'$, and let
  $I'':=I\smallsetminus I'$ be the set of indices without this
  property. We set $\varepsilon':=\bigwedge_{i\in I'}\varepsilon_i$
  and $\varepsilon'':=\bigwedge_{i\in I''}\varepsilon_i$, where the
  formula $\varepsilon_i:=\bigwedge_{y\in V_i}(x_i \eq y)$ expresses the
  equivalence of the variables in~$V_i$.  Note that the formulas
  $\psi\land\varepsilon$ and
  $\chi:= \psi'\land\varepsilon'\land\varepsilon''$ contain the same
  variables and have identical sets of models.

  Now consider the formula $\varphi':=\psi'\land\varepsilon'$ and the
  assignment $m' := m\Restriction_{V'}$, where $V'$ is the set of
  variables occurring in~$\varphi'$. The pair $(\varphi', m')$ is an
  $\NSOL(\Gamma\cup\set{\eq})$ instance with the additional properties
  stated in the lemma.  By construction we have
  $\chi = \varphi' \land \varepsilon''$, where the set~$V'$ of
  variables in $\varphi'$ and the set~$V''$ of variables
  in~$\varepsilon''$ are disjoint.  By Lemma~\ref{lem:split} we obtain
  $\OPT(\varphi,m) = \OPT(\chi,m) =\OPT(\varphi',m') +
  \OPT(\varepsilon'',m\Restriction_{V''})$.

  %

  An optimal solution $s_{\varepsilon''}$ of~$\varepsilon''$ and the
  optimal value $d:=\OPT(\varepsilon'',m\Restriction_{V''})$ can
  obviously be computed in polynomial time. Therefore the instance
  $(\varphi,m,k)$ of $\dNSOL(\Gamma\cup\set{\eq})$ corresponds
  to the instance $(\varphi',m',k-d)$ of the restricted
  decision problem in the polynomial-time many-one reduction.

Moreover, if~$s'$ is an $r$-approximate solution of $(\varphi',m')$ for
some $r\geq 1$, then $s:=s'\cup s_{\varepsilon''}$ is a
solution of $\varphi$, and we have
\begin{equation*}
\hd(s,m) = \hd(s',m') + d \leq r\OPT(\varphi',m') + d \leq
r\OPT(\varphi',m') + rd = r\OPT(\varphi,m),
\end{equation*}
so the constructed solution~$s$ of $\varphi$ is also $r$-approximate.
This concludes the proof of the $\AP$-reduction with factor $\alpha=1$.
\end{proof}

When dealing with $\NSOL(\Gamma\cup\set{\eq})$, the previous lemma enables
us to concentrate on instances where the formula~$\varphi$ has the form
$\psi(z_1,\dotsc,z_n,x_1,\dotsc,x_t)
       \land \bigwedge_{i=1}^t\bigwedge_{x\in V_i} (x_i \eq x)$,
where $V_1,\dotsc,V_t$ are disjoint sets of variables, being also
disjoint from the variables of the $\Gamma$-formula $\psi$. For each
$1\leq i\leq t$ the given assignment $m$ can have equal distance to
the zero vector and the all-ones vector on the variables in
$V_i\cup\set{x_i}$, or it can be closer to one of the constant vectors.
It is convenient to group the equality constraints according to these
three cases. The following lemma discusses how to remove those equality
constraints, on whose variables $m$ is not equidistant from~$\vec{0}$
and~$\vec{1}$.

\begin{lemma}\label{lem:Heindl}%
  Let~$\Gamma$ be a constraint language
  and~$\psi(z_1,\dotsc,z_{n},x_1,\dotsc,x_\alpha,v_1,\dotsc,v_\beta,w_1,\dotsc,w_\gamma)$ be any $\Gamma$-formula containing
  precisely the distinct variables $z_1,\dotsc,z_{n}$,
  $x_1,\dotsc,x_\alpha$, $v_1,\dotsc,v_\beta$ and
  $w_1,\dotsc,w_\gamma$.
  Consider a formula
  \[\varphi := \psi \land
            \bigwedge_{a=1}^\alpha\bigwedge_{x\in I_a'} (x_a \eq x) \land
            \bigwedge_{b=1}^\beta \bigwedge_{x\in J_b'} (v_b \eq x) \land
            \bigwedge_{c=1}^\gamma\bigwedge_{x\in K_c'} (w_c \eq x)
  \]
  where $I_1',\dotsc,I_\alpha'$, $J_1',\dotsc,J_\beta'$ and
  $K_1',\dotsc,K_\gamma'$ are non-empty sets of variables that are
  pairwise disjoint and disjoint from the variables in~$\psi$.
  For $1\leq a\leq \alpha$, $1\leq b\leq \beta$ and $1\leq c\leq\gamma$
  we put $I_a:= I_a' \cup\set{x_a}$, $J_b:= J_b' \cup\set{v_b}$
  and $K_c := K_c'\cup\set{w_c}$.
  Moreover, let $m$ be an assignment for $\varphi$, such that
  for $1\leq a\leq \alpha$, $1\leq b\leq \beta$ and $1\leq c\leq\gamma$
  \begin{align*}
  d_{0,I_a}&:= \hd(m\Restriction_{I_a},\vec{0}),&
  d_{1,I_a}&:= \hd(m\Restriction_{I_a},\vec{1}),&
  &&& d_{1,I_a} - d_{0,I_a} &&= 0\\
  d_{0,J_b}&:= \hd(m\Restriction_{J_b},\vec{0}),&
  d_{1,J_b}&:= \hd(m\Restriction_{J_b},\vec{1}),&
  &\text{satisfy}&
  e_b :={}& d_{1,J_b} - d_{0,J_b} &&> 0\\
  d_{0,K_c}&:= \hd(m\Restriction_{K_c},\vec{0}),&
  d_{1,K_c}&:= \hd(m\Restriction_{K_c},\vec{1}),&
  &&f_c :={}& d_{0,K_c} - d_{1,K_c} &&> 0
  \end{align*}
  It is possible to construct a formula~$\psi'$, whose size is polynomial
  in the size of~$\varphi$, and an assignment~$M$ for
  $\varphi':= \psi'\land\bigwedge_{a=1}^\alpha\bigwedge_{x\in I_a'} (x_a\eq x)$
  such that the following holds
  \begin{compactitem}
  \item $\psi$, $\varphi$, $\varphi'$ and $\psi'$ are equisatisfiable;
  \item if $\psi$ is satisfiable, then
        $\OPT(\varphi,m) = \OPT(\varphi',M) + d$
        where $d = \sum_{b=1}^\beta d_{0,J_b}
                  +\sum_{c=1}^\gamma d_{1,K_c}$;
  \item for every $r\in[1,\infty)$, one can produce an ($r$-approximate)
        solution of $(\varphi,m)$ from any
        ($r$-\hskip0pt{}approximate) solution of
        $(\varphi',M)$ in polynomial time.
  \end{compactitem}
\end{lemma}
\begin{proof}
  First, we describe how to construct the formula~$\psi'$. We abbreviate
  $Z:=\set{z_1,\dotsc,z_n,x_1,\dotsc,x_\alpha}$,
  $Z':=Z\cup \bigcup_{a=1}^\alpha I_a$,
  $V:=\set{v_1,\dotsc,v_\beta}$ and $W:=\set{w_1,\dotsc,w_\gamma}$.
  For every variable
  $u\in Z\cup V\cup W$ define a set $B(u)$ of variables as follows:
  \[
  B(u) = \begin{cases}
  \set{u} & \text{if } u\in Z\\
  \set{u^1,\dotsc,u^{e_b}} & \text{if } u=v_b\in V\\
  \set{u^1,\dotsc,u^{f_c}} & \text{if } u=w_c\in W.
  \end{cases}
  \]
  For each atom~$R(u_1,\dotsc,u_q)$ of~$\psi$ define a set of atoms
  $\Set{R(u_1',\dotsc,u_q')}{(u_1',\dotsc,u_q')\in\prod_{i=1}^q B(u_i)}$,
  take the union over all these sets and define $\psi'$ as the
  conjunction of all its members, giving a formula over
  $Z\cup V'\cup W'$ where $V' =\bigcup_{u\in V} B(u)$ and
  $W' = \bigcup_{u\in W} B(u)$. Adding again the equality
  constraints, where~$m$ has equal distance from~$\vec{0}$ and~$\vec{1}$
  we get
  $\varphi' = \psi'\land\bigwedge_{a=1}^\alpha\bigwedge_{x\in I_a'} (x_a\eq x)$
  over~$Z'\cup V'\cup W'$. This is a polynomial time construction since
  the arities of relations in~$\Gamma$ are bounded.

  Moreover, we define an assignment~$M$ to the variables~$u$ of~$\varphi'$
  as follows:
  \begin{align*}
  M(u)=\begin{cases}
  m(u)&\text{if } u\in Z'\\
  0   &\text{if } u\in V'\\
  1   &\text{if } u\in W'.
  \end{cases}
  \end{align*}
  Let~$S'$ be a solution of~$(\varphi',M)$. If $S'$ is constant
  on~$B(u)$, for each $u\in V\cup W$, then put $S'':= S'$. Otherwise, by letting
  $S''(u):=S'(u)$ for $u\in Z'$ and for $u\in B(u')$ where $u'\in V\cup W$
  is such that $S'$ is constant on $B(u')$, and by defining
  $S''(u):= M(u) = 0$ for the remaining variables $u\in V'$ and
  $S''(u):= M(u) = 1$ for the remaining variables $u\in W'$,
  we obtain a model~$S''$ of~$\varphi'$ satisfying
  $\hd(S'',M)\leq \hd(S',M)$ and being constant on $B(u)$ for each
  $u\in V\cup W$.
  From~$S''$ we construct an assignment~$S$ of~$\varphi$ by defining
  $S(u):=S''(u)$ for~$u\in Z'$,
  $S(u):=S''(v_b^1)$ for $u\in J_b$ and $1\leq b\leq \beta$, and
  $S(u):=S''(w_c^1)$ for $u\in K_c$ and $1\leq c\leq \gamma$.
  It satisfies~$\varphi$ as $e_b,f_c>0$ for $1\leq b\leq \beta$
  and $1\leq c\leq \gamma$.
  From these definitions, it follows
  \begin{alignat*}{5}
    \hd(S'',M) &= \hd(S''\Restriction_{Z'},M\Restriction_{Z'})
              &&+ \sum_{b=1}^\beta
                  \hd(S''\Restriction_{B(v_b)},M\Restriction_{B(v_b)})
              &&+ \sum_{c=1}^\gamma
                  \hd(S''\Restriction_{B(w_c)},M\Restriction_{B(w_c)})\\
              &= \hd(S''\Restriction_{Z'},m\Restriction_{Z'})
              &&+ \sum_{b=1}^\beta S''(v_b^1)\cdot e_b
              &&+ \sum_{c=1}^\gamma (1-S''(w_c^1))\cdot f_c,\\
  \intertext{because $S''$ is constant on $B(u)$ for $u\in V\cup W$ and
             $\lvert B(v_b)\rvert = e_b$, $\lvert B(w_c)\rvert = f_c$
             for $1\leq b\leq \beta$ and $1\leq c\leq \gamma$; and}
    \hd(S,m)  &= \hd(S\Restriction_{Z'},m\Restriction_{Z'})
              &&+ \sum_{b=1}^\beta
                  \hd(S\Restriction_{J_b},m\Restriction_{J_b})
              &&+ \sum_{c=1}^\gamma
                  \hd(S\Restriction_{K_c},m\Restriction_{K_c})\\
              &= \hd(S''\Restriction_{Z'},m\Restriction_{Z'})
              &&+ \sum_{b=1}^\beta
                  \left(S''(v_b^1)\cdot e_b + d_{0,J_b}\right)
              &&+ \sum_{c=1}^\gamma
                  \left((1-S''(w_c^1))\cdot f_c + d_{1,K_c}\right).
  \end{alignat*}
  Consequently, $\hd(S,m) = \hd(S'',M) + d$,
  where $d=\sum_{b=1}^\beta d_{0,J_b} + \sum_{c=1}^\gamma d_{1,K_c}$.
  \par
  Using this, we shall prove below that
  $\OPT(\varphi',M) + d = \OPT(\varphi,m)$.
  Thus, if~$S'$ now takes the role of an $r$-approximate solution
  of~$(\varphi',M)$ for some $r\geq 1$, then it follows that
  \begin{align*}
  \hd(S,m) =    \hd(S'',M) + d
           \leq \hd(S',M)  + d
           &\leq r\OPT(\varphi',M) + d\\
           &\leq r\OPT(\varphi',M) + rd
           =    r\OPT(\varphi,m).
  \end{align*}
  \par

  Let subsequently $S'$ be such that $\OPT(\varphi',M) = \hd(S',M)$,
  and let $s$ be a model of~$\varphi$. Construct a model~$s'$
  of~$\varphi'$ by putting $s'(u):=s(u)$ for~$u\in Z'$ and
  $s'(u):=s(u')$ for $u\in B(u')$ and $u'\in V\cup W$.
  As above we get
  $\hd(s,m) = \hd(s',M) + d$ because the definitions imply
  \begin{alignat*}{5}
    \hd(s',M) &= \hd(s'\Restriction_{Z'},M\Restriction_{Z'})
             &&+ \sum_{b=1}^\beta
                 \hd(s'\Restriction_{B(v_b)},M\Restriction_{B(v_b)})
             &&+ \sum_{c=1}^\gamma
                 \hd(s'\Restriction_{B(w_c)},M\Restriction_{B(w_c)})\\
              &= \hd(s\Restriction_{Z'},m\Restriction_{Z'})
             &&+ \sum_{b=1}^\beta s(v_b)\cdot e_b
             &&+ \sum_{c=1}^\gamma (1-s(w_c))\cdot f_c\enspace;\\
    \hd(s,m) &= \hd(s\Restriction_{Z'},m\Restriction_{Z'})
             &&+ \sum_{b=1}^\beta
                 \hd(s\Restriction_{J_b},m\Restriction_{J_b})
             &&+ \sum_{c=1}^\gamma
                 \hd(s\Restriction_{K_c},m\Restriction_{K_c})\\
             &= \hd(s\Restriction_{Z'},m\Restriction_{Z'})
             &&+ \sum_{b=1}^\beta
                 \left(s(v_b)\cdot e_b + d_{0,J_b}\right)
             &&+ \sum_{c=1}^\gamma
                 \left((1-s(w_c))\cdot f_c + d_{1,K_c}\right).
  \end{alignat*}
  By minimality of $S'$, we obtain
  $\hd(S'',M)\leq \hd(S',M)\leq \hd(s',M)$. If we additionally require
  that~$s$ be an optimal solution of~$(\varphi,m)$, then
  $\hd(s',M) = \hd(s,m) -d \leq \hd(S,m) -d = \hd(S'',M)$.
  Thus, the distances $\hd(S'',M)$, $\hd(S',M)$ and $\hd(s',M)$ coincide,
  which implies the desired equality
  $\OPT(\varphi,m)=\hd(s,m) = \hd(s',M)+d = \hd(S',M)+d=\OPT(\varphi',M)+d$.
\end{proof}

The previous lemma, in fact, describes an $\AP$-reduction from the
specialized version of the problem $\NSOL(\Gamma\cup\set{\eq})$ discussed
in Lemma~\ref{lem:pre-processing-equality} to an even more specialized
variant (the analogous statement is true for the decision
version---instances $(\varphi,m,k)$ can be decided by considering
$(\varphi',M,k-d)$ instead):
namely all equality constraints touch variables in
$\Gamma$-atoms and the given assignment has equal distance from the
constant tuples on each variable block connected by equalities. In the
next result we show how to remove also these equality constraints.

\begin{proposition}\label{prop:equality}%
  For constraint languages~$\Gamma$ we have
  $\dNSOL(\Gamma)\meq\dNSOL(\Gamma\cup\set{\eq})$ and
  $\NSOL(\Gamma)\apeq\NSOL(\Gamma\cup\set{\eq})$.
\end{proposition}
\begin{proof}
  The reduction from left to right is trivial.
  For the other direction, consider first an instance of
  $\dNSOL(\Gamma\cup\set{\eq})$ with formula~$\varphi$, assignment~$m$,
  and bound~$k$.
  Applying the reductions indicated in
  Lemmas~\ref{lem:pre-processing-equality} and~\ref{lem:Heindl}, we can
  assume (also for $\NSOL(\Gamma\cup\set{\eq})$) that $\varphi$ is of the form
  $\psi \land \bigwedge_{a=1}^\alpha \bigwedge_{x\in I'_{a}} (x_a \eq x)$
  with a $\Gamma$-formula~$\psi$ containing the distinct variables
  $z_1,\dotsc,z_n,x_1,\dotsc,x_\alpha$ ($n\geq 0$, $\alpha\geq 1$) and non-empty disjoint
  (from each other and from $\psi$) variable sets $I'_a$ for
  $1\leq a\leq \alpha$. Moreover, we can suppose that
  $\hd(m\Restriction_{I_a},\vec{0}) = \hd(m,\Restriction_{I_a},\vec{1})
  =:c_a$ for all $1\leq a\leq \alpha$, where $I_a$ denotes the set
  $I'_a\cup\set{x_a}$.

  We define $c := \sum_{a=1}^\alpha c_a$, and we
  choose some $\ell$-element index set $I$ such that $\alpha/\ell <1$,
  that is, $\ell\geq \alpha+1$ (we shall place another condition on $\ell$
  at the end). We construct a formula $\varphi'$ as follows:
  For each atom $R(u_1,\dotsc,u_q)$ of~$\psi$ we introduce the set
  $\Set{R(u_1^{i_1},\dotsc,u_q^{i_q})}{(i_1,\dotsc,i_q)\in I^q}$ of atoms
  where for $1\leq \nu\le q$ and $i\in I$ we let
  $u_{\nu}^{i} := z_{j,i}$ if $u_{\nu} = z_j$
  for some $1\leq j\leq n$ and $u_{\nu}^{i} = u_{\nu}$ if else
  $u_{\nu}\in \set{x_1,\dotsc,x_\alpha}$.
  Take the union over all these sets and let~$\varphi'$ be the
  conjunction of all atoms in this union.
  This construction can be carried out in polynomial time since there is
  a bound on the arities of relations in~$\Gamma$.
  Define an assignment~$M$ by $M(x_a):= m(x_a)$ for $1\leq a\leq \alpha$
  and $M(z_{j,i}):=m(z_j)$ for $1\leq j\leq n$ and $i\in I$. We claim that existence of solutions
  for $(\varphi,m,k)$ can be decided by checking for solutions of
  $(\varphi',M,\ell(k-c)+\alpha)$.
  The argument is similar to that of
  Proposition~\ref{prop:quantifiers}: $\psi$ is (un)satisfiable if and only if
  $\varphi$ and $\varphi'$ are, so we have a correct answer in the
  unsatisfiable case. Otherwise, consider a solution $s$ to
  $(\varphi,m,k)$. Letting $Z:=\set{z_1,\dotsc,z_n}$, we have
  \begin{equation*}
  \hd(s,m) = \hd(s\Restriction_Z,m\Restriction_Z) +
  \sum_{a=1}^\alpha \hd(s\Restriction_{I_a},m\Restriction_{I_a})
   =\hd(s\Restriction_Z,m\Restriction_Z) + \sum_{a=1}^\alpha c_a
   =\hd(s\Restriction_Z,m\Restriction_Z) + c,
  \end{equation*}
  i.e.\ $\hd(s\Restriction_Z,m\Restriction_Z)\leq k-c$.
  Putting $s'(x_a):=s(x_a)$ for $1\leq a\leq \alpha$ and
  $s'(z_{j,i}):=s(z_j)$ for $1\leq j\leq n$ and $i\in I$ we get a model
  of~$\varphi'$, and it follows that
  $\hd(s'\Restriction_{Z'},M\Restriction_{Z'})
  =\ell\cdot\hd(s\Restriction_{Z},m\Restriction_{Z})\leq \ell\cdot(k-c)$,
  where $Z':=\set{z_{j,i}\mid 1\leq j\leq n, i\in I}$.
  Therefore, abbreviating $X := \set{x_1,\dotsc,x_\alpha}$, we obtain
  $\hd(s',M) = \hd(s'\Restriction_{Z'},M\Restriction_{Z'}) +
  \hd(s'\Restriction_X,M\Restriction_X) \leq \ell\cdot(k-c)+\alpha$.

  Conversely, let $S'$ be a solution of $(\varphi',M,\ell(k-c)+\alpha)$. As in
  Proposition~\ref{prop:quantifiers} we can construct a solution $S''$ being
  constant on $\set{z_{j,i}\mid i\in I}$ for each $1\leq j\leq n$. Letting
  $S(x):=S''(x_a)$ for $x\in I_a$ and $1\leq a\leq \alpha$ and
  $S(z_j):= S''(z_{j,i})$ for some fixed index $i\in I$ and all $1\leq j\leq n$, one obtains a model
  of~$\varphi$. If $S(z_j)\neq m(z_j)$ for some $1\leq j\leq n$, then we
  have $S''(z_{j,i}) = S(z_j) \neq m(z_j) = M(z_{j,i})$ for all $i\in I$. Hence,
  we have
  $
  \ell\cdot\hd(S\Restriction_{Z},m\Restriction_{Z})
  \leq\hd(S''\Restriction_{Z'},M\Restriction_{Z'})
  \leq \hd(S'',M) \leq \hd(S',M)
  $.
  Division by~$\ell$ implies
  $\hd(S\Restriction_{Z},m\Restriction_{Z})\leq \hd(S',M)/\ell
  \leq k-c + \alpha/\ell < k-c+1$,
  i.e.\ $\hd(S\Restriction_{Z},m\Restriction_{Z})\leq k-c$.
  From this we finally infer that
  $\hd(S,m) 
            = \hd(S\Restriction_{Z},m\Restriction_{Z}) + c \leq k$.%

  Suppose now that~$S'$ is an $r$-approximate solution for $(\varphi',M)$
  for some $r\geq 1$, i.e.\ we have $\hd(S',M)\leq r\OPT(\varphi',M)$.
  Constructing a model $S$ of~$\varphi$ as before, we obtain
  $\ell\hd(S\Restriction_{Z},m\Restriction_{Z})\leq \hd(S',M)\leq r\OPT(\varphi',M)$. Further
  from an optimal solution of~$\varphi$, we get a model~$s'$
  of~$\varphi'$ satisfying
  \begin{align*}
  \OPT(\varphi',M)\leq \hd(s',M)
  &=   \hd(s'\Restriction_{Z'},M\Restriction_{Z'})
     + \hd(s'\Restriction_X,M\Restriction_X)\\
  &=   \ell(\OPT(\varphi,m) - c) + \hd(s'\Restriction_X,M\Restriction_X)
  \leq \ell(\OPT(\varphi,m)-c) + \alpha.
  \end{align*}
  Multiplying this inequality by~$r$, combining it with previous inequalities and dividing by~$\ell$ we
  thus have $\hd(S\Restriction_{Z},m\Restriction_{Z})\leq
  r\OPT(\varphi,m)-rc +r\alpha/\ell$.
  Note that $\OPT(\varphi,m)>0$, because if $\OPT(\varphi,m) =0$, then we
  would have a unique optimal model of~$\varphi$, namely~$m$. Then
  $m\Restriction_{I_1}$ would have to be constant, implying
  $\hd(m\Restriction_{I_1},\vec{0}) \neq \hd(m\Restriction_{I_1},\vec{1})$,
  as one distance would be zero and the other one $\lvert I_1\rvert >0$.
  Therefore, for $\ell\in \LOmega(\lvert\varphi\rvert^2)$ we have
  $\hd(S,m) = \hd(S\Restriction_{Z},m\Restriction_{Z}) + c \leq \hd(S\Restriction_{Z},m\Restriction_{Z}) + rc \leq
  r\OPT(\varphi,m)+r\alpha/\ell
  \leq \OPT(\varphi,m)(r + r\alpha/\ell)
  = \OPT(\varphi,m)(r + \Lo(1))$.
  This demonstrates an $\AP$-reduction with factor~$1$.
\end{proof}

Propositions~\ref{prop:quantifiers} and~\ref{prop:equality}
allow us to switch freely between formulas
with quantifiers and equality and those without. Hence we may
derive upper bounds in the setting without quantifiers and equality
while using the latter in hardness reductions. In particular,
we can use pp-definability when implementing a constraint
language~$\Gamma$ by another constraint language~$\Gamma'$. Hence it
suffices to consider Post's lattice of co-clones to characterize the
complexity of $\NSOL(\Gamma)$ for every finite constraint
language~$\Gamma$.

\begin{corollary}\label{cor:coclones}
 For constraint languages~$\Gamma$ and $\Gamma'$, for which the inclusion
 $\Gamma'\subseteq \cc \Gamma$ holds, we have the reductions
 $\dNSOL(\Gamma')\mle\dNSOL(\Gamma)$ and
 $\NSOL(\Gamma')\aple\NSOL(\Gamma)$. Thus, if
 $\cc{\Gamma'} = \cc\Gamma$ is satisfied, then the equivalences
 $\dNSOL(\Gamma)\meq\dNSOL(\Gamma')$ and
 $\NSOL(\Gamma)\apeq\NSOL(\Gamma')$ hold.
\end{corollary}

Next we prove that, in certain cases, unit clauses in the formula do not
change the complexity of $\NSOL$.

\begin{proposition}\label{prop:unary}
  Let $\Gamma$ be a constraint language such that feasible solutions
  of\/ $\NSOL(\Gamma)$ can be found in polynomial time. Then we have
  $\NSOL(\Gamma) \apeq \NSOL(\Gamma\cup \set{[x], [\neg x]})$.
\end{proposition}
\begin{proof}
  The direction from left to right is obvious.  For the other
  direction, we give an $\AP$-reduction from
  $\NSOL(\Gamma\cup \set{[x], [\neg x]})$ to
  $\NSOL(\Gamma\cup \set{\eq})$. The latter is $\AP$-equivalent
  to $\NSOL(\Gamma)$ by Proposition~\ref{prop:equality}.

  The idea of the construction is to introduce two sets of variables
  $y_1, \ldots, y_{n^2}$ and $z_1, \ldots, z_{n^2}$ such that in any
  feasible solution all $y_i$ and all $z_i$ take the same value. By
  setting $m(y_i)=1$ and $m(z_i)=0$ for each $i$, any feasible
  solution $m'$ of small Hamming distance to $m$ will have $m'(y_i)=1$
  and $m'(z_i)=0$ for all $i$ as well, because deviating from this
  would be prohibitively expensive. Finally, we simulate the unary
  relations $x$ and $\neg x$ by $x\eq y_1$ and $x\eq z_1$,
  respectively. We now describe the reduction formally.

  Consider a formula~$\varphi$ over $\Gamma \cup \set{[x], [\neg x]}$
  with the variables $x_1, \ldots, x_n$ and an assignment~$m$.
  If~$(\varphi,m)$ fails to have feasible solutions, i.e.,
  if~$\varphi$ is unsatisfiable, we can detect this in polynomial time
  by the assumption of the lemma and return the generic unsatisfiable
  instance~$\bot$. Otherwise, we construct a
  $(\Gamma \cup \set{\eq})$-formula~$\varphi'$ over the
  variables $x_1, \ldots x_n$, $y_1, \ldots, y_{n^2}$,
  $z_1, \ldots, z_{n^2}$ and an assignment~$m'$. We obtain~$\varphi'$
  from~$\varphi$ by replacing every occurrence of a
  constraint~$[x]$ by $x \eq y_1$ and
  every occurrence of~$[\neg x]$ by
  $x \eq z_1$. Finally, we add the atoms $y_i \eq y_1$
  and $z_i \eq z_1$
  for all $i \in \set{2, \ldots, n^2}$. Let~$m'$ be
  the assignment of the variables of~$\varphi'$ given by
  $m'(x_i) = m(x_i)$ for each $i \in \set{1, \ldots, n}$, and
  $m'(y_i)=1$
  and $m'(z_i)=0$
  for all $i \in \set{1, \ldots, n^2}$. To any feasible solution~$m''$
  of~$\varphi'$ we assign $g(\varphi, m, m'')$ as follows.
  \begin{compactenum}
  \item \label{case1}
    If $\varphi$ is satisfied by~$m$, we define $g(\varphi, m,
    m'')$ to be equal to~$m$.
  \item \label{case2}
    Else if $m''(y_i)=0$ holds for all $i \in \set{1, \ldots, n^2}$
    or $m''(z_i)=1$ for all $i \in \set{1, \ldots, n^2}$, we define
    $g(\varphi, m, m'')$ to be any satisfying assignment of~$\varphi$.
  \item \label{case3}
    Otherwise, we have $m''(y_i)=1$ and $m''(z_i)=0$ for all
    $i \in \set{1, \ldots, n^2}$. In this case we define
    $g(\varphi, m, m'')$ to be the restriction of~$m''$ onto
    $x_1, \ldots, x_n$.
  \end{compactenum}
  Observe that all variables~$y_i$ and all~$z_i$ are forced to take
  the same value in any feasible solution, respectively, so
  $g(\varphi, m, m'')$ is always well-defined.  The construction is an
  $\AP$-reduction. Assume that~$m''$ is an $r$-approximate
  solution. We will show that $g(\varphi, m, m'')$ is also an
  $r$-approximate solution.

  \paragraph{Case~\ref{case1}:} $g(\varphi, m, m'')$ computes the
  optimal solution, so there is nothing to show.

  \paragraph{Case~\ref{case2}:}
  Observe first that~$\varphi$ has a solution because otherwise it
  would have been mapped to~$\bot$ and~$m''$ would not exist.
  Thus, $g(\varphi, m, m'')$ is well-defined and feasible by
  construction. Observe that~$m'$ and~$m''$ disagree on all~$y_i$ or
  on all~$z_i$, so we have $\hd(m', m'')\ge n^2$. Moreover,
  since~$\varphi$ has a feasible solution, it follows that
  $\OPT(\varphi', m') \leq n$. Since~$m''$ is an $r$-approximate
  solution, we have that
  $n\OPT(\varphi',m') \leq n^2\leq \hd(m',m'')\leq r\OPT(\varphi',m')$.
  If $\OPT(\varphi',m')=0$, then $m'$ would have to be a model of
  $\varphi'$, and so would be its restriction to the~$x_i$, i.e.~$m$, a
  model of~$\varphi$. This is handled in the first case, which is
  disjoint from the current one; hence, we infer $n\leq r$. Consequently,
  the distance $\hd(m, g(\varphi, m, m''))$ is bounded above by
  $n \leq r \leq r \cdot \OPT(\varphi,m)$, where the last inequality
  holds because~$\varphi$ is not satisfied by~$m$ and thus the
  distance of any optimal solution from~$m$ is at least~$1$.

  \paragraph{Case~\ref{case3}:}
  The variables~$x_i$, for which the relation~$[x_i]$ is a constraint,
  all satisfy $g(\varphi,m,m'')(x_i)=1$ by construction. Moreover, we
  have $g(\varphi, m, m'')(x_i)=0$ for all~$x_i$ for which
  $[\neg x_i]$ is a constraint of~$\varphi$. Consequently,
  $g(\varphi, m, m'')$ is feasible. Again,
  $\OPT(\varphi', m') \leq n$, so any optimal solution to
  $(\varphi', m')$ must set all variables~$y_i$ to~$1$ and all~$z_i$
  to~$0$. It follows that $\OPT(\varphi,m) = \OPT(\varphi', m')$. Thus
  we get
  \begin{displaymath}
    \hd(m, g(\varphi, m, m'')) = \hd(m', m'') \leq r \cdot
    \OPT(\varphi',m') = r \cdot \OPT(\varphi,m),
  \end{displaymath}
  which completes the proof.
\end{proof}

\subsection{Inapplicability of Clone Closure}
\label{sec:inap}

Corollary~\ref{cor:coclones} shows that the complexity of~$\NSOL$ is
not affected by existential quantification by giving an explicit
reduction from~$\NSOLpp$ to~$\NSOL$. It does not seem possible to
prove the same for $\XSOL$ and $\MSD$.  However, similar results hold
for the conjunctive closure; thus we resort to minimal or irredundant
weak bases of co-clones instead of usual bases.

\begin{proposition}\label{prop:coclonesmincd}\label{prop:coclonesminhd}
  Let~$\Gamma$ and~$\Gamma'$ be constraint languages.  If
  $\Gamma'\subseteq \ccc \Gamma$ holds then we have the reductions
  $\dXSOL(\Gamma')\mle\dXSOL(\Gamma)$ and
  $\XSOL(\Gamma')\aple\XSOL(\Gamma)$, as well as
  $\dMSD(\Gamma')\mle\dMSD(\Gamma)$ and
  $\MSD(\Gamma')\aple\MSD(\Gamma)$.
\end{proposition}
\begin{proof}
  We prove only the part that $\Gamma'\subseteq \ccc \Gamma$ implies
  $\XSOL(\Gamma')\aple\XSOL(\Gamma)$. The other results will be clear
  from that reduction since the proof is generic and therefore holds for
  both $\XSOL$ and $\MSD$, as well as for their decision variants.

  Let a $\Gamma'$-formula~$\varphi$ be an instance of $\XSOL(\Gamma')$.
  Since $\Gamma' \subseteq \ccc \Gamma$, every constraint
  $R(x_1, \ldots, x_k)$ of~$\varphi$ can be written as a conjunction
  of constraints over relations from~$\Gamma$. Substitute the latter
  into~$\varphi$, obtaining $\varphi'$. Now $\varphi'$ is an instance
  of $\XSOL(\Gamma)$, where~$\varphi'$ is only polynomially larger
  than~$\varphi$.  As~$\varphi$ and~$\varphi'$ have the same variables
  and hence the same models, also the closest distinct models
  of~$\varphi$ and~$\varphi'$ are the same.
\end{proof}

\subsection{Duality}
\label{sec:duality}

Given a relation $R \subseteq \set{0,1}^n$, its \emph{dual} relation is
$\dual(R) = \Set{\cmpl m}{m \in R}$, i.e., the relation containing the
complements of tuples from~$R$. Duality naturally extends to sets of
relations and co-clones. We define
$\dual(\Gamma) = \Set{\dual(R)}{R \in \Gamma}$ as the set of dual
relations to~$\Gamma$. Since taking complements is involutive, duality
is a symmetric relation. If a relation~$R'$ (a set of
relations~$\Gamma'$) is a dual relation to~$R$ (a set of dual
relations to~$\Gamma$), then~$R$ ($\Gamma$) is also dual to~$R'$
(to~$\Gamma'$).  By a simple inspection of the bases of co-clones in
Table~\ref{tab:clones}, we can easily see that many co-clones are dual
to each other. For instance $\iE_2$ is dual to~$\iV_2$.  The following
proposition shows that it is sufficient to consider only one half of
Post's lattice of co-clones.

\begin{proposition}\label{prop:dual}
  For any constraint language~$\Gamma$ we have
  $\dNSOL(\Gamma)\meq\dNSOL(\dual(\Gamma))$ and
  $\NSOL(\Gamma)\apeq\NSOL(\dual(\Gamma))$,
  $\dXSOL(\Gamma)\meq\dXSOL(\dual(\Gamma))$ and
  $\XSOL(\Gamma)\apeq\XSOL(\dual(\Gamma))$,
  as well as
  $\dMSD(\Gamma)\meq\dMSD(\dual(\Gamma))$ and
  $\MSD(\Gamma)\apeq\MSD(\dual(\Gamma))$.
\end{proposition}
\begin{proof}
  Let~$\varphi$ be a $\Gamma$-formula and~$m$ an assignment
  to~$\varphi$. We construct a $\dual(\Gamma)$-formula $\varphi'$ by
  substitution of every atom $R(\vec x)$ by $\dual(R)(\vec x)$. The
  assignment~$m$ satisfies~$\varphi$ if and only if~$\cmpl m$
  satisfies~$\varphi'$, where~$\cmpl m$ is the pointwise complement
  of~$m$. Moreover, $\hd(m, m') = \hd(\cmpl m, \cmpl m')$.
\end{proof}

\section{Finding the Nearest Solution}
\label{sec:proofsNSOL}

This section contains the proof of Theorem~\ref{thm:NSol}. We first
consider the polynomial-time cases followed by the cases of higher
complexity.

\subsection{Polynomial-Time Cases}
\label{ssec:proofsNSOL-PO}

\begin{proposition}\label{prop:NSOL-iD1}
  If a constraint language~$\Gamma$ is both bijunctive and affine
  $(\Gamma\subseteq\iD_1)$,
  then $\NSOL(\Gamma)$ can be solved in polynomial time.
\end{proposition}
\begin{proof}
  Since~$\Gamma\subseteq \iD_1= \cc{\Gamma'}$ with
  $\Gamma':=\set{[x\oplus y], [x]}$, we have the reduction
  $\NSOL(\Gamma)\aple \NSOL(\Gamma')$ by Corollary~\ref{cor:coclones}.
  Every $\Gamma'$-formula~$\varphi$ is equivalent to a linear system
  of equations over the Boolean ring~$\ZZ_2$ of type
  $x\oplus y = 1$ and $x=1$. Substitute the fixed values $x=1$
  into the equations of the type $x\oplus y = 1$ and
  propagate. If a contradiction is found thereby, reject the input.
  After an exhaustive application of this rule only
  equations of the form $x\oplus y = 1$ remain. For each of them
  put an edge $\set{x,y}$ into~$E$, defining an undirected graph
  $G=(V,E)$, whose vertices $V$ are the unassigned variables.
  If $G$ is not bipartite, then $\varphi$ has no solutions, so we can
  reject the input. Otherwise, compute a bipartition
  $V = L \dot\cup R$.  We assume that~$G$ is connected; if not perform
  the following algorithm for each connected component (cf.\
  Lemma~\ref{lem:split}). Assign the
  value~$0$ to each variable in~$L$ and the value~$1$ to each variable
  in~$R$, giving the satisfying assignment~$m_1$. Swapping the roles
  of $0$ and $1$ w.r.t.\ $L$ and $R$ we get a model~$m_2$.  Return
  $\min\set{\hd(m,m_1),\hd(m,m_2)}$.
\end{proof}

\begin{proposition}\label{prop:NSOL-iM2}
  If a constraint language~$\Gamma$ is monotone
  $(\Gamma\subseteq\iM_2)$, then
  $\NSOL(\Gamma)$ can be solved in polynomial time.
\end{proposition}
\begin{proof}
  We have $\iM_2 = \cc{\Gamma'}$ where
  $\Gamma':=\set{[x\to y], [\neg x], [x]}$. Thus
  Corollary~\ref{cor:coclones} and $\Gamma\subseteq\cc{\Gamma'}$ imply
  $\NSOL(\Gamma)\aple\NSOL(\Gamma')$.  The relations $[\neg x]$
  and~$[x]$ determine a unique value for the respective variable,
  therefore we can eliminate unit clauses and propagate the values. If a
  contradiction occurs, we reject the input.
  It thus remains to consider formulas~$\varphi$ containing only
  binary implicative clauses of type $x \to y$.

  Let~$V$ be the set of variables in~$\varphi$, and for $i\in\set{0,1}$
  let $V_i = \Set{x \in V}{m(x) = i}$ be the variables mapped to
  value~$i$ by assignment~$m$. We transform the
  formula~$\varphi$ to a linear programming problem as
  follows. For each clause $x \to y$ we add the inequality $y \geq x$,
  and for each variable $x \in V$ we add the constraints $x \geq 0$
  and $x \leq 1$.  As linear objective function we use
  $f(\vec x) = \sum_{x \in V_0} x + \sum_{x \in V_1} (1 - x)$.
  For an arbitrary solution~$m'$, it returns the number of variables
  that change their parity between~$m$ and~$m'$, i.e.,
  $f(m')=\hd(m,m')$. This way we obtain the (integer) linear programming
  problem $(f, A \vec x \geq \vec b)$, where~$A$ is a
  totally unimodular matrix and~$\vec b$ is an integral column
  vector.

  The rows of~$A$ consist of the left-hand sides of inequalities
  $y - x \geq 0$, $x \geq 0$, and $-x \geq -1$, which constitute the
  system $A \vec x \geq \vec b$. Every entry in~$A$ is~$0$, $+1$,
  or~$-1$. Every row of~$A$ has at most two non-zero entries. For the
  rows with two entries, one entry is~$+1$, the other
  is~$-1$. According to Condition~(iv) in Theorem~19.3
  in~\cite{Schrijver-86}, this is a sufficient condition for~$A$ being
  totally unimodular.
  As~$A$ is totally unimodular and~$\vec b$ is an integral vector,
  $f$ has integral minimum points, and one of them can be computed in polynomial
  time (see~e.g.~\cite[Chapter~19]{Schrijver-86}).
\end{proof}

\subsection{Hard Cases}
\label{ssec:proofsNSOL-hard}

We start off with an easy corollary of Schaefer's dichotomy.

\begin{proposition}\label{prop:NSOL-iN_2-BR}
  Let $\Gamma$ be a finite set of Boolean relations. If\/
  $\iN_2\subseteq \cc \Gamma$, then it is $\NP$-complete to decide whether
  a feasible solution exists for $\NSOL(\Gamma)$; otherwise,
  $\NSOL(\Gamma)\in \pAPX$.
\end{proposition}
\begin{proof}
  If $\iN_2\subseteq \cc \Gamma$ holds, checking the existence of
  feasible solutions for $\NSOL(\Gamma)$-instances is $\NP$-hard by
  Schaefer's theorem~\cite{Schaefer-78}.

  Let $(\varphi,m)$ be an instance of $\NSOL(\Gamma)$.  We give an
  $n$-approximate algorithm for the other cases, where $n$ denotes the
  number of variables in~$\varphi$.
  If~$m$ satisfies~$\varphi$,
  return~$m$. Otherwise compute an arbitrary solution~$m'$
  of~$\varphi$, which can be done in polynomial time by Schaefer's
  theorem.  This algorithm is $n$-approximate: If~$m$
  satisfies~$\varphi$, the algorithm returns the optimal solution;
  otherwise we have $\OPT(\varphi, m) \geq 1$ and $\hd(m,m') \leq n$,
  hence the answer~$m'$ of the algorithm is $n$-approximate.
\end{proof}

\subsubsection{APX-Complete Cases}

We start with reductions from the optimization version of vertex
cover. Since the relation $[x \lor y]$ is a straightforward Boolean
encoding of vertex cover, we immediately get the following result.

\begin{proposition}\label{prop:NSOL-iS0^2-iS1^2}
  $\NSOL(\Gamma)$ is $\APX$-hard for every constraint
  language~$\Gamma$ satisfying
  $\iS_0^2 \subseteq \cc \Gamma$ or $\iS_1^2 \subseteq \cc \Gamma$.
\end{proposition}
\begin{proof}
  We have $\iS_0^2 = \cc{\set{[x\lor y]}}$ and
  $\iS_1^2 = \cc{\set{[\neg x \lor \neg y]}}$.  We discuss the
  former case, the latter one being symmetric and provable from the
  first one by Proposition~\ref{prop:dual}.

  We encode $\prbname{VertexCover}$ into $\NSOL(\set{[x\lor y]})$.
  For each edge $\set{x,y} \in E$ of a graph $G=(V,E)$ we add the
  clause $(x \lor y)$ to the formula~$\varphi_G$.  Every model $m'$
  of~$\varphi_G$ yields a vertex cover $\set{v\in V\mid m'(v) = 1}$,
  and conversely, the characteristic function of any vertex cover
  satisfies~$\varphi_G$. Moreover, we choose $m = \vec{0}$. Then
  $\hd(\vec{0}, m')$ is minimal if and only if the number of~$1$s
  in~$m'$ is minimal, i.e., if~$m'$ is a minimal model of~$\varphi_G$,
  i.e., if~$m'$ represents a minimal vertex cover of~$G$.  Since
  $\prbname{VertexCover}$ is $\APX$-complete (see
  e.g.~\cite{AusielloCGKMSP-99}) and
  $\NSOL(\set{[x\lor y]})\aple\NSOL(\Gamma)$ (see
  Corollary~\ref{cor:coclones}), the result follows.
\end{proof}

\begin{proposition}\label{prop:NSOL-D_2}
  We have $\NSOL(\Gamma) \in \APX$ for constraint languages
  $\Gamma\subseteq \iD_2$.
\end{proposition}
\begin{proof}
  $\Gamma':=\set{[x \oplus y], [x\to y]}$ is a base of $\iD_2$. By
  Corollary~\ref{cor:coclones} it suffices to show that $\NSOL(\Gamma')$
  is in $\APX$. Let $(\varphi, m)$ be an instance of
  this problem. Feasibility for $\varphi$ can be encoded as an integer
  program as follows: Every constraint $x \oplus y$ induces an
  equation $x + y=1$, every constraint $x \to y$ an inequality
  $x\le y$. If we restrict all variables to $\set{0,1}$ by the
  appropriate inequalities, it is clear that an assignment~$m'$
  satisfies~$\varphi$ if it satisfies the linear system with
  inequality side conditions.  As objective function we use
  $f(\vec x):=\sum_{x\in V_0} x + \sum_{x\in V_1} (1-x)$, where $V_i$
  is the set of variables mapped to~$i$ by~$m$.  Clearly, for every
  solution~$m'$ we have $f(m') = \hd(m, m')$.  The $2$-approximation
  algorithm from~\cite{HochbaumMNT-93} for integer linear programs,
  where every inequality contains at most two variables, completes the
  proof.
\end{proof}

\begin{proposition}\label{prop:NSOL-S_00}
  We have $\NSOL(\Gamma) \in \APX$ for constraint languages
  $\Gamma\subseteq\iS_{00}^\ell$ with $\ell\geq2$.
\end{proposition}
\begin{proof}
  $\Gamma':=\set{[x_1\lor \cdots \lor x_\ell], [x\to y], [\neg x], [x]}$
  is a base of $\iS_{00}^\ell$. By Corollary~\ref{cor:coclones} it
  suffices to show that $\NSOL(\Gamma')$ is in $\APX$.
  Let $(\varphi,m)$ be an instance of this problem. We use an approach
  similar to the one for the corresponding case
  in~\cite{KhannaSTW-01}, again writing~$\varphi$ as an integer
  program. We write constraints $x_1 \lor \cdots \lor x_\ell$ as
  inequalities $x_1+\cdots + x_\ell \ge 1$, constraints $x \to y$ as
  $x\le y$, $\neg x$ as $x=0$, and $x$ as $x=1$. Moreover, we add
  $x\geq 0$ and $x\leq 1$ for each variable~$x$.  It is easy to check
  that the feasible Boolean solutions of $\varphi$ and of the linear
  system coincide. As objective function we use
  $f(\vec x):=\sum_{x\in V_0} x + \sum_{x\in V_1} (1-x)$, where $V_i$
  is the set of variables mapped to~$i$ by~$m$.  Clearly, for every
  solution~$m'$ we have $f(m') = \hd(m, m')$.  Therefore it suffices
  to approximate the optimal solution for the integer linear program.

  To this end, let~$m''$ be a (generally non-integer) solution to the
  relaxation of the linear program, which can be computed in
  polynomial time. We construct~$m'$ by setting $m'(x) = 0$ if
  $m''(x)< 1 / \ell$ and $m'(x) = 1$ if $m''(x) \geq 1 / \ell$. As
  $\ell\geq 2$, we get
  $\hd(m,m')=f(m') \leq \ell f(m'') \leq \ell \cdot \OPT(\varphi, m)$.
  It is easy to check that~$m'$ is a feasible solution, which
  completes the proof.
\end{proof}


\subsubsection{NearestCodeword-Complete Cases}

This section essentially uses the facts that $\OptMinOnes$ is
$\NCW$-complete for the co-clone~$\iL_2$ and that it is a special case
of $\NSOL$. The following result was stated by Khanna et al. for
completeness via $\A$-reductions~\cite[Theorem 2.14]{KhannaSTW-01}. A
closer look at the proof reveals that it also holds for the stricter
notion of completeness via $\AP$-reductions that we use.

\begin{proposition}
  \label{prop:MinOnes-NCW-complete}
  The problem $\OptMinOnes(\Gamma)$ is $\NCW$-complete via
  $\AP$-reductions for constraint languages~$\Gamma$ satisfying
  $\cc{\Gamma} = \iL_2$.
\end{proposition}
\begin{proof}
  According to~\cite[Lemma 8.13]{KhannaSTW-01}, $\OptMinOnes(\Gamma)$
  is $\NCW$-hard for $\iL \subseteq \cc \Gamma$.  The proof uses
  $\AP$-reductions, i.e., we have $\NCW\aple\OptMinOnes(\Gamma)$.

  Regarding the other direction, $\OptMinOnes(\Gamma)\aple\NCW$, we
  first observe that $\odd^3=\SET{(a_1,a_2,a_3)\in\set{0,1}^3}{$\sum_i
    a_i$ odd}$ and $\even^3=\SET{(a_1,a_2,a_3)\in\set{0,1}^3}{$\sum_i
    a_i$ even}$ perfectly implement every constraint in $\iL_2$, i.e.,
  $\cc{\set{\odd^3,\even^3}}=\iL_2$~\cite[Lemma 7.6]{KhannaSTW-01}.
  Therefore, for $\Gamma\subseteq\iL_2$, the problem
  $\OptWMinOnes(\Gamma)$ $\AP$-reduces to
  $\OptWMinOnes(\set{\odd^3,\even^3})$~\cite[Lemma 3.9]{KhannaSTW-01}.
  The latter problem $\AP$-reduces to
  $\WMinCSP(\set{\odd^3,\even^3,[\lnot x]})$~\cite[Lemma
  8.1]{KhannaSTW-01}, which further $\AP$-reduces to
  $\WMinCSP(\set{\odd^3,\even^3})$ because of
  $[\lnot x]=[\even^3(x,x,x)]$.  In total we thus have that
  $\OptWMinOnes(\Gamma)$ $\AP$-reduces to
  $\WMinCSP(\set{\odd^3,\even^3})$.
  We conclude by observing that
  $\OptMinOnes$ is a particular case of~$\OptWMinOnes$ and that $\NCW$
  is the same as $\WMinCSP(\set{\odd^3,\even^3})$, yielding
  $\OptMinOnes(\Gamma)\aple\NCW$.
\end{proof}

\begin{lemma}\label{lem:MinOnes-aple-NSOL}%
  We have $\OptMinOnes(\Gamma)\aple \NSOL(\Gamma)$ for any constraint
  language~$\Gamma$.
\end{lemma}
\begin{proof}
  $\OptMinOnes(\Gamma)$ is a special case of $\NSOL(\Gamma)$ where~$m$
  is the constant $\vec 0$-assignment.%
\end{proof}

\begin{corollary}\label{cor:NSOL-NCW-hard-L}
  $\NSOL(\Gamma)$ is $\NCW$-hard for constraint languages~$\Gamma$
  satisfying $\iL \subseteq \cc \Gamma$.
\end{corollary}
\begin{proof}
  $\Gamma':=\set{\even^4,[x],[\neg x]}$ is a base of $\iL_2$.
  By Proposition~\ref{prop:MinOnes-NCW-complete},
  $\OptMinOnes(\Gamma')$ is $\NCW$-complete.  By
  Lemma~\ref{lem:MinOnes-aple-NSOL}, $\OptMinOnes(\Gamma')$ reduces to
  $\NSOL(\Gamma')$.  By Proposition~\ref{prop:unary}, $\NSOL(\Gamma')$
  is $\AP$-equivalent to $\NSOL(\set{\even^4})$.  Finally, because of
  $\even^4\in\iL\subseteq\cc{\Gamma}$ and
  Corollary~\ref{cor:coclones}, $\NSOL(\set{\even^4})$ reduces to
  $\NSOL(\Gamma)$.
\end{proof}

\begin{proposition}\label{prop:NSOL-MinOnes-to-weak_base-L_2}
  We have
  $\NSOL(\Gamma)\aple\OptMinOnes(\set{\even^4, [\neg x], [x]})$ for
  constraint languages~$\Gamma\subseteq\iL_2$.
\end{proposition}
\begin{proof}
  $\Gamma':=\set{\even^4, [\neg x], [x]}$ is a base of
  $\iL_2$. By Corollary~\ref{cor:coclones} it suffices
  to show $\NSOL(\Gamma') \aple \OptMinOnes(\Gamma')$.

  We proceed by reducing $\NSOL(\Gamma')$ to a subproblem of
  $\NSOLpp(\Gamma')$, where only instances $(\varphi,\vec{0})$ are
  considered. Then, using Proposition~\ref{prop:quantifiers} and
  Remark~\ref{rem:zero}, this reduces to a subproblem of
  $\NSOL(\Gamma')$ with the same restriction on the assignments, which
  is exactly $\OptMinOnes(\Gamma')$.  Note that $[x \oplus y]$ is
  equal to
  $\left[\exists z \exists z' (\even^4(x,y, z, z') \land \neg z \land
    z')\right]$ so we can freely use $[x \oplus y]$ in any
  $\Gamma'$-formula.  Let formula~$\varphi$ and assignment~$m$ be an
  instance of $\NSOL(\Gamma')$. We copy all clauses of~$\varphi$
  to~$\varphi'$.  For each variable~$x$ of~$\varphi$ for which
  $m(x)=1$, we take a new variable~$x'$ and add the constraint
  $x \oplus x'$ to~$\varphi'$. Moreover, we existentially
  quantify~$x$. Clearly, there is a bijection~$I$ between the
  satisfying assignments of~$\varphi$ and those of~$\varphi'$: For
  every solution~$s$ of~$\varphi$ we get a solution~$I(s)$
  of~$\varphi'$ by setting for each~$x'$ introduced in the
  construction of~$\varphi'$ the value $I(s)(x')$ to the complement of
  $s(x)$. Moreover, we have that $\hd(m, s) = \hd(\vec 0, I(s))$. This
  yields a trivial $\AP$-reduction with $\alpha=1$.
\end{proof}

\subsubsection{MinHornDeletion-Complete Cases}

\begin{proposition}[Khanna et al.~\cite{KhannaSTW-01}]\label{prop:kstw-minhd}
  The optimization problems
  $\OptMinOnes(\set{x\lor y\lor \neg z, x,\neg x})$ and
  $\OptWMinOnes(\set{x\lor y\lor \neg z, x\lor y})$ are
  $\MinHD$-complete via $\AP$-reductions.
\end{proposition}
\begin{proof}
  These results are stated in~\cite[Theorem 2.14]{KhannaSTW-01}
  for completeness via $\A$-reductions. The actual proof
  in~\cite[Lemma 8.7 and Lemma 8.14]{KhannaSTW-01}, however,
  uses $\AP$-reductions, hence the results also hold for our
  stricter notion of completeness.
\end{proof}

\begin{lemma}\label{lem:horncontain}
  $\NSOL(\set{x\lor y\lor \neg z})
   \aple \OptWMinOnes(\set{x\lor y\lor \neg z, x\lor y})$.
\end{lemma}
\begin{proof}
  Let formula~$\varphi$ and assignment~$m$ be an instance of
  $\NSOL(\set{x\lor y\lor \neg z})$ over the variables $x_1, \ldots,x_n$.
  Let $V_1$ be the set of variables $x_i$ with $m(x_i)=1$. We construct a
  $\set{x\lor y\lor \neg z, x\lor y}$-formula $\varphi'$ by adding to~$\varphi$
  for each $x_i\in V_1$ the constraint $x_i\lor x_i'$ where $x_i'$ is
  a new variable. We set the weights of the variables of~$\varphi'$ as
  follows. For $x_i\in V_1$ we set $w(x_i) = 0$, all other variables
  get weight~$1$. To each satisfying assignment $m'$ of $\varphi'$ we
  construct the assignment $m''$ which is the restriction of $m'$ to
  the variables of $\varphi$.  This construction is an
  $\AP$-reduction.

  Note that~$m''$ is feasible if $m'$ is. Let~$m'$ be an
  $r$-approximation of $\OPT(\varphi')$. Note that whenever for
  $x_i\in V_1$ we have $m'(x_i)=0$ then $m'(x_i')= 1$. The other way
  round, we may assume that whenever $m'(x_i)=1$ for $x_i\in V_1$
  then $m'(x_i') = 0$. If this is not the case, then we can
  change~$m'$ accordingly, decreasing the weight that way. It follows
  that $w(m') = n_0 + n_1$ where we have
  \begin{align*}
  n_0 &= \card{\Set{i}{x_i \in V_1, m'(x_i)=0}} = \card{\Set{i}{x_i
      \in V_1, m'(x_i)\neq m(x_i)}}\\
  n_1 &= \card{\Set{i}{x_i \notin V_1, m'(x_i)=1}} =
  \card{\Set{i}{x_i \notin V_1, m'(x_i)\neq m(x_i)}},
  \end{align*}
  which means that $w(m')$ equals $\hd(m,m'')$.  Analogously, any model
  $s\in[\varphi]$ can be extended to a model $m'\in[\varphi']$ by putting
  $m'(x_i') =1$ if $x_i\in V_1$ and $s(x_i)=0$, and $m'(x_i') =0$ for
  the remaining $x_i\in V_1$; thereby $w(m')=\hd(m,s)$. Consequently, the
  optima in both problems correspond, that is, we get
  $\OPT(\varphi') = \OPT(\varphi, m)$.  Hence we deduce
  $\hd(m, m'') = w(m') \leq r \OPT(\varphi')= r\OPT(\varphi,m)$.
\end{proof}

\begin{proposition}\label{prop:iV_2_to_MinOnes}
  For every dual Horn constraint language $\Gamma \subseteq \iV_2$ we
  have the reduction
  $\NSOL(\Gamma)\aple\OptWMinOnes(\set{x\lor y\lor \neg z, x\lor y})$.
\end{proposition}
\begin{proof}
  Since $\set{x\lor y\lor \neg z, x, \neg x}$ is a base of~$\iV_2$,
  by Corollary~\ref{cor:coclones} it suffices to prove the reduction
  $\NSOL(\set{x\lor y\lor \neg z, x, \neg x}) \aple
  \OptWMinOnes(\set{x\lor y\lor \neg z, x\lor y})$. To this end, first
  reduce $\NSOL(\set{x\lor y\lor \neg z, x, \neg x})$ to
  $\NSOL(x\lor y\lor \neg z)$ by Proposition~\ref{prop:unary} and then
  use Lemma~\ref{lem:horncontain}.
\end{proof}

\begin{proposition}\label{prop:MinOnes_to_iV_2}
  $\NSOL(\Gamma)$ is $\MinHD$-hard for finite~$\Gamma$
  with $\iV_2\subseteq \cc \Gamma$.
\end{proposition}
\begin{proof}
  For $\Gamma':=\set{x\lor y\lor\neg z, x,\neg x}$ we have
  $\MinHD\apeq \OptMinOnes(\Gamma')$ by
  Proposition~\ref{prop:kstw-minhd}.  Now it follows
  $\OptMinOnes(\Gamma')\aple \NSOL(\Gamma')\aple \NSOL(\Gamma)$ using
  Lemma~\ref{lem:MinOnes-aple-NSOL} and Corollary~\ref{cor:coclones}
  on the assumption $\Gamma'\subseteq\iV_2\subseteq \cc{\Gamma}$.
\end{proof}

\subsubsection{Poly-APX-Hardness}

\begin{proposition}\label{prop:duphardnc}
    The problem
    $\NSOL(\Gamma)$ is $\pAPX$-hard for constraint languages~$\Gamma$
    satisfying $\iN\subseteq \cc{\Gamma} \subseteq \iI_0$ or
    $\iN\subseteq \cc{\Gamma}\subseteq \iI_1$.
\end{proposition}
\begin{proof}
  The constraint language
  $\Gamma_1:=\set{\even^4, [x\to y], [x]}$ is a base of~$\iI_1$.
  $\OptMinOnes(\Gamma_1)$ is $\pAPX$-hard by Theorem~2.14
  of~\cite{KhannaSTW-01} and reduces to $\NSOL(\Gamma_1)$ by
  Lemma~\ref{lem:MinOnes-aple-NSOL}.
  Since
  $[x\to y] = [\dup^3(x,y,1)] = [\exists z (\dup^3(x, y, z) \land z)]$,
  as well as $\cc{\set{\even^4}} = \iL$, $\cc{\set{\dup^3}} = \iN$,
  and $\iL \subseteq \iN$,
  we have the reductions
  \begin{displaymath}
    \NSOL(\Gamma_1)\aple \NSOL(\Gamma_1\cup\set{\dup^3})
    \aple \NSOL(\set{\even^4,\dup^3,x}) \apeq \NSOL(\set{\dup^3,x})
  \end{displaymath}
  by Corollary~\ref{cor:coclones}. The problem of finding feasible
  solutions of $\NSOL(\Gamma)$, where
  $\iN\subseteq \cc{\Gamma} \subseteq \iI_0$ or
  $\iN\subseteq \cc{\Gamma} \subseteq \iI_1$,
  is polynomial-time decidable.
  Indeed, such~$\Gamma$ is $0$- or $1$-valid, therefore the all-zero or
  all-one tuple is always guaranteed to exist as a feasible solution.
  Therefore Proposition~\ref{prop:unary} implies
  $\NSOL(\set{\dup^3,x})\apeq \NSOL(\set{\dup^3})$;
  the latter problem reduces to $\NSOL(\Gamma)$ because
  $\dup^3 \in \iN\subseteq \cc{\Gamma}$
  and Corollary~\ref{cor:coclones}.
\end{proof}

\section{Finding Another Solution Closest to the Given One}
\label{sec:proofsXSOL}

In this section we study the optimization problem $\NextSol$. We first
consider the polynomial-time cases and then the cases of higher
complexity.

\subsection{Polynomial-Time Cases}
\label{ssec:proofsXSOL-PO}

Since we cannot take advantage of clone closure, we must proceed
differently. We use the following result
based on a theorem by Baker and Pixley~\cite{BakerP-75}.

\begin{proposition}[Jeavons et al.~\cite{JeavonsCG-97}]\label{prop:BakerPixley}
  Every bijunctive constraint $R(x_1, \ldots, x_n)$
  is equivalent to the conjunction $\bigwedge_{1 \leq i \leq j} R_{ij}(x_i,x_j)$,
  where~$R_{ij}$ is the projection of~$R$ to the
  coordinates~$i$ and~$j$.
\end{proposition}

\begin{proposition}\label{prop:XSOL-iD2}
  If\/ $\Gamma$ is bijunctive $(\Gamma\subseteq\iD_2)$ then
  $\XSOL(\Gamma)$ is in $\PO$.
\end{proposition}
\begin{proof}
  According to Proposition~\ref{prop:BakerPixley} we may assume that
  the formula~$\varphi$ is a conjunction of atoms $R(x,y)$
  or a unary constraint $R(x,x)$ of the form $[x]$ or
  $[\neg x]$.

  Unary constraints fix the value of the constrained variable and can
  be eliminated by propagating the value to the other clauses. For
  each of the remaining variables, $x$, we attempt to construct a model~$m_x$
  of~$\varphi$ with $m_x(x)\neq m(x)$ such that $\hd(m_x,m)$ is
  minimal among all models with this property. This can be done in
  polynomial time as described below.  If the construction of~$m_x$
  fails for every variable~$x$, then~$m$ is the sole model of~$\varphi$ and the
  problem is not solvable. Otherwise choose one of the variables~$x$
  for which $\hd(m_x,m)$ is minimal and return~$m_x$ as second
  solution~$m'$.

  It remains to describe the computation of~$m_x$. Initially we set
  $m_x(x)$ to $1-m(x)$ and $m_x(y):=m(y)$ for all variables~$y\neq x$,
  and mark~$x$ as flipped. If~$m_x$ satisfies all atoms we are
  done. Otherwise let $R(u,v)$ be an atom falsified by~$m_x$.  If~$u$
  and~$v$ are marked as flipped, the construction fails, a model~$m_x$
  with the property $m_x(x)\neq m(x)$ does not exist. Otherwise
  $R(u,v)$ contains a uniquely determined variable~$v$ not marked as
  flipped. Set $m_x(v) := 1-m(v)$, mark~$v$ as flipped, and repeat
  this step. This process terminates after flipping every variable at
  most once.
\end{proof}

\begin{proposition}\label{prop:XSOL-iI00m}
  If\/ $\Gamma \subseteq \iS_{00}^k$ or $\Gamma \subseteq \iS_{10}^k$
  for some $k \geq 2$ then $\XSOL(\Gamma)$ is in $\PO$.
\end{proposition}
\begin{proof}
  We perform the proof only for~$\iS_{00}^k$.
  Proposition~\ref{prop:dual} implies the same result
  for~$\iS_{10}^k$.

  The co-clone $\iS_{00}^k$ is generated by
  $\Gamma':=\set{\bor^k, [x\rightarrow y], [x], [\neg x]}$.  In fact,
  $\Gamma'$ is even a \emph{plain base} of
  $\iS_{00}^k$~\cite{CreignouKZ-08}, meaning that every relation
  in~$\Gamma$ can be expressed as a conjunctive formula over relations
  in~$\Gamma'$, without existential quantification or explicit equalities.
  Hence we may assume that $\varphi$ is given as a conjunction of
  $\Gamma'$-atoms.

  Note that $x \lor y$ is a polymorphism of $\Gamma'$, i.e., for any
  two solutions $m_1$, $m_2$ of $\varphi$ their disjunction
  $m_1 \lor m_2$~-- defined by $(m_1\lor m_2)(x) = m_1(x)\lor m_2(x)$
  for all~$x$~-- is also a solution of~$\varphi$. Therefore we get the
  optimal solution~$m'$ of an instance $(\varphi,m)$ by flipping
  in~$m$ either some ones to zeros or some zeros to ones, but not
  both. To see this, assume the optimal solution~$m'$ flips both ones
  and zeros. Then $m' \lor m$ is a solution of~$\varphi$ that is
  closer to $m$ than $m'$, which contradicts the optimality of~$m'$.

  Unary constraints fix the value of the constrained variable and can
  be eliminated by propagating the value to the other clauses (including
  removal of disjunctions containing implied positive literals and
  shortening disjunctions containing implied negative literals).
  This propagation does not lead to contradictions since~$m$ is a model
  of~$\varphi$. For
  each of the remaining variables, $x$, we attempt to construct a model~$m_x$
  of~$\varphi$ with $m_x(x)\neq m(x)$ such that $\hd(m_x,m)$ is
  minimal among all models with this property. This can be done in
  polynomial time as described below.  If the construction of~$m_x$
  fails for every variable~$x$, then~$m$ is the sole model of~$\varphi$ and the
  problem is not solvable. Otherwise choose one of the variables~$x$
  for which $\hd(m_x,m)$ is minimal and return~$m_x$ as second
  solution~$m'$.

  It remains to describe the computation of~$m_x$. If $m(x)=0$, we
  flip~$x$ to~$1$ and propagate this change iteratively along the
  implications, i.e., if $x \to y$ is a constraint of~$\varphi$ and
  $m(y)=0$, we flip~$y$ to~$1$ and iterate. This kind of flip never
  invalidates any disjunctions, it could only lead to contradictions
  with conditions imposed by negative unit clauses (and since their
  values were propagated before such a contradiction would be immediate).
  For $m(x)=1$ we proceed dually, flipping~$x$ to~$0$, removing~$x$
  from disjunctions if applicable, and propagating
  this change backward along implications $y \to x$ where $m(y)=1$. This
  can possibly lead to immediate inconsistencies with already inferred
  unit clauses, or it can produce contradictions through empty
  disjunctions, or it can create the necessity for further flips
  from~$0$ to $1$ in order to obtain a solution (because in a disjunctive
  atom all variables with value~$1$ have been flipped, and thus removed).
  In all these three cases the resulting assignment does not
  satisfy~$\varphi$, and there is no model that differs from~$m$
  in~$x$ and that can be obtained by flipping in one way only.
  Otherwise, the resulting assignment satisfies~$\varphi$, and this is
  the desired~$m_x$.
  Our process terminates after flipping every variable at most once,
  since we flip only in one way (from zeros to ones or from ones to
  zeros). Thus, $m_x$ is computable in polynomial time.
\end{proof}

\subsection{Hard Cases}
\label{ssec:proofsXSOL-hard}

\begin{proposition}\label{prop:XSOL-iI_0-iI_1}
  Let $\Gamma$ be a constraint language. If\/
  $\iI_1\subseteq \cc \Gamma$ or $\iI_0 \subseteq \cc \Gamma$ holds
  then it is $\NP$-complete to decide whether a feasible solution for
  $\XSOL(\Gamma)$ exists. Otherwise, $\XSOL(\Gamma)\in \pAPX$.
\end{proposition}
\begin{proof}
  Finding a feasible solution to $\XSOL(\Gamma)$ is exactly the
  problem $\AnotherSat(\Gamma)$ which is $\NP$-hard if and only if
  $\iI_1\subseteq \cc \Gamma$ or $\iI_0 \subseteq \cc \Gamma$
  according to Juban~\cite{Juban-99}.  If $\AnotherSat(\Gamma)$ is
  polynomial-time decidable, we can always find a feasible solution
  for $\XSOL(\Gamma)$ if it exists. Obviously, every feasible solution
  is an $n$-approximation of the optimal solution, where $n$ is the
  number of variables in the input.
\end{proof}

\subsubsection{Tightness results}\label{sssec:xsoltight}

It will be convenient to consider the following decision problem
asking for another solution that is not the complement, i.e., that does
not have maximal distance from the given one.

\decproblem{$\AnotherSatNC(\Gamma)$}%
{A conjunctive formula~$\varphi$ over relations from~$\Gamma$ and an
  assignment~$m$ satisfying~$\varphi$.}%
{Is there another satisfying assignment~$m'$ of~$\varphi$, different
  from~$m$, such that $\hd(m,m') < n$, where~$n$ is the number of
  variables in~$\varphi$?}

\begin{remark}
  $\AnotherSatNC(\Gamma)$ is $\NP$-complete
  for $\iI_0 \subseteq \cc \Gamma$ and $\iI_1 \subseteq \cc \Gamma$,
  since already $\AnotherSat(\Gamma)$ is $\NP$-complete for these
  cases, as shown in~\cite{Juban-99}.  Moreover,
  $\AnotherSatNC(\Gamma)$ is polynomial-time decidable if~$\Gamma$
  is Horn $(\Gamma \subseteq \iE_2)$, dual Horn
  $(\Gamma \subseteq \iV_2)$, bijunctive $(\Gamma \subseteq \iD_2)$,
  or affine $(\Gamma \subseteq \iL_2)$, for the same reason as
  for $\AnotherSat(\Gamma)$: For
  each variable~$x_i$ we flip the value of $m[i]$, substitute
  $\cmpl m(x_i)$ for~$x_i$, and construct another satisfying
  assignment if it exists. Consider now the solutions which we get for
  every variable~$x_i$.  Either there is no solution for any variable,
  then $\AnotherSatNC(\Gamma)$ has no solution; or there are only
  the solutions which are the complement of~$m$, then
  $\AnotherSatNC(\Gamma)$ has no solution as well; or else we get a
  solution~$m'$ with $\hd(m,m') < n$, leading also to a solution for
  $\AnotherSatNC(\Gamma)$. Hence, taking into account
  Proposition~\ref{prop:AScomplhard} below, we obtain a dichotomy
  result also for $\AnotherSatNC(\Gamma)$.
\end{remark}

Note that $\AnotherSatNC(\Gamma)$ is not compatible with
existential quantification. Let $\varphi(y, x_1, \ldots, x_n)$ with
model~$m$ be an instance of
$\AnotherSatNC(\Gamma)$ and let $m'$ be a solution satisfying
$\hd(m,m') < n+1$. Now consider the formula $\varphi_1(x_1, \ldots, x_n) =
\exists y \, \varphi(y, x_1, \ldots, x_n)$, obtained by existentially
quantifying the variable~$y$, and the tuples $m_1$ and~$m'_1$ obtained from
$m$ and~$m'$ by omitting the first component. Both, $m_1$ and~$m'_1$,
are still solutions of~$\varphi'$, but we cannot guarantee
$\hd(m_1, m'_1) < n$.
Hence we need the equivalent of Proposition~\ref{prop:coclonesmincd}
for this problem, whose proof is analogous.

\begin{proposition}\label{prop:coclonesanothersat}
  $\AnotherSatNC(\Gamma') \mle \AnotherSatNC(\Gamma)$ for
  constraint languages $\Gamma$ and $\Gamma'$ satisfying
  $\Gamma'\subseteq \ccc \Gamma$.
\end{proposition}

\begin{proposition}\label{prop:AScomplhard}
  If a constraint language~$\Gamma$ satisfies $\cc \Gamma = \iI$ or
  $\iN \subseteq \cc \Gamma \subseteq \iN_2$, then
  $\AnotherSatNC(\Gamma)$ is $\NP$-complete.
\end{proposition}
\begin{proof}
  Containment in $\NP$ is clear, it remains to show
  hardness. Since $\AnotherSatNC$ is not compatible with
  existential quantification, we cannot use clone theory, but have to
  consider the three co-clones $\iN_2$, $\iN$, and
  $\iI$ separately and make use of minimal weak bases.

  \paragraph{Case $\cc \Gamma = \iN$:}
  Putting $R:=\set{000, 101, 110}$, we give a reduction from
  $\AnotherSat(\set{R})$, which is $\NP$-hard~\cite{Juban-99} as
  $\cc{\set{R}} = \iI_0$. The problem remains
  $\NP$-complete if we restrict it to instances $(\varphi, \vec 0)$,
  since $R$ is $0$-valid and any given model~$m$ other than the
  constant $0$-assignment admits the trivial solution $m'=\vec 0$.
  Thus we can perform a reduction from this restricted
  problem.

  Consider the relation $R_{\iN}=\set{0000,1010,1100,1111,0101,0011}$.
  Given a formula~$\varphi$ over~$R$, we construct a
  formula~$\psi$ over $R_{\iN}$ by replacing every constraint
  $R(x, y, z)$ with $R_{\iN}(x, y, z, w)$, where~$w$ is a
  new global variable. Moreover, we set~$m$ to the constant
  $0$-assignment. This construction is a many-one reduction from the
  restricted version of $\AnotherSat(\set R)$ to
  $\AnotherSatNC(\set{R_{\iN}})$.

  To see this, observe that the tuples in~$R_{\iN}$ that have a~$0$ in
  the last coordinate are exactly those in $R\times \set{0}$. Thus any
  solution of~$\varphi$ can be extended to a solution of~$\psi$ by
  assigning~$0$ to~$w$. Conversely we observe that any solution~$m'$ of the
  $\AnotherSatNC(\set{R_{\iN}})$ instance $(\psi,\vec0)$ is different from
  $\vec 0$ and~$\vec 1$. As $R_{\iN}$ is complementive, we may assume
  $m'(w)=0$. Then~$m'$ restricted to the variables of~$\varphi$
  solves the $\AnotherSat(\set{R})$ instance $(\varphi,\vec0)$.

  Finally, observe that $R_{\iN}$ is a minimal weak base
  and $\Gamma$ is a base of the co-clone~$\iN$, therefore
  we have $R_{\iN}\in\ccc\Gamma$ by Theorem~\ref{thm:weakbases}.
  Now the $\NP$-hardness of $\AnotherSatNC(\Gamma)$ follows
  from the one of $\AnotherSatNC(\set{R_{\iN}})$ by
  Proposition~\ref{prop:coclonesanothersat}.

  \paragraph{Case $\cc \Gamma = \iN_2$:} We give a reduction from
  $\AnotherSatNC(\set{R_{\iN}})$, which is $\NP$-hard by the previous
  case. By Theorem~\ref{thm:weakbases}, $\ccc\Gamma$ contains the
  relation $R_{\iN_2} = \Set{ m\cmpl m}{m \in R_{\iN}}$. For an
  $R_{\iN}$-formula~$\varphi(x_1, \ldots, x_n)$, we construct an
  $R_{\iN_2}$-formula $\psi(x_1, \ldots, x_n, x_1', \ldots, x_n')$ by
  replacing every constraint $R_{\iN}(x, y, z, w)$ with
  $R_{\iN_2}(x, y, z, w, x', y', z', w')$. Assignments~$m$ for
  $\varphi$ extend to assignments $M$ for~$\psi$ by setting
  $M(x'):= \cmpl{m}(x)$. Conversely, assignments for~$\psi$ yield
  assignments for~$\varphi$ by restricting them to the variables
  in~$\varphi$. Because every variable $x_1,\dotsc,x_n$ assigned by
  models of $\varphi$ actually occurs in some $R_{\iN}$-atom in
  $\varphi$ and hence in some $R_{\iN_{2}}$-atom of $\psi$, and
  because of the structure of $R_{\iN_{2}}$, any model of $\psi$ distinct
  from $M$ and $\cmpl{M}$ restricts to a model of $\varphi$ other than
  $m$ or $\cmpl{m}$.
  Consequently, this construction is again a reduction from
  $\AnotherSatNC(\set{R_{\iN}})$ to
  $\AnotherSatNC(\set{R_{\iN_2}})$, reducing itself to
  $\AnotherSatNC(\Gamma)$ by Proposition~\ref{prop:coclonesanothersat}.

  \paragraph{Case $\cc \Gamma = \iI$:} We proceed as in Case
  $\cc \Gamma = \iN$, but use $R_{\iI}=\set{0000,0011,0101,1111}$
  instead of~$R_{\iN}$, and $\set{000, 011, 101}$ for~$R$. Note that
  the $R_{\iI}$-tuples with first coordinate~$0$ are exactly those in
  $\set{0}\times R$.  The relation $R_{\iI}$ is not complementive, but
  (as every variable assigned by any model of $\psi$ occurs in some
  atomic $R_{\iI}$-constraint) the only
  solution $m'$ such that $m'(w)=1$ is the constant $1$-assignment,
  which is ruled out by the requirement $\hd(m,m') < n$. Hence we may
  again assume $m'(w)=0$.
\end{proof}

\begin{proposition}\label{prop:tightxsol}
  For a constraint language~$\Gamma$ satisfying
  $\cc \Gamma = \iI$ or $\iN \subseteq \cc \Gamma \subseteq \iN_2$ and
  any $\varepsilon>0$ there is no polynomial-time
  $n^{1-\varepsilon}$-approximation algorithm for $\XSOL(\Gamma)$,
  unless $\P = \NP$.
\end{proposition}
\begin{proof}
  Assume that there is a constant $\varepsilon>0$ with a
  polynomial-time $n^{1-\varepsilon}$-approximation algorithm for
  $\XSOL(\Gamma)$. We show how to use this algorithm to solve
  $\AnotherSatNC(\Gamma)$ in polynomial time.
  Proposition~\ref{prop:AScomplhard} completes the proof.

  Let $(\varphi, m)$ be an instance of $\AnotherSatNC(\Gamma)$
  with~$n$ variables. If $n=1$, then we reject the
  instance. Otherwise, we construct a new formula~$\varphi'$ and a new
  assignment~$m'$ as follows. Let~$k$ be the smallest integer greater
  than $1/\varepsilon$. Choose a variable~$x$ of~$\varphi$ and
  introduce $n^k-n$ new variables~$x^i$ for $i = 1, \ldots, n^k-n$.
  For every $i \in \set{1, \ldots, n^k-n}$ and every constraint
  $R(y_1, \ldots , y_\ell)$ in~$\varphi$, such that
  $x \in \set{y_1, \ldots, y_\ell}$, construct a new constraint
  $R(z_1^i, \ldots, z_\ell^i)$ by $z_j^i = x^i$ if $y_j = x$ and
  $z_j^i = y_j$ otherwise; add all the newly constructed constraints
  to~$\varphi$ in order to get~$\varphi'$.  Moreover, we extend~$m$ to
  a model of~$\varphi'$ by setting $m'(x^i)=m(x)$. Now run the
  $n^{1-\varepsilon}$-approximation algorithm for $\XSOL(\Gamma)$ on
  $(\varphi', m')$. If the answer is $\cmpl{m'}$ then reject,
  otherwise accept.

  We claim that the algorithm described above is a correct
  polynomial-time algorithm for the decision problem
  $\AnotherSatNC(\Gamma)$ when~$\Gamma$ is
  complementive. Polynomial runtime is clear. It remains to show its
  correctness. If the only solutions to~$\varphi$ are~$m$
  and~$\cmpl m$, then, as $n>1$, the only models of~$\varphi'$ are $m'$
  and $\cmpl{m'}$. Hence the approximation algorithm must
  answer $\cmpl{m'}$ and the output is correct. Now assume that there is
  a satisfying assignment~$m_s$ different from~$m$ and~$\cmpl m$. The
  relation~$[\varphi]$ is complementive, hence we may assume that
  $m_s(x)=m(x)$. It follows that~$\varphi'$ has a satisfying
  assignment~$m_s'$ for which $0<\hd(m_s', m')<n$ holds.  But then the
  approximation algorithm must find a satisfying assignment~$m''$
  for~$\varphi'$ with
  $\hd(m', m'') < n \cdot (n^k)^{1-\varepsilon} =
  n^{k(1-\varepsilon)+1}$. Since the inequality $k > 1/\varepsilon$
  holds, it follows that $\hd(m', m'')< n^k$. Consequently, $m''$ is
  not the complement of~$m'$ and the output of our algorithm is again
  correct.

  When~$\Gamma$ is not complementive but both $0$-valid and $1$-valid
  $(\cc \Gamma = \iI)$, we perform the expansion algorithm described
  above for each variable of the formula~$\varphi$ and reject if the
  result is the complement for each run. The runtime remains
  polynomial. If $[\varphi] = \set{m,\cmpl{m}}$, then indeed every run
  results in the corresponding $\cmpl{m'}$, and we correctly reject.
  Otherwise, we have a model
  $m_s\in [\varphi]\smallsetminus\set{m,\cmpl{m}}$, so there is a
  variable~$x$ of $\varphi$, where $m_s(x)\neq \cmpl{m}(x)$, i.e.\
  $m_s(x) = m(x)$. For this instance $(\varphi',m')$ the
  approximation algorithm does not return~$\cmpl{m'}$, wherefore we
  correctly accept.
\end{proof}

\subsubsection{MinDistance-Equivalent Cases}\label{sssec:XSOLaffine}

In this section we show that affine co-clones give rise to
problems equivalent to $\MinDist$.

\begin{lemma}\label{lem:affine-MinDist}
  For affine constraint languages $\Gamma$ $(\Gamma\subseteq
  \iL_2)$ we have $\XSOL(\Gamma)\aple\MinDist$.
\end{lemma}
\begin{proof}
  Let the formula~$\varphi$ and the satisfying assignment~$m$ be an
  instance of $\XSOL(\Gamma)$ over the variables $x_1, \ldots,
  x_n$. The input $\varphi$ can be written as $A \vec x = \vec b$,
  with~$m$ being a solution of this affine system. A tuple~$m'$ is a solution
  of $A\vec x = \vec b$ if and only if it can be written as
  $m' = m+m_0$ where~$m_0$ is a solution of $A \vec x = \vec 0$.
  The Hamming distance is invariant with respect to affine translations:
  namely we have $\hd(m',m) = \hd(m'+m'', m+m'')$ for any tuple $m''$, in
  particular, for $m'' = -m$ we obtain $\hd(m',m)=\hd(m'-m,\vec{0})$. Therefore
  $m'\neq m$ is a solution of $A \vec x = \vec b$ with minimal Hamming
  distance to~$m$ if and only if $m_0 = m'-m$ is a non-zero solution of the
  homogeneous system $A \vec x = \vec 0$ with minimum Hamming
  weight. Hence, the problem $\XSOL(\Gamma)$ for affine
  languages~$\Gamma$ is equivalent to computing the non-trivial solutions of
  homogeneous systems with minimal weight, which is exactly the
  $\MinDist$ problem.
\end{proof}

We need to express an affine sum of even number of variables by means
of the minimal weak base for each of the affine co-clones. In the
following lemma, the existentially quantified variables are uniquely
determined, therefore the existential quantifiers serve only to hide
superfluous variables and do not pose any problems as they were
mentioned before.

\begin{lemma}\label{lem:represent-sums}
For every $n\in\NN$, $n\geq 1$, the constraint
$x_1 \oplus x_2 \oplus\dotsm \oplus x_{2n} = 0$ can be equivalently
expressed by each of the following formulas:
\begin{enumerate}
\item\label{item:sumiL}
      $\exists y_0,\dotsc,y_n( y_0 = 0\land y_{n} = 0\land\
       \begin{alignedat}[t]{2}%
       &R_{\iL}(y_0,x_1,x_2,y_1) &\land\
       & R_{\iL}(y_1,x_3,x_4,y_2) \land \dotsm\land{}\\
       &R_{\iL}(y_{n-1},x_{2n-1},x_{2n},y_{n})),%
       \end{alignedat}$
\item $\exists y_0,\dotsc,y_{2n}(
       \begin{alignedat}[t]{2}%
       &R_{\iL_0}(y_0,x_1,y_1,y_0) &\land\
       &R_{\iL_0}(y_1,x_2,y_2,y_0) \land \dotsm\land\\
       &R_{\iL_0}(y_{2n-1},x_{2n},y_{2n},y_{2n})),%
       \end{alignedat}$
\item $\exists y_0,\dotsc,y_{2n}(
       \begin{alignedat}[t]{2}%
       &R_{\iL_1}(y_0,x_1,y_1,y_0) &\land\
       &R_{\iL_1}(y_1,x_2,y_2,y_0) \land \dotsm\land\\
       &R_{\iL_1}(y_{2n-1},x_{2n},y_{2n},y_{2n})),
       \end{alignedat}$
\item $\exists y_0,\dotsc,y_n,z_0,\dotsc,z_n,w_1,\dotsc,w_{2n}(
       y_0 = 0\land y_n = 0\land{}$\\
       \strut\hfill$\begin{alignedat}[t]{1}%
       &R_{\iL_3}(y_0,x_1,x_2,y_1,z_0,w_1,w_2,z_1) \land\
        R_{\iL_3}(y_1,x_3,x_4,y_2,z_1,w_3,w_4,z_2) \land \dotsm\land\\
       &R_{\iL_3}(y_{n-1},x_{2n-1},x_{2n},y_{n},
                  z_{n-1},w_{2n-1},w_{2n},z_{n})),
       \end{alignedat}$
\item\label{item:sumiL2}
      $\exists y_0,\dotsc,y_{2n},z_0,\dotsc,z_{2n},w_1,\dotsc,w_{2n}(%
       \begin{alignedat}[t]{1}
       &R_{\iL_2}(y_0,x_1,y_1,z_0,w_1,z_1,y_0,z_0) \land{}\\
       &R_{\iL_2}(y_1,x_2,y_2,z_1,w_2,z_2,y_0,z_0) \land \dotsm\land{}\\
       &R_{\iL_2}(y_{2n-1},x_{2n},y_{2n},z_{2n-1},w_{2n},z_{2n},
                 y_{2n},z_{2n})),%
       \end{alignedat}$
\end{enumerate}
where the number of existentially quantified variables is linearly
bounded in the length of the constraint. Note moreover that in each case
any model of $x_1 \oplus x_2 \oplus\dotsm \oplus x_{2n} = 0$ uniquely
determines the values of the existentially quantified variables.
\end{lemma}
\begin{proof}
  Write out the constraint relations following the existential
  quantifiers as (conjunctions of) equalities. From this uniqueness of
  valuations for the existentially quantified variables is easy to see,
  and likewise that any model of $\bigoplus_{i=1}^{2n} x_i = 0$ also
  satisfies each of the formulas~\ref{item:sumiL}. up
  to~\ref{item:sumiL2}. Adding up the equalities behind the existential
  quantifiers shows the converse direction.
\end{proof}

The following lemma shows that $\MinDist$ is $\AP$-equivalent to a
restricted version, containing only constraints generating the minimal
weak base, for each co-clone in the affine case.

\begin{lemma}\label{lem:MD-red-to-weak-base-and-zero}
For each co-clone $\mathcal{B}\in\set{\iL,\iL_0,\iL_1,\iL_2,\iL_3}$ we
have the reduction $\MinDist\aple \XSOL(\set{R_{\mathcal{B}},[\neg x]})$.
\end{lemma}
\begin{proof}
  Consider a co-clone $\mathcal{B}\in\set{\iL,\iL_0,\iL_1,\iL_2,\iL_3}$
  and a $\OptMinDist$-instance represented by a matrix
  $A\in\ZZ_2^{k\times l}$. If one of the columns of~$A$, say the $i$-th,
  is zero, then the $i$-th unit vector is an optimal solution to this
  instance with optimal value~$1$. Hence, we assume from now on that none
  of the columns equals a zero vector.
  \par
  Every row of $A$ expresses the fact that a sum
  of $n\leq l$ variables equals zero. If $n$ is odd, we extend
  this sum to one with $n+1$ summands, thereby introducing a new
  variable~$v$, which we existentially quantify and confine to zero using
  a unary $[\neg x]$-constraint. Then we replace the expanded sum by the
  existential formula from Lemma~\ref{lem:represent-sums} corresponding
  to the co-clone~$\mathcal{B}$ under consideration. This way we have
  introduced only linearly many new variables in~$l$ for every row, and
  for any feasible solution for the $\OptMinDist$-problem the values of
  the existential variables needed to encode it are uniquely determined.
  Thus, taking the conjunction over all these formulas we only
  have a linear growth in the size of the instance.
  \par
  Next, we show how to deal with the existential quantifiers:
  First we transform the expression to prenex normal form getting a
  formula~$\psi$ of the
  form~$\exists y_1,\dotsc,y_p(\varphi(y_1,\dotsc,y_p,x_1,\dotsc,x_l))$,
  which holds if and only if $A\vec{x} = \vec{0}$ for
  $\vec{x} = (x_1,\dotsc,x_l)$.
  We use the same blow-up construction regarding $x_1,\dotsc,x_l$ as in
  Proposition~\ref{prop:quantifiers} and Lemma~\ref{lem:Heindl} to
  make the influence of $y_1,\dotsc,y_p$ on the Hamming distance
  negligible.
  For this we put $J:=\set{1,\dotsc,t}$ and introduce
  new variables $x_i^j$ where $1\leq i\leq l$ and $j\in J$.
  If $u$ is among $x_1,\dotsc,x_l$, we define its blow-up set to be
  $B(u) = \Set{x_i^j}{j\in J}$, otherwise, for
  $u\in\set{y_1,\dotsc,y_p}$, we set $B(u) = \set{u}$. Now for each atom
  $R(u_1,\dotsc,u_q)$ of $\varphi$ we form the set of atoms
  $\Set{R(u_1',\dotsc,u_q')}{(u_1',\dotsc,u_q')\in\prod_{i=1}^q B(u_i)}$,
  and define the quantifier free formula $\varphi'$ to be the
  conjunction of all atoms in the union of these sets. Note that this
  construction takes time polynomial in the size of~$\psi$ and hence
  in the size of the input $\OptMinDist$-instance whenever~$t$ is
  polynomial in the input size because the atomic relations
  in~$\psi$ are at most octonary.

  If $s$ is an assignment of values to $\vec{x}$ making
  $A\vec{x} =\vec{0}$ true, we define $s'(x_i^j):= s(x_i)$ and extend
  this to a model of~$\varphi'$ assigning the uniquely determined values
  to $y_1,\dotsc,y_p$. Let~$m'$ be the model arising in this way from the
  zero assignment $m$. If~$s'$ is any model of~$\varphi'$, then for every
  $1\leq i\leq l$, all $j\in J$ and each atom $R(u_1,\dotsc,u_q)$
  of~$\varphi$, $s'$ satisfies, in particular, the conjunction
  $R(u_1',\dotsc,u_q')\land R(u_1'',\dotsc,u_q'')$ where for
  $u\in\set{u_1,\dotsc,u_q}$ we have
  $u' = u'' = u$ if $u \in\set{y_1,\dotsc,y_p}$,
  $u'=x_i^1$, $u'' = x_i^j$ if $u=x_i$, and
  $u'=u''=x_k^1$ if $u=x_k$ for some
  $k\in\set{1,\dotsc,l}\smallsetminus\set{i}$.
  Hence, the vectors $(s'(x_1^1),\dotsc,s'(x_l^1))$
  and
  $(s'(x_1^1),\dotsc,s'(x_{i-1}^1),s'(x_i^j),s'(x_{i+1}^1),\dotsc,
                                                          s'(x_l^1)))$
  both belong to the kernel of~$A$ and so does their difference, which
  is $s'(x_i^j) - s'(x_i^1)$ times the $i$-th unit vector. As the $i$-th
  column of~$A$ is non-zero, we must have $s'(x_i^j) = s'(x_i^1)$.
  This also implies that if $s'$ is zero on $x_1^1,\dotsc,x_n^1$, then
  it must be zero on all $x_i^j$ ($1\leq i\leq l$, $j\in J$) and thus it
  must coincide with $m'$. Therefore, every feasible solution to
  the~$\XSOL$-instance $(\varphi',m')$ yields a non-zero vector
  $(s'(x_1^1),\dotsc,s'(x_l^1))$ in the kernel of~$A$.
  \par
  Further, if $s'$ is an $r$-approximation to an optimal solution,
  i.e.\ we have $\hd(s',m')\leq r\OPT(\varphi',m')$, then, as
  $s'(x_i^1)=s'(x_i^j)$ holds for all $j\in J$ and all $1\leq i\leq l$,
  we obtain a solution to the $\MinDist$ problem with Hamming weight~$w$
  such that $t\cdot w \leq \hd(s',m')$. Also, any optimal solution to
  the $\MinDist$-instance can be extended to a not-necessarily optimal
  solution $s''$ of~$(\varphi',m')$, for which one can bound the distance
  to~$m'$ as follows:
  $\OPT(\varphi',m')\leq \hd(s'',m') \leq t\cdot\OPT(A) + p$. Combining
  these inequalities, we can infer
  $t\cdot w  \leq r\cdot t\cdot \OPT(A) + r\cdot p$, or
  $w\leq \OPT(A)\cdot(r+ r/\OPT(A)\cdot p/t)$. We noted above that $p$ is
  linearly bounded in the size of the input, thus choosing $t$ quadratic
  in the size of the input bounds $w$ by $\OPT(A)(r+ \Lo(1))$, whence we
  have an AP-reduction with $\alpha=1$.
\end{proof}

\begin{lemma}\label{lem:XSOL-reduce-constants}
  For constraint languages $\Gamma$, where one can decide the
  existence of (and then find) a feasible solution of\/~$\XSOL(\Gamma)$
  in polynomial time, we have
  $\XSOL(\Gamma) \aple
  \XSOL((\Gamma\smallsetminus\set{[x],[\neg x]})\cup\set{\eq})$.
\end{lemma}
\begin{proof}
If an instance $(\varphi,m)$ does not have feasible solutions, then it
does not have nearest other solutions either. So we map it to the
generic unsolvable instance~$\bot$. Consider now formulas $\varphi$ over
variables $x_1,\dotsc,x_n$ with models~$m$ where some feasible solution
$s_0\neq m$ exists (and has been computed).

We can assume $\varphi$ to be of the form
$\psi(x_1,\dotsc,x_n) \land \bigwedge_{i\in I_1} [x_i]
                      \land \bigwedge_{i\in I_0} [\neg x_i]$, where
$\psi$ is a $(\Gamma\smallsetminus\set{[x],[\neg x]})$-formula and
$I_1,I_0\subseteq \set{1,\dotsc,n}$. We transform $\varphi$ to
$\varphi':= \psi(x_1,\dotsc,x_n) \land \bigwedge_{i\in I_1} x_i \eq y_1
                                 \land \bigwedge_{i\in I_0} x_i \eq z_1
                                 \land \bigwedge_{i=1}^{1+n^2}
                                          (y_i \eq y_1 \land z_i \eq z_1)$
and extend models of~$\varphi$ to models of~$\varphi'$ in the natural
way. Conversely, if~$s'$ is a model of~$\varphi'$ and $s'(y_i) = 1$ and
$s'(z_i) = 0$ hold for all $1\leq i\leq 1+n^2$, %
then we can restrict it to a model of~$\varphi$. Other models
of~$\varphi'$ are not optimal and are mapped to $s_0$. It is not hard to
see that this provides an AP-reduction with $\alpha=1$.
%
\end{proof}

\begin{proposition}\label{prop:MinDist-hardness-XSOL}
  For every constraint language $\Gamma$ satisfying
  $\iL\subseteq \cc \Gamma \subseteq \iL_2$ we have
  $\MinDist \apeq \XSOL(\Gamma)$.
\end{proposition}
\begin{proof}
  Since we lack compatibility with existential quantification, we shall
  deal with each co-clone $\mathcal{B} = \cc\Gamma$ in the interval
  $\set{\iL,\iL_0,\iL_1,\iL_2,\iL_3}$ separately. First we perform the
  reduction from Lemma~\ref{lem:MD-red-to-weak-base-and-zero} to
  $\XSOL(\set{R_{\mathcal{B}}, [\neg x]})$. We need to find a reduction
  to $\XSOL(\set{R_{\mathcal{B}}})$ as this reduces to~$\XSOL(\Gamma)$
  by Proposition~\ref{prop:coclonesmincd} and
  Theorem~\ref{thm:weakbases}.

  This is simple in the case of $\iL_0$ and
  $\iL_2$ since
  $[\neg x] = \set{x\mid R_{\iL_0}(x,x,x,x)}\in\ccc{\set{R_{\iL_0}}}$
  (see Proposition~\ref{prop:coclonesmincd}) and
  $[\neg x] = \set{x\mid \exists y(R_{\iL_2}(x,x,x,y,y,y,x,y))}$, where
  the existential quantifier can be handled by an AP-reduction with
  $\alpha=1$ which drops the quantifier and extends every model by
  assigning $1$ to all previously existentially quantified variables.
  Thereby (optimal) distances between models do not change at all.
  \par

  In the remaining cases, we reduce
  $\XSOL(\set{R_{\mathcal{B}},[\neg x]})\aple
   \XSOL(\set{R_{\mathcal{B}},[x],[\neg x]})$ and the latter to
  $\XSOL(\set{R_{\mathcal{B}}, \eq})$ by
  Lemma~\ref{lem:XSOL-reduce-constants}, which now has to be reduced to
  $\XSOL(\set{R_{\mathcal{B}}})$. This is obvious for
  $\mathcal{B} = \iL$ where equality constraints $x\eq y$ can be expressed
  as $R_{\iL}(x,x,x,y) \in \ccc{\set{R_{\iL}}}$ (cf.\
  Proposition~\ref{prop:coclonesmincd}). For $\iL_1$ the same can be
  done using the formula $\exists z(R_{\iL_1}(x,y,z,z))$, where the
  existential quantifier can be removed by the same sort of simple
  AP-reduction with $\alpha=1$ as employed for $\iL_2$.
  Finally, for $\iL_3$ we want to express equality as
  $\exists u\exists v(R_{\iL_3}(x,x,x,y,u,u,u,v))$. Here, in an
  AP-reduction, the quantifiers cannot simply be disregarded, as the
  values of the existentially quantified variables are not constant for
  all models. They are uniquely determined by the values of $x$ and $y$
  for each particular model, though, which allows us to perform a
  similar blow-up construction as in the proof of
  Lemma~\ref{lem:MD-red-to-weak-base-and-zero}.
  \par

  In more detail, given a $\set{R_{\iL_3}, \eq}$-formula $\psi$ containing
  variables $x_1,\dotsc,x_l$, first note that each atomic
  $R_{\iL_3}$-constraint $R_{\iL_3}(x_1,\dotsc,x_8)$ can be represented
  as a linear system of equations, namely $\oplus_{i=1}^4 x_i = 0$ and
  $x_i \oplus x_{i+4} = 1$ for $1\leq i\leq 4$. Since equalities
  $x_i \eq x_j$ can be written as $x_i\oplus x_j = 0$, the formula $\psi$
  is equivalent to an expression of the form $A\vec{x} = \vec{b}$ where
  $\vec{x} = (x_1,\dotsc,x_l)$. Replacing each equality constraint by
  the existential formula above and bringing the result into prenex
  normal form, we get a formula
  $\exists y_1,\dotsc,y_p(\varphi(y_1,\dotsc,y_p,x_1,\dotsc,x_l))$,
  which is equivalent to $\psi$ and where $\varphi$ is a conjunctive
  $\set{R_{\iL_3}}$-formula. By construction any two models of $\varphi$
  that agree on $x_1,\dotsc,x_l$ must coincide. Thus, introducing
  variables $x_i^j$ for $1\leq i\leq l$ and $j\in J:=\set{1,\dotsc,t}$
  and defining $\varphi'$ in literally the same way as in the proof of
  Lemma~\ref{lem:MD-red-to-weak-base-and-zero}, any model~$s$
  of~$\psi$ yields a model~$s'$ of $\varphi'$ by putting
  $s'(x_i^j):=s(x_i)$ for $1\leq i\leq l$ and $j\in J$ and extending
  this with the unique values for $y_1,\dotsc,y_p$ satisfying
  $\varphi(y_1,\dotsc,y_p,x_1,\dotsc,x_l)$. In this way we obtain a
  model~$m'$ of~$\varphi'$ from a given solution~$m$ of~$\psi$.
  Besides, if $s'$ is any model of~$\varphi'$, then as
  in Lemma~\ref{lem:MD-red-to-weak-base-and-zero}, the vectors
  $(s'(x_1^1),\dotsc,s'(x_l^1))$ and
  $(s'(x_1^1),\dotsc,s'(x_{i-1}^1),s'(x_i^j),s'(x_{i+1}^1),\dotsc,
                                                          s'(x_l^1)))$
  both satisfy~$\psi$, and thus their difference is in the kernel of~$A$.
  Since the variable $x_i$ occurs in at least one of the atoms
  of~$\psi$, the $i$-th column of~$A$ is non-zero, implying that
  $s'(x_i^j) = s'(x_i^1)$ for $j\in J$ and all $1\leq i\leq l$. Thus,
  any  model $s'\neq m'$ of~$\varphi'$ gives a model~$s\neq m$
  of~$\psi$ by defining $s(x_i):= s'(x_i^1)$ for all $1\leq i\leq l$.

  The presented construction is an AP-reduction with $\alpha=1$, which
  can be proven completely analogously to the last paragraph of the
  proof of Lemma~\ref{lem:MD-red-to-weak-base-and-zero}, choosing~$t$
  quadratic in the size of~$\psi$.
\end{proof}

\subsubsection{MinHornDeletion-Equivalent Cases}

As in Proposition~\ref{prop:AScomplhard} the need to use conjunctive
closure instead of $\cc{\ }$ causes a case distinction in the proof of
the following result.

\begin{lemma}\label{lem:constructimpl}
  If\/~$\Gamma$ is exactly dual Horn
  $(\iV \subseteq \cc \Gamma \subseteq \iV_2)$ then one of the
  following relations is in $\ccc \Gamma$: $[x\to y]$,
  $[x\to y]\times \set{0}$, $[x\to y]\times \set{1}$, or
  $[x\to y]\times \set{01}$.
\end{lemma}
\begin{proof}
  The co-clone $\cc\Gamma$ is equal to $\iV$, $\iV_0$, $\iV_1$, or
  $\iV_2$.  In the case $\cc \Gamma= \iV$ the relation $R_\iV$ belongs
  to~$\ccc \Gamma$ by Theorem~\ref{thm:weakbases}; because of
  $R_\iV(y, y, y, x) = [x \to y]$ we have
  $[x\to y] \in \ccc {R_\iV}\subseteq \ccc \Gamma$.  The case
  $\cc \Gamma = \iV_1$ leads to
  $[x\to y]\times \set{1} \in \ccc \Gamma$ in an analogous manner.
  The cases $\cc \Gamma = \iV_0$ and $\cc \Gamma = \iV_2$ lead to
  $[x\to y]\times \set{0}\in \ccc \Gamma$ and
  $[x\to y]\times \set{01} \in \ccc \Gamma$, respectively, by
  observing that
  $[S_1(y,y,x)]=[S_0(\lnot y,\lnot y, \lnot x,\lnot y)]
  =[(\lnot y\land\lnot y)\eq(\lnot y\land\lnot x)] = [x\to y]$.
\end{proof}

\begin{lemma}\label{lem:Horn_MinHD-hard}
  If\/~$\Gamma$ is exactly dual Horn
  $(\iV \subseteq \cc \Gamma \subseteq \iV_2)$, then $\XSOL(\Gamma)$
  is $\MinHD$-hard.
\end{lemma}
\begin{proof}
  There are four cases to consider, namely
  $\cc{\Gamma}\in\set{\iV,\iV_0,\iV_1,\iV_2}$. For simplicity we only
  present the situation where $\cc{\Gamma} = \iV_1$; the
  case $\cc{\Gamma}=\iV_2$ is very similar, and the other possibilities
  are even less complicated. At the end we shall give a few hints how to
  adapt the proof in these cases.
  \par
  The basic structure of the proof is follows: we choose a suitable weak
  base of~$\iV_1$ consisting of an irredundant relation~$R_1$, and
  identify a relation~$H_1\in\ccc{\set{R_1}}$ which allows us to encode
  a sufficiently complicated variant of the $\OptMinOnes$-problem into
  $\XSOL(\set{H_1})$. Thus by Theorem~\ref{thm:weakbases} and
  Lemma~\ref{lem:constructimpl} we have
  $H_1\in\ccc{\set{R_1}}\subseteq\ccc{\Gamma}$ and
  $[x\to y]\times\set{1} \in\ccc{\Gamma}$, wherefore
  Proposition~\ref{prop:coclonesminhd} implies
  $\XSOL(\Gamma')\aple \XSOL(\Gamma)$ where
  $\Gamma'=\set{H_1,[x\to y]\times\set{1}}$.
  According to~\cite[Theorem~2.14(4)]{KhannaSTW-01}, $\MinHD$ is
  equivalent to $\OptMinOnes(\Delta)$ for constraint languages~$\Delta$
  being dual Horn, not $0$-valid and not implicative hitting set
  bounded$+$ with any finite bound, that is, if
  $\cc{\Delta}\in\set{\iV_1,\iV_2}$. The key point of the construction
  is to choose~$R_1$ and~$H_1$ in such a way that we can find a
  relation~$G_1$ satisfying $\iV_1\subseteq\cc{\set{G_1}}\subseteq\iV_2$
  and $((G_1\times\set{1})\cup\set{\vec 0})\times\set{1}=H_1$. The
  latter property will allow us to prove an $\AP$-reduction
  $\MinHD\apeq\OptMinOnes(\set{G_1})\aple\XSOL(\Gamma')$,
  completing the chain.
  \par

  We first check that $R_1 = \GammaF{\fV_1}{\graphic[4]}$ satisfies
  $\cc{\set{R_1}} = \iV_1$: namely, by construction, this relation is
  preserved by the disjunction and by the constant operation with
  value~$1$, i.e., $\cc{R_1}\subseteq \iV_1$. This inclusion cannot be
  proper, since $\vec{0}\notin R_1$ ($\cc{R_1}\not\subseteq\iI_0$)
  and $x \lor (y \land z) \notin R_1$
  while $x = (e_1\circ\beta)\lor (e_4\circ\beta)$,
        $y = (e_1\circ\beta)\lor (e_2\circ\beta)$ and
        $z = (e_1\circ\beta)\lor (e_3\circ\beta)$
  belong to $\GammaF{\fV_1}{\graphic[4]}$ (cf.\ before
  Theorem~\ref{thm:weakbases-from-graphics} for the notation), i.e.\ the
  generating function $(x,y,z)\mapsto x\lor (y\land z)$ of the clone
  $\fS_{00}$~\cite[Figure~2, p.~8]{CreignouV-08} fails to be a
  polymorphism of $R_1$.
  For later we note that when~$\beta$ is chosen such that the
  coordinates of~$\graphic[4]$ are ordered lexicographically (and we are
  going to assume this from now on), then this failure can already be
  observed within the first seven coordinates of~$R_1$.
  Now according to Theorem~\ref{thm:weakbases-from-graphics}, the
  sedenary 
  relation $R_1:=\GammaF{\fV_1}{\graphic[4]}$ is a weak base
  relation for~$\iV_1$ without duplicate coordinates, and a brief moment
  of inspection shows that none of them is fictitious either. Therefore,
  $R_1$ is an irredundant weak base relation for $\iV_1$.
  We define $H_1$ to be
  $\Set{(x_0,\dotsc,x_8)}{(x_0,\dotsc,x_7,x_8,\dotsc,x_8)\in R_1}$,
  then clearly $H_1 \in\ccc{\set{R_1}}$.
  Now we put $G_1 := G_1'\smallsetminus\set{\vec{0}}$ where
  $G_1' := \Set{(x_0,\dotsc,x_6)}{(x_0,\dotsc,x_8)\in H_1}$,
  and one quickly verifies that
  $((G_1\times\set{1})\cup\set{\vec 0})\times\set{1}=H_1$.
  Since $G_1'\in \cc{H_1}\subseteq \cc{R_1} = \iV_1$ and removing the
  bottom-element $\vec{0}$ of a non-trivial join-semilattice with
  top-element still yields a join-semilattice with top-element, we have
  $G_1\in\iV_1$. With the analogous counterexample as for the
  relation~$R_1$ above, we can show that
  $(x,y,z)\mapsto x\lor(y\land z)$
  is not a polymorphism of~$G_1$ (because the non-membership is
  witnessed among the first seven coordinates).
  Thus, $\cc{\set{G_1}} = \iV_1$; in particular $G_1$, and any relation
  conjunctively definable from it, is not $0$-valid.
  \par

  For the reduction let now
  $\varphi(\vec x) = G_1(\vec{x_1}) \land\cdots\land G_1(\vec{x_k})$
  be an instance of $\OptMinOnes(\set{G_1})$.
  We construct a corresponding $\Gamma'$-formula~$\varphi'$ as follows.
  \begin{align*}
    \varphi''(\vec x,y,z) &:= H_1(\vec{x_1},y,z) \land \cdots \land
                              H_1(\vec{x_k},y,z) \\
    \varphi'''(\vec x,  \vec{x^{(2)}}, \cdots, \vec{x^{(\ell)}},z) &:=
    \bigwedge_{i=1}^\ell \left(
       (x_i \toequals{z} x_i^{(2)}) \land
       \bigwedge_{j=2}^{\ell-1} (x_i^{(j)} \toequals{z} x_i^{(j+1)}) \land
                                (x_i^\ell \toequals{z} x_i) \right)\\
    \varphi'(\vec x, \vec{x^{(2)}}, \cdots, \vec{x^{(\ell)}},y,z) &:=
     \varphi''(\vec x,y,z) \land \varphi'''(\vec x, \vec{x^{(2)}}, \cdots,\vec{x^{(\ell)}},z)
  \end{align*}
  where $\ell = \card{\vec x}$ is the number of variables of~$\varphi$,
  $y$ and~$z$ are new global variables, and where we have written
  $(u\toequals{w}v)$ to denote $([x\to y]\times\set{1})(u,v,w)$.
  Let $m_0$ be the assignment to the $\ell^2+2$ variables of~$\varphi'$
  given by $m_0(z) = 1$ and $m_0(x)=0$ elsewhere. It is clear that
  $(\varphi',m_0)$ is an instance of $\XSOL(\Gamma')$, since~$m_0$
  satisfies~$\varphi'$. The formula~$\varphi'''$ only multiplies each
  variable~$x$ from~$\varphi$ $\ell$-times and forces
  $x \eq x^{(2)} \eq \cdots \eq x^{(\ell)}$, which is just a
  technicality for establishing an $\AP$-reduction. The main idea of
  this proof is the correspondence between the solutions of~$\varphi$
  and~$\varphi''$.

  For each solution~$s$ of $\varphi(\vec x)$ there exists a
  solution~$s'$ of $\varphi''(\vec x, y)$ with $s'(y)=1$ (and
  $s'(z)=1$). Each solution~$s'$ of~$\varphi''$ has always $s'(z)=1$ and
  either $s'(y)=0$ or $s'(y)=1$. Because every variable
  from~$\vec{x}$ is part of one of the $\vec{x_i}$, the assignment~$m_0$
  restricted to $(\vec{x}, y,z)$ is the only solution~$s'$
  of~$\varphi''$ satisfying $s'(y)=0$. If otherwise $s'(y)$ equals~$1$,
  then~$s'$ restricted to the variables~$\vec x$ satisfies
  $\varphi(\vec x)$, following the correspondence between the
  relations~$G_1$ and~$H_1$.

  For~$r\in[1,\infty)$ let~$s'$ be an $r$-approximate solution of the
  $\XSOL(\Gamma')$-instance $(\varphi',m_0)$.
  Let $s := s'\Restriction_{\vec x}$ be
  the restriction of~$s'$ to the variables of~$\varphi$. Since
  $s'\neq m_0$, by what we showed before, $s'(y)=1$ and $s$~is a
  solution of $\varphi(\vec x)$. We have
  $\OPT(\varphi', m_0) \geq 2$ and $\OPT(\varphi) \geq 1$, since
  solutions of the $\XSOL(\Gamma')$-instance $(\varphi',m_0)$ must be
  different from~$m_0$, whereby~$y$ is forced to have value~$1$,
  and $[\varphi]\in\ccc{\set{G_1}}$ is not $0$-valid.
  Moreover,
  $\hw(s) = \hd(\vec 0, s)$, $\hd(s',m_0) = \ell \hw(s)+1$,
  $\OPT(\varphi',m_0) = \ell \OPT(\varphi) + 1$, and
  $\hd(s',m_0)\leq r\OPT(\varphi',m_0)$.
  From this and $\OPT(\varphi)\geq 1$ it follows that
  \begin{align*}
  \ell\hw(s) < \ell\hw(s)+ 1 = \hd(s',m_0)&\leq r\OPT(\varphi',m_0)
   = r\ell\OPT(\varphi) + r\\
   &\leq r\ell\OPT(\varphi) + r\OPT(\varphi)\\
   &\leq r\ell\OPT(\varphi) + r\OPT(\varphi) + (r-1)\ell\OPT(\varphi)\\
   &= (2r-1+r/\ell)\ell\OPT(\varphi)
    = (1+2(r-1) +r/\ell)\ell\OPT(\varphi).
  \end{align*}
  Hence~$s$ is an $(1+\alpha(r-1)+\Lo(1))$-approximate solution of
  the instance~$\varphi$ of $\OptMinOnes(\set{G_1})$ where $\alpha=2$.
  \par

  In the case when $\cc{\Gamma} = \iV_2$, the proof goes through with
  minor changes:
  $R_2 = \GammaF{\fV_2}{\graphic[4]} = R_1\smallsetminus\set{\vec{1}}$,
  so we define~$H_2$ and~$G_2$ like~$H_1$ and~$G_1$ just using~$R_2$
  and~$H_2$ in place of~$R_1$ and~$H_1$.
  Then we have $H_2 = H_1\smallsetminus\set{\vec{1}}$,
  $G_2 = G_1\smallsetminus\set{\vec{1}}$ and
  $\cc{\set{G_2}}= \iV_2$. Moreover, for the reduction we
  shall need an additional global variable~$w$ for~$\varphi'''$ (and
  $\varphi'$) since the encoding of the implication from
  Lemma~\ref{lem:constructimpl} requires it (and forces it to zero in
  every model).
  \par
  For $\cc{\Gamma}= \iV_0$ we can use
  $R_0=\GammaF{\fV_0}{\graphic[4]} = R_2\cup\set{\vec{0}}$; then,
  letting
  $H_0=\Set{(x_0,\dotsc,x_7)}{(x_0,\dotsc,x_7,x_7,\dotsc,x_7)\in R_0}
      \in\ccc{\set{R_0}}$,
  we have $H_0 = (G_2\times\set{1})\cup\set{\vec{0}}$.
  On a side note, we observe that $H_0 = \GammaF{\fV_0}{\graphic[3]}$,
  which we can use alternatively without detouring via~$R_0$. Given the
  relationship between $G_2$ and $H_0$, we do not need the global
  variable~$z$ in the definition of~$\varphi''$, but we need to have it
  in the definition of $\varphi'''$, where the relation given by
  Lemma~\ref{lem:constructimpl} necessitates atoms of the form
  $(u\xrightarrow{z=0}v)$ forcing $z$ to zero in every model.
  \par
  The case where $\cc{\Gamma} = \iV$ is similar to the previous:
  we can use the irredundant weak base relation
  $H = \GammaF{\fV}{\graphic[3]} = H_0 \cup\set{\vec{1}}
     = (G_1\times\set{1})\cup\set{\vec{0}}$. Except for~$y$ in the
  definition of~$\varphi''$ no additional global variables are needed in
  the definition of~$\varphi'$, because $[u\to v]$ atoms are directly
  available for~$\varphi'''$.
\end{proof}

\begin{corollary}\label{cor:XSOL-Horn-dual_Horn}
  If\/~$\Gamma$
  is exactly Horn $(\iE \subseteq \cc \Gamma \subseteq
  \iE_2)$ or exactly dual-Horn $(\iV \subseteq \cc \Gamma \subseteq
  \iV_2)$ then $\XSOL(\Gamma)$ is $\MinHD$-complete under
  $\AP$-Turing-reductions.
\end{corollary}
\begin{proof}
  Hardness follows from Lemma~\ref{lem:Horn_MinHD-hard} and duality.
  Moreover, $\XSOL(\Gamma)$ can be $\AP$-Turing-reduced to
  $\NSOL(\Gamma\cup \set{[x], [\neg x]})$ as follows: Given a
  $\Gamma$-formula~$\varphi$ and a model~$m$, we construct for
  every variable~$x$ of~$\varphi$ a formula $\varphi_x= \varphi \land
  (x\eq \cmpl{m}(x))$. Then for every~$x$ where
  $[\varphi_x]\neq\emptyset$ we run an oracle algorithm for
  $\NSOL(\Gamma\cup \set{[x], [\neg x]})$ on $(\varphi_x, m)$ and output
  one result of these oracle calls that is closest to~$m$.

  We claim that this algorithm provides indeed an $\AP$-Turing reduction.
  To see this observe first that the instance $(\varphi,m)$ has
  feasible solutions if and only if this holds for $(\varphi_x,m)$ and at
  least one variable~$x$.
  Moreover, we have
  $\OPT(\varphi,m) = \min_{x,[\varphi_x]\neq\emptyset}(\OPT(\varphi_x, m))$. Let $A(\varphi,m)$
  be the answer of the algorithm on $(\varphi, m)$ and let
  $B(\varphi_x, m)$ be the answers to the oracle calls. Consider a
  variable~$x^*$ such that $\OPT(\varphi,m) =
  \min_{x,[\varphi_x]\neq\emptyset}(\OPT(\varphi_x,m)) = \OPT(\varphi_{x^*},m)$, and assume that
  $B(\varphi_{x^*}, m)$ is an $r$-approximate solution of
  $(\varphi_{x^*},m)$. Then we get
  \begin{displaymath}
    \frac{\hd(m, A(\varphi, m))}{\OPT(\varphi, m)} =
    \frac{\min_{y,[\varphi_y]\neq\emptyset}(\hd(m, B(\varphi_y, m))}{\OPT(\varphi_{x^*}, m)} \leq
    \frac{\hd(m, B(\varphi_{x^*}, m))}{\OPT(\varphi_{x^*}, m)} \leq r .
  \end{displaymath}
  Thus the algorithm is indeed an $\AP$-Turing-reduction from
  $\XSOL(\Gamma)$ to $\NSOL(\Gamma\cup \set{[x], [\neg x]})$. Note that
  for $\Gamma\subseteq \iV_2$ the problem
  $\NSOL(\Gamma\cup \set{[x], [\neg x]})$ reduces to $\MinHD$ according
  to Propositions~\ref{prop:iV_2_to_MinOnes} and~\ref{prop:kstw-minhd}. Duality completes the
  proof.
\end{proof}

\section{Finding the Minimal Distance Between Solutions}
\label{sec:proofsMSD}

In this section we study the optimization problem
$\MinSolDistance$. We first consider the polynomial-time cases and
then the cases of higher complexity.

\subsection{Polynomial-Time Cases}
\label{ssec:proofsMSD-PO}

We show that for bijunctive constraints the problem $\MinSolDistance$
can be solved in polynomial time. After stating the result we present
an algorithm and analyze its complexity and correctness.

\begin{proposition}\label{prop:MSD-iD2}
  If\/ $\Gamma$ is a bijunctive constraint language
  $(\Gamma\subseteq\iD_2)$ then $\MSD(\Gamma)$ is in $\PO$.
\end{proposition}

By Proposition~\ref{prop:BakerPixley}, an algorithm for bijunctive
constraint languages~$\Gamma$ can be restricted to at most binary
clauses. Alternatively, one can use the plain base
$\set{[x],[\neg x],[x\lor y], [\neg x\lor y], [\neg x\lor\neg y]}$
of~$\iD_2$ exhibited in~\cite{CreignouKZ-08} to see that every relation
in $\Gamma$ can be written as a conjunction of disjunctions of two not
necessarily distinct literals. We shall treat these
disjunctions as one- or two-element sets of literals when extending the
algorithm of Aspvall, Plass, and Tarjan~\cite{AspvallPT-79} to compute
the minimum distance between distinct models of a bijunctive
constraint formula.

\algorithm{\textsc{bijunctive} $\MSD$}%
{An $\iD_2$-formula~$\varphi$ given as a collection of one- or
  two-element sets of literals (bijunctive clauses).}%
{``$\leq 1$ model'' or the minimal Hamming distance of any two
  distinct models of~$\varphi$.}%
{\mbox{}\\
  Let $\Var$ be the set of variables occurring in~$\varphi$.\\
  Let $\Lit := \set{v,\neg v \mid v\in \Var}$ be the set of literals.\\
  Let~$\bar{u}$ denote the literal complementary to~$u\in\Lit$.
  \medskip\par\noindent
  Construct the relation $R:=\Set{(\bar u,v),(\bar v,u)}{\set{u,v}\in\varphi\land u\neq v}\cup\Set{(\bar u,u)}{\set{u}\in\varphi}$.\\
    Let $\leq$ be the reflexive and transitive
    closure of~$R$, i.e.\ the least preorder on~$\Lit$ extending~$R$.\\
    Construct the sets
    \begin{alignat*}{2}
      \Var_0 &:= \SET{v\in\Var}{$v\leq x\leq\neg x\phantom{\mbox{}\leq v}$\quad or\quad $\phantom{\neg v\leq\mbox{}}\neg x\leq x\leq\neg v$\quad for some $x\in\Var$}\\
      \Var_1 &:= \SET{v\in\Var}{$\phantom{v\leq\mbox{}}x\leq\neg x\leq v$\quad or\quad $\neg v\leq\neg x\leq x\phantom{\mbox{}\leq\neg v}$\quad for some $x\in\Var$}
    \end{alignat*}
    If $\Var_0\cap\Var_1 \neq\emptyset$ or $\Var_0\cup\Var_1 = \Var$ holds, then return ``$\leq 1$ model''.
  \medskip\par\noindent
    Let $\Lit' := \Lit\smallsetminus\Set{v,\neg v}{v \in \Var_0\cup \Var_1}$.\\
    Let ${\sim}:= \set{(u,v)\in\Lit' \times \Lit' \mid u\leq v \land v\leq u}$.\\
    Return $\min\Set{\card L}{L \in \Lit'/{\sim}}$ as minimal Hamming distance.
  \medskip\par\noindent
  }%

\paragraph{Complexity:}
The size of $\Lit$ is linear in the number of variables, the reflexive
closure can be computed in time linear in~$\card{\Lit}$, the
transitive closure in time cubic in~$\card{\Lit}$,
see~\cite{WarshallTransClosure1962}. The equivalence relation~$\sim$
is the intersection of~$\leq$ restricted to $\Lit'$ and its inverse (quadratic in
$\card{\Lit'}$); from it we can obtain the partition~$\Lit'/{\sim}$ in linear
time in~$\card{\Lit'}\leq \card{\Lit}$, including the cardinalities of
the equivalence classes and their minimization.
Similarly, the remaining sets from the proof ($\Var_0$, $\Var_1$,
their intersection and union, and thus also $\Lit'$) can be computed
with polynomial time complexity.

\paragraph{Correctness:}
The pairs in~$R$ arise from interpreting the atomic constraints
in~$\varphi$ as implications. By transitivity of implication, the
inequality $u\leq v$ for literals $u,v$ means that every model~$m$
of~$\varphi$ satisfies the implication $u\to v$ or, equivalently,
$m(u)\leq m(v)$.  In particular, $x\leq\neg x$ implies $m(x)=0$ and
$\neg x\leq x$ implies $m(x)=1$. Therefore $\Var_0$ can be seen to be
the set of variables that have to be false in every model
of~$\varphi$, and $\Var_1$ the set of variables true in every model.

If $\Var_0 \cap \Var_1 \neq \emptyset$ holds then the
formula~$\varphi$ is inconsistent and has no solution.  If
$\Var_0 \cup \Var_1 = \Var$ holds, then every variable has a unique
fixed value, hence $\varphi$ has only one solution.
Otherwise the formula is consistent and not all variables are
fixed, hence there are at least two models.

To determine the minimal number of variables, whose values can be
flipped between any two models of~$\varphi$, it suffices to consider
the literals without fixed value, $\Lit'$. If we have $u\leq v$ and
$v\leq u$, the literals are equivalent, $u\sim v$, and must have the
same value in every model. This means that any two distinct models have to
differ on all literals of at least one equivalence class in
$\Lit'/{\sim}$. Therefore, the return value of the algorithm is a
lower bound for the minimal distance.

To prove that the return value can indeed be attained, we exhibit two
models $m_0 \neq m_1$ of~$\varphi$ having the least cardinality of any
equivalence class in~$\Lit'/{\sim}$ as their Hamming distance.
Let $L \in \Lit'/{\sim}$ be a class of minimum cardinality. Define
$m_0(u):= 0$ and $m_1(u):= 1$ for all literals $u \in L$. We
extend this by setting $m_0(w) := m_1(w):= 0$ for all $w\in\Lit$ such that
$w\leq u$ for some $u\in L$, and by $m_0(w):= m_1(w):= 1$ for all
$w\in\Lit$ such that $u\leq w$ for some $u\in L$.  For variables
$v\in\Var$ satisfying $v\leq \neg v$ or $\neg v\leq v$ we have
$v\in \Var_0\cup\Var_1$, and thus $v\notin \Lit'$; in other words, for
$[v]_{\sim} \in \Lit'/{\sim}$ the classes $[v]_{\sim}$ and
$[\neg v]_{\sim}$ are incomparable. Thus, so far, we have not defined
$m_0$ and $m_1$ on a variable $v\in \Var$ and on its negation $\neg v$
at the same time. Of course, fixing a value for a negative literal
$\neg v$ implicitly means that we bind the assignment
for $v\in \Var$ to the opposite value.

It remains to fix the value of literals in $\Lit'$ that are neither related to
the literals in~$L$ nor have fixed values in all models.  Suppose
$(\bar u, v) \in R$ is a constraint such that the value of at least
one literal has not yet been defined. There are three cases: either
both literals have not yet received a value, or $\bar u$ is undefined
and $v$ has been assigned the value~$1$ (either as a fixed value in
all models or because of being greater than a literal in~$L$ or because of
being lesser than a complement of a literal in~$L$), or $v$
is undefined and $\bar u$ has been assigned the value~$0$ (either as a
fixed value in all models or because of being smaller than a literal
in~$L$ or greater than a complement of a literal in~$L$). All three cases can be handled by defining both models, $m_0$
and $m_1$, on the remaining variables identically: starting with a
minimal literal $u$, where $m_0$ and $m_1$ are not yet defined, we assign
$m_0(u) := m_1(\cmpl{u}) := 0$ and $m_1(u):= m_0(\cmpl{u}):=1$.
This way none of the constraints is violated, and $m_0$ and $m_1$ are
distinct only on variables corresponding to literals in $L$. Iterate this
procedure until all variables (and their complements) have been assigned
values. If $\Lit''\subseteq\Lit'$ denotes the literals remaining after
propagating the values of~$m_0$ and~$m_1$ on~$L$, then the presented
method can be implemented by partitioning $\Lit''$ into two classes $L_0$
and $L_1$ such that $L_0\cap\set{u,\cmpl{u}}$ is a singleton for every
$u\in \Lit''$ and each weakly connected component of the quasiordered set
$(\Lit'',\leq)$ is either a subset of $L_0$ or $L_1$. Then set $m_0$ and
$m_1$ to $k$ on the literals belonging to $L_k$ for $k\in\set{0,1}$.

By construction, $m_0$ differs from~$m_1$ only in the variables
corresponding to the literals in $L$, so their Hamming distance is
$\card{L}$ as desired. Moreover, both assignments respect the order
constraints in $(\Lit,\leq)$. As these faithfully reflect all original
atomic constraints, $m_0$ and $m_1$ are indeed models of~$\varphi$.

\begin{proposition}\label{prop:MSD-iE2-iV2}
  If\/~$\Gamma$ is a Horn $(\Gamma\subseteq\iE_2)$ or a dual Horn
  $(\Gamma\subseteq\iV_2)$ constraint language then $\MSD(\Gamma)$ is
  in $\PO$.
\end{proposition}
We only discuss the Horn case ($\Gamma\subseteq\iE_2$), dual Horn
($\Gamma\subseteq\iV_2$) being symmetric.

\algorithm{\textsc{Horn} $\MSD$}%
{A Horn formula~$\varphi$ given as a set of Horn clauses (cf.\ the plain
base of $\iE_2$ given in~\cite{CreignouKZ-08}).}%
{``$\leq 1$ model'' or the minimal Hamming distance of any two
  distinct models of~$\varphi$.}%
{\medskip\par\noindent
  For each variable~$x$ in~$\varphi$, add the clause $(\neg x \lor x)$.\\
  Let $\U := \emptyset$.\\
  Apply the following rules to~$\varphi$ until no more clauses and
  literals can be removed and no new clauses can be added.

  \emph{Unit resolution and unit subsumption:} Let $\bar{u}$ denote
  the complement of a literal~$u$. If the clause set contains a unit
  clause $u$, remove all clauses containing the literal $u$ and remove
  all literals $\bar u$ from the remaining clauses.  Add $u$ to the
  set $\U$.

  \emph{Hyper-resolution with binary implications:} Resolve all
  negative literals of a clause simultaneously with binary
  implications possessing identical premises.
  \begin{displaymath}
    \infer{(\neg x \lor z)}%
    {(\neg x\lor y_1) \cdots (\neg x\lor y_k)&(\neg y_1\lor \dots
      \lor \neg y_k\lor z)}
    \qquad\quad
    \infer{(\neg x)}%
    {(\neg x\lor y_1) \cdots (\neg x\lor y_k)&(\neg
      y_1\lor\dots\lor\neg y_k)}
  \end{displaymath}
  Let $\D$ be the set of clauses after applying the two rules exhaustively.\\
  If $\D$ contains the empty clause, return ``$\leq 1$ model''.\\
  If $\U$ contains a literal for every variable in~$\varphi$, return ``$\leq 1$ model''.\\
  If $\varphi$ contains a variable that appears neither in~$\D$ nor
  in~$\U$, return $1$ as minimal Hamming distance.
  \medskip\par\noindent
  Otherwise, let $\Var$ be the set of variables occurring in~$\D$, and let
  ${\sim}\subseteq \Var^2$ be the relation defined by $x\sim y$ if
  $\set{\neg x\lor y,\neg y\lor x}\subseteq\D$.  Note that $\sim$ is
  an equivalence, since the tautological clauses ensure reflexivity
  and resolution of implications computes their transitive closure.
  We say that a variable~$z$ depends on variables $y_1,\dots,y_k$, if
  $\D$ contains the clauses $ \neg y_1\lor \dots\lor \neg y_k\lor z$,
  $\neg z\lor y_1$, \ldots, $\neg z\lor y_k$ and $z\not\sim y_i$ holds
  for all $i=1,\dots,k$.\\
  Return $\min\SET{\card X}{$X\in\Veq$, $X$ does not contain dependent
  variables}$ as minimal Hamming distance.
  \medskip\par\noindent
}%

\paragraph{Complexity:}
The run-time of the algorithm is polynomial in the number of clauses
in~$\varphi$: Unit resolution/subsumption
can be applied at most once for each variable, and hyper-resolution
has to be applied at most once for each variable~$x$ and each clause
$\neg y_1\lor \dots \lor \neg y_k \lor z$ and
$\neg y_1\lor \dots\lor \neg y_k$.

\paragraph{Correctness:}
Adding resolvents and removing subsumed clauses maintains logical
equivalence, therefore $\D\cup\U$ is logically equivalent
to~$\varphi$, i.e., both clause sets have the same models. We note
that the sets of variables of~$\U$ and of~$\D$ are disjoint.
The unit clauses in~$\U$ are always (uniquely) satisfiable, thus~$\D$
and~$\varphi$ are equisatisfiable.
Therefore, if~$\D$ contains the empty clause, $\varphi$ is also
unsatisfiable; otherwise~$\D$ is satisfiable, e.g., by assigning~$0$ to
every $x\in\Var$. In this case, if~$\U$ contains a literal for every
variable of~$\varphi$, the unit clauses in~$\U$ define a unique model
of~$\varphi$.

Otherwise $\varphi$ has at least two models $m_1\neq m_2$. In the simplest
case some variable~$x$ in~$\varphi$ has been left unconstrained by
$\D$ and $\U$; in this case we can pick any model of $\D$ and $\U$ and
extend it to two different models of~$\varphi$ with Hamming
distance~$1$ by setting $m_1(x)=0$ and $m_2(x)=1$ and setting $m_1(y) =
m_2(y) = 0$ for any other variable~$y$ outside $\D$ and $\U$.  For the remaining
situations it is sufficient to consider the models of~$\D$ only, as
each model~$m$ of~$\D$ uniquely extends to a model of~$\varphi$ by
defining $m(x)=1$ for $(x)\in\U$ and $m(x)=0$ for $(\neg x)\in\U$;
hence the minimal Hamming distances of the models of $\varphi$
and~$\D$ will be the same.

We are thus looking for models $m_1,m_2$ of~$\D$ such that the size of
the difference set $\Delta(m_1, m_2)=\Set{x}{m_1(x)\neq m_2(x)}$ is
minimal. In fact, since the models of Horn formulas are closed under
minimum, we may assume $m_1<m_2$, i.e., we have $m_1(x)=0$ and
$m_2(x)=1$ for all variables $x\in \Delta(m_1, m_2)$.  Indeed, given
two models $m_2$ and $m_2'$ of~$\D$ where neither $m_2\leq m_2'$ nor
$m_2'\leq m_2$, $m_1= m_2 \land m_2'$ is also a model, and it is distinct
from $m_2$.  Since $\hd(m_1,m_2) \leq \hd(m_2,m_2')$, the minimal
Hamming distance will occur between models~$m_1$ and~$m_2$ satisfying
$m_1 < m_2$.

\medskip

\noindent%
Note the following facts regarding the equivalence relation~$\sim$ and
dependent variables.
\begin{asparaitem}
\item If $x\sim y$ then the two variables must have the same value in
  every model of~$\D$ in order to satisfy the implications
  $\neg x\lor y$ and $\neg y\lor x$. This means that for all
  models~$m$ of~$\D$ and all $X\in\Veq$, we have either $m(x)=0$ for
  all $x\in X$ or $m(x)=1$ for all $x\in X$.
\item The dependence of variables is acyclic: If, for some $l\geq 2$, for
  every $1\leq i<l$ we have that $z_i$ depends on variables including
  one, say $y_i$, which is equivalent to~$z_{i+1}$, and $z_l = z_1$, then
  there is a cycle of binary implications between the variables and thus
  $z_i\sim y_i \sim z_j$ for all $i,j$, contradicting the definition of
  dependence.
\item If a variable~$z$ depending on $y_1$, \dots, $y_k$ belongs to a
  difference set~$\Delta(m_1, m_2)$, then at least one of the $y_i$s
  also has to belong to~$\Delta(m_1, m_2)$: $m_2(z)=1$ implies
  $m_2(y_j)=1$ for all $j=1,\dots,k$ (because of the clauses
  $\neg z\lor y_i$), and $m_1(z)=0$ implies $m_1(y_i)=0$ for at least
  one~$i$ (because of the clause
  $\neg y_1\lor \dots\lor \neg y_k\lor z$). Therefore
  $\Delta(m_1, m_2)$ is the union of at least two sets in~$\Veq$,
  namely the equivalence class of~$z$ and the one of~$y_i$.
\item If some $z_1\in \Delta(m_1,m_2)$ is equivalent to a variable
      $z_1'$ that depends on some other variables, then we have a
      variable $z_2$ among them, which also belongs to
      $\Delta(m_1,m_2)$. If the equivalence class of $z_2$ still contains
      a variable $z_2'$ depending on other variables, we can iterate this
      procedure. In this way we obtain a sequence $z_1\sim z_1',
      z_2\sim z_2', z_3\sim z_3', \dotsc$ where $z_i'$ depends on
      variables including $z_{i+1}$, which is equivalent to $z_{i+1}'$.
      Because there are only finitely many variables and because of
      acyclicity, after a linear number of steps we must reach a variable
      $z_n\in\Delta(m_1,m_2)$ such that its equivalence class (being a
      subset of the difference set) does not contain any dependent
      variables.
\end{asparaitem}

Hence the difference between any two models cannot be smaller than the
cardinality of the smallest set in~$\Veq$ without dependent
variables. It remains to show that we can indeed find two such models.

Let~$X$ be a set in~$\Veq$ which has minimal cardinality among the
sets without dependent variables, and let $m_0,m_1$ be interpretations
defined as follows:
\begin{inparaenum}[(1)]
\item $m_0(y) = 0$ and $m_1(y) = 1$ if $y\in X$;
\item $m_0(y) = 1$ and $m_1(y) = 1$ if $y\notin X$ and
  $(\neg x\lor y)\in\D$ for some $x\in X$;
\item $m_0(y) = 0$ and $m_1(y) = 0$ otherwise.
\end{inparaenum}
We have to show that~$m_0$ and~$m_1$ satisfy all clauses in~$\D$. Let
$m$ be any of these models. $\D$ contains two types of clauses.
\begin{asparaenum}[Type~1:]
\item Horn clauses with a positive literal
  $\neg y_1\lor \dots\lor \neg y_k\lor z$. If $m(y_i)=0$ for any~$i$,
  we are done. So suppose $m(y_i)=1$ for all $i=1,\dots,k$; we have to
  show $m(z)=1$. The condition $m(y_i)=1$ means that either $y_i\in X$
  (for $m=m_1$) or that there is a clause $(\neg x_i \lor y_i)\in\D$
  for some $x_i\in X$.  We distinguish the two cases $z\in X$ and
  $z\notin X$.

  Let $z\in X$. If $z\sim y_i$ for any $i$, we are done for we have
  $m(z)=m(y_i)=1$. So suppose $z\not\sim y_i$ for all~$i$. As the
  elements in~$X$, in particular $z$ and the $x_i$s, are equivalent
  and the binary clauses are closed under resolution, $\D$ contains
  the clause $\neg z\lor y_i$ for all~$i$. But this would mean that
  $z$ is a variable depending on the $y_i$s, contradicting the
  assumption $z\in X$.

  Let $z\notin X$, and let $x\in X$.  As the elements in~$X$ are
  equivalent and the binary clauses are closed under resolution, $\D$
  contains $\neg x\lor y_i$ for all~$i$. Closure under
  hyper-resolution with the clause
  $\neg y_1 \lor \dots\lor \neg y_k\lor z$ means that $\D$ also
  contains $\neg x\lor z$, whence $m(z)=1$.

\item Horn clauses with only negative literals
  $\neg y_1\lor \dots\lor \neg y_k$. If $m(y_i)=0$ for any~$i$, we are
  done. It remains to show that the assumption $m(y_i)=1$ for all
  $i=1,\dots,k$ leads to a contradiction.  The condition $m(y_i)=1$
  means that either $y_i\in X$ (for $m=m_1$) or that there is a clause
  $(\neg x_i\lor y_i)\in\D$ for some $x_i\in X$.
  Let $x$ be some particular element of~$X$.  Since the elements
  in~$X$ are equivalent and the binary clauses are closed under
  resolution, $\D$ contains the clause $\neg x\lor y_i$ for
  all~$i$. But then a hyper-resolution step with the clause
  $\neg y_1 \lor \dots\lor \neg y_k$ would yield the unit clause
  $\neg x$, which by construction does not occur in~$\D$. Therefore at
  least one $y_i$ is neither in~$X$ nor part of a clause
  $\neg x\lor y_i$ with $x\in X$, i.e., $m(y_i)=0$.
\end{asparaenum}

\subsection{Hard Cases}
\subsubsection{Two Solution Satisfiability}

In this section we study the feasibility problem of $\MSD(\Gamma)$
which is, given a $\Gamma$-formula~$\varphi$, to decide if~$\varphi$
has two distinct solutions.

\decproblem{$\TSSAT(\Gamma)$}%
{A conjunctive formula~$\varphi$ over the relations from the constraint
  language~$\Gamma$.}%
{Are there two satisfying assignments~$m \neq m'$ of~$\varphi$?}

A priori it is not clear that the tractability of $\TSSAT$ is fully
characterized by co-clones. The problem is that the implementation of
relations of some language~$\Gamma$ by another language~$\Gamma'$
might not be parsimonious, that is, in the implementation one solution
to a constraint might be blown up into several ones in the
implementation. Fortunately we can still determine the tractability
frontier for $\TSSAT$ by combining the corresponding results for
$\SAT$ and $\AnotherSat$.

\begin{lemma}\label{lem:sattotwo}
  Let $\Gamma$ be a constraint language for which $\SAT(\Gamma)$
  is $\NP$-hard. Then $\TSSAT(\Gamma)$ is $\NP$-hard.
\end{lemma}
\begin{proof}
  Since $\SAT(\Gamma)$ is $\NP$-hard, there must be a relation~$R$
  in~$\Gamma$ having more than one tuple, because every relation
  containing only one tuple is at the same time Horn, dual Horn,
  bijunctive, and affine. Given an instance $\varphi$ for
  $\SAT(\Gamma)$, construct $\varphi'$ as
  $\varphi \land R(y_1, \ldots, y_\ell)$ where $\ell$ is the arity of
  $R$ and $y_1, \ldots, y_\ell$ are new variables not appearing in
  $\varphi$. Obviously, $\varphi$ has a solution if and only if
  $\varphi'$ has at least two solutions.  Hence, we have proved
  $\SAT(\Gamma)\mle \TSSAT(\Gamma)$.
\end{proof}

\begin{lemma}\label{lem:asattotwo}
  Let $\Gamma$ be a constraint language for which
  $\AnotherSat(\Gamma)$ is $\NP$-hard. Then the problem
  $\TSSAT(\Gamma)$ is $\NP$-hard.
\end{lemma}
\begin{proof}
  Let a $\Gamma$-formula~$\varphi$ and a satisfying assignment~$m$ be
  an instance of $\AnotherSat(\Gamma)$. Then~$\varphi$ has a solution
  other than~$m$ if and only if it has two distinct solutions.
\end{proof}

\begin{lemma}\label{lem:twoeasy}
  Let $\Gamma$ be a constraint language for which both $\SAT(\Gamma)$
  and $\AnotherSat(\Gamma)$ are in~$\P$. Then $\TSSAT$ is also
  in~$\P$.
\end{lemma}
\begin{proof}
  Let $\varphi$ be an instance of $\TSSAT(\Gamma)$. All
  polynomial-time decidable cases of $\SAT(\Gamma)$ are constructive,
  i.e., whenever that problem is polynomial-time decidable, there
  exists a polynomial-time algorithm computing a satisfying
  assignment provided it exists. If $\varphi$ is not satisfiable, we
  reject the instance. Otherwise, we can compute in polynomial time a satisfying
  assignment~$m$ of~$\varphi$. Now use the algorithm for
  $\AnotherSat(\Gamma)$ on the instance $(\varphi, m)$ to decide if
  there is a second solution to~$\varphi$.
\end{proof}

\begin{corollary}\label{cor:sat-tssat-complexity}
  For any constraint language~$\Gamma$, the problem $\TSSAT(\Gamma)$
  is in~$\P$ if both $\SAT(\Gamma)$ and $\AnotherSat(\Gamma)$ are
  in~$\P$. Otherwise, $\TSSAT(\Gamma)$ is $\NP$-hard.
\end{corollary}

\begin{proposition}\label{prop:linearapprox}
  Let~$\Gamma$ be a constraint language for which $\TSSAT(\Gamma)$ is
  in~$\P$. Then there is a polynomial-time $n$-approximation algorithm
  for $\MSD(\Gamma)$, where~$n$ is the number of variables of the
  $\Gamma$-formula on input.
\end{proposition}
\begin{proof}
  Since~$\TSSAT(\Gamma)$ is in~$\P$, both $\SAT(\Gamma)$ and
  $\AnotherSat(\Gamma)$ must be in~$\P$ by
  Corollary~\ref{cor:sat-tssat-complexity}. Since $\SAT(\Gamma)$ is
  in~$\P$, we can compute a model~$m$ of the input~$\varphi$ in
  polynomial time if it exists. Now we check the
  $\AnotherSat(\Gamma)$-instance $(\varphi,m)$. If it has a
  solution~$m'\neq m$, it is also polynomial time computable, and we
  return $(m,m')$. If we fail somewhere in this process, then the
  $\MSD(\Gamma)$-instance~$\varphi$ does not have feasible solutions; otherwise,
  $\hd(m,m')\leq n \leq n\cdot\OPT(\varphi)$.
\end{proof}

\subsubsection{MinDistance-Equivalent Cases}

In this section we show that, as for the $\NextSol$ problem, the
affine cases of $\MSD$ are $\MinDist$-complete.

\begin{proposition}\label{prop:msdaffine}
  $\MSD(\Gamma)$ is $\MinDist$-complete if the constraint
  language~$\Gamma$ satisfies the inclusions
  $\iL\subseteq \cc{\Gamma} \subseteq \iL_2$.
\end{proposition}
\begin{proof}
  We prove $\MSD(\Gamma)\apeq \NextSol(\Gamma)$, which is
  $\MinDist$-com\-plete for each constraint language~$\Gamma$
  satisfying the inclusions $\iL\subseteq \cc{\Gamma} \subseteq \iL_2$,
  according to Proposition~\ref{prop:MinDist-hardness-XSOL}.  As the
  inclusion~$\Gamma\subseteq \iL_2 = \cc{\set{\even^4,[x],[\neg x]}}$
  holds, any $\Gamma$-formula~$\psi$ is expressible as
  $\exists y(A_1 x +A_2 y \eq c)$.  The projection of the affine
  solution space is again an affine space, so it can be understood as
  solutions of a system~$Ax = b$.  If $(\psi,m_0)$ is an instance of
  $\XSOL(\Gamma)$, then $\psi$ is a $\MSD(\Gamma)$-instance, and a
  feasible solution $m_1\neq m_2$ satisfying~$\psi$ gives a feasible
  solution $m_3:=m_0+(m_2 - m_1)$ for $(\psi,m_0)$, where
  $\hd(m_0,m_3) 
  = \hd(m_2,m_1)$.  Conversely, a solution $m_3\neq m_0$ to
  $(\psi,m_0)$ yields a feasible answer to the $\MSD$-instance~$\psi$.
  Thus, $\OPT(\psi) = \OPT(\psi,m_0)$ and so
  $\XSOL(\Gamma)\aple\MSD(\Gamma)$. The other way round, if~$\psi$ is
  an $\MSD$-instance, then attempt to solve the system $Ax=b$ defined by
  it; if there is no or a unique solution, then the instance does not
  have feasible solutions. Otherwise, we have at least two distinct
  models of~$\psi$; let $m_0$ be one of these. As above we conclude
  $\OPT(\psi) = \OPT(\psi,m_0)$, and therefore,
  $\MSD(\Gamma)\aple\XSOL(\Gamma)$.
\end{proof}

\subsubsection{Tightness Results}

We prove that Proposition~\ref{prop:linearapprox} is essentially tight
for some constraint languages. This result builds heavily on the
previous results from Section~\ref{sssec:xsoltight}.

\begin{proposition}\label{prop:tightnessOfLinearapprox}
  For a constraint language~$\Gamma$ satisfying the inclusions
  $\iN \subseteq \cc \Gamma \subseteq \iI$ and any $\varepsilon>0$
  there is no polynomial-time $n^{1-\varepsilon}$-approximation
  algorithm for $\MSD(\Gamma)$, unless $\P = \NP$.
\end{proposition}
\begin{proof}
  We show that any polynomial time $n^{1-\varepsilon}$-approximation
  algorithm for $\MSD(\Gamma)$ would also allow to decide in
  polynomial time the problem $\AnotherSatNC(\Gamma)$, which is
  $\NP$-complete by Proposition~\ref{prop:AScomplhard}.

  The algorithm works as follows. Given an instance $(\varphi, m)$ for
  $\AnotherSatNC(\Gamma)$, the algorithm accepts if~$m$ is not a
  constant assignment. Since~$\Gamma$ is $0$-valid (and $1$-valid),
  this output is correct. If~$\varphi$ has only one variable, reject
  because~$\varphi$ has only two models; otherwise, proceed as
  follows.

  For each variable~$x$ of~$\varphi$, we construct a new
  formula~$\varphi'_{x}$ as follows.  Let~$k$ be the smallest integer
  greater than $1/\varepsilon$. Introduce $n^k-n$ new variables~$x^i$
  for $i = 1, \ldots, n^k-n$. For every $i \in \set{1, \ldots, n^k-n}$
  and every constraint $R(y_1, \ldots , y_\ell)$ in~$\varphi$, such
  that $x \in \set{y_1, \ldots, y_\ell}$, construct a new constraint
  $R(z_1^i, \ldots, z_\ell^i)$ by $z_j^i = x^i$ if $y_j = x$ and
  $z_j^i = y_j$ otherwise; add all the newly constructed constraints
  to~$\varphi$ in order to get~$\varphi'_{x}$. Note, that we can
  extend models~$s$ of~$\varphi$ to models~$s'$ of~$\varphi'_{x}$ by
  setting $s'(x^i)=s(x)$. In particular, this can be done for~$m$,
  yielding $m'\in[\varphi'_{x}]$.
  As $\Gamma\subseteq \iI = \iI_0\cap \iI_1$, the $\MSD(\Gamma)$-instance
  $\varphi'_{x}$ has feasible solutions; thus run the
  $n^{1-\varepsilon}$-approximation algorithm for $\MSD(\Gamma)$ on
  $\varphi'_{x}$. If for every $x$ the answer is a pair $(m_1,m_2)$
  with $m_2 = \cmpl{m_1}$, then reject, otherwise accept.

  This procedure is a correct polynomial-time algorithm for
  $\AnotherSatNC(\Gamma)$. For polynomial runtime is clear, it
  remains to show correctness. If~$\varphi$ has only constant models,
  then the same is true for every~$\varphi'_{x}$ since~$\varphi$
  contains a variable distinct from~$x$. Thus each
  approximation must result in a pair of complementary constant
  assignments, and the output is correct. Assume now that there is a
  model~$s$ of~$\varphi$ different from~$\vec{0}$
  and~$\vec{1}$. Hence, there exists a variable $x$ such that
  $s(x)=m(x)$ because $m$ is constant. It follows that~$\varphi'_{x}$ has a
  model~$s'$ fulfilling $\OPT(\varphi'_{x})\leq \hd(s', m')<n$, where~$n$ is the number of variables
  of~$\varphi$. But then the approximation algorithm must find two
  distinct models~$m_1\neq m_2$ of~$\varphi'_{x}$ satisfying
  $\hd(m_1, m_2) < n \cdot (n^k)^{1-\varepsilon} =
  n^{k(1-\varepsilon)+1}$.
  Since we stipulated $k > 1/\varepsilon$, it follows that
  $\hd(m_1,m_2) < n^k$. Consequently, we have $m_2\neq \cmpl{m_1}$ and
  the output of our algorithm is again correct.
\end{proof}

\section{Concluding Remarks}
\label{sec:conclusion}

The problems investigated in this paper are quite natural. In the
space of bit-vectors we search for a solution of a formula that is
closest to a given point, or for a solution next to a given
solution, or for two solutions witnessing the smallest Hamming
distance between any two solutions. Our results describe the
complexity of exploring the solution space for arbitrary families of
Boolean relations.  Moreover, our problems generalize problems
familiar from the literature: $\OptMinOnes$, $\NCW$, and $\DistanceSAT$
are instances of our $\NearestSol$, while $\MinDist$ is the same as
our problem $\MinSolDistance$ when restricting the latter to affine relations.

To prove the results, we first had to extend the notion of
$\AP$-reduction.  The optimization problems considered in the literature
have the property that each instance has at least one feasible
solution. This is not the case when looking for nearest solutions
regarding a given solution or a prescribed Boolean tuple,
as a formula may have just a single solution or no
solution at all.  Therefore we had to refine the notion of
$\AP$-reductions such that it correctly handles instances without
feasible solutions.

The complexity of $\NearestSol$ can be classified by the usual
approach: We first show that for each constraint language the
complexity of the problem does not change when admitting existential
quantifiers and equality, and then check all finitely related clones
according to Post's lattice. This approach does not work for the problems
$\NextSol$ and $\MinSolDistance$: It does not seem to be possible to
show a priori that the complexity remains unaffected under such
language extensions. In principle the complexity of a problem might
well differ for two constraint languages $\Gamma_1$ and $\Gamma_2$
that represent the same clone ($\cc{\Gamma_1} = \cc{\Gamma_2}$) but
that differ with respect to partial polymorphisms
($\ccc{\Gamma_1\cup\set{\eq}} \neq \ccc{\Gamma_2\cup\set{\eq}}$).
Theorems~\ref{thm:NOSol} and~\ref{thm:MSD} finally show that this is
not the case, but we learn this only a posteriori.
Our method of proof fundamentally relies on irredundant weak bases that
seem to be the perfect fit for such a situation: a priori compatibility
with existential quantification is not required, but it will follow once
the proof succeeds just using weak bases.

\begin{figure}
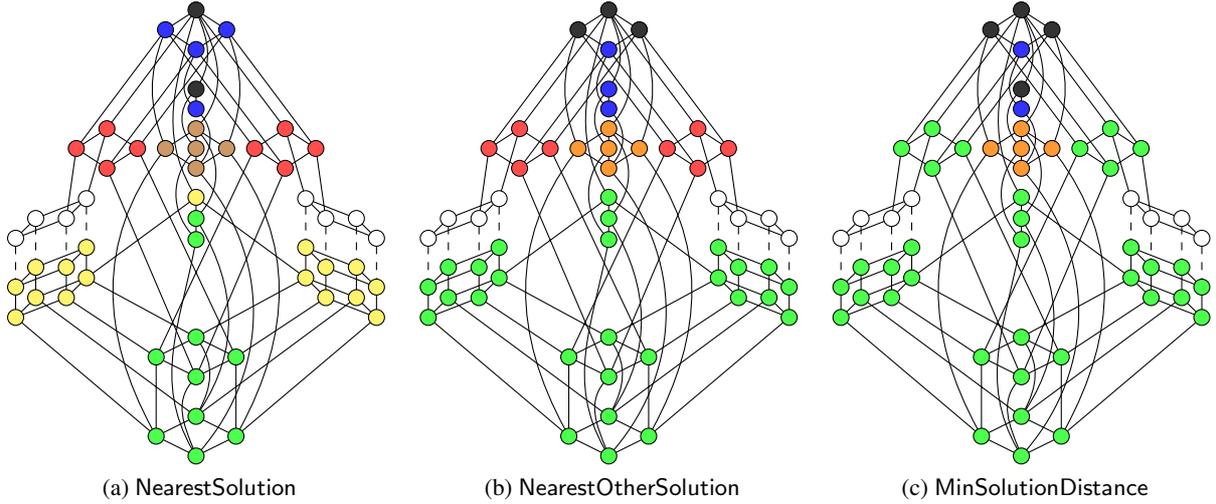

  \subfloat[$\NearestSol$]%
    {\expandafter\postlattice\expandafter[\NSOLstyle,
     scale=0.35,labels=none,cshape/circle/.append style={minimum width=6pt}]
    }%
  \hfill
  \subfloat[$\NextSol$]%
    {\expandafter\postlattice\expandafter[\NOSOLstyle,
     scale=0.35,labels=none,cshape/circle/.append style={minimum width=6pt}]
    }%
  \hfill
  \subfloat[$\MinSolDistance$]%
    {\expandafter\postlattice\expandafter[\MSDstyle,
     scale=0.35,labels=none,cshape/circle/.append style={minimum width=6pt}]
    }
    \caption{Comparing the complexities: The hard cases (colored blue
      and black) are the same, whereas the polynomial cases (green)
      increase from left to right.}
    \label{fig:comparison}
\end{figure}

Figure~\ref{fig:comparison} compares the complexity classifications of
the three problems. Regarding $\NearestSol$ and $\NextSol$, knowing
that an assignment is a solution apparently helps in finding a
solution nearby.  For expressive constraint languages it is
$\NP$-complete to decide whether a feasible solution exists at all;
for $\NearestSol$ this requires the existence of at least one
satisfying assignment, while the other two problems need even
two. Kann proved in~\cite{Kann-94} that $\OptMinOnes(\Gamma)$ is
$\NPOPB$-complete for $\cc \Gamma = \BR$, where $\NPOPB$ is the class
of $\NPO$ problems with a polynomially bounded objective
function. This result implies that $\NearestSol(\Gamma)$ is
$\NPOPB$-complete for $\cc \Gamma = \BR$ as well. It is unclear
whether this result also holds for $\cc \Gamma = \iN_2$. It may be
possible to find a suitable constraint language~$\Gamma'$ satisfying
$\cc{\Gamma'}=\BR$ such that $\OptMinOnes(\Gamma')$ reduces to $\NextSol(\Gamma)$ for
$\iI_0 \subseteq \cc \Gamma$ or $\iI_1 \subseteq \cc \Gamma$, proving
thus that $\XSOL(\Gamma)$ is $\NPOPB$-complete for these
cases. Likewise, the $\NPOPB$-hardness of $\MSD(\Gamma)$ for
$\iN_2\subseteq\cc \Gamma$ or $\iI_0\subseteq\cc \Gamma$ or
$\iI_1\subseteq\cc \Gamma$ remains open for the time being.

\bibliographystyle{plain}
\bibliography{../central}

\end{document}